\documentclass[11pt,english]{article}
\usepackage[T1]{fontenc}

\usepackage{graphicx, subfigure}
\expandafter\def\csname ver@subfig.sty\endcsname{}
\usepackage[latin9]{inputenc}
\usepackage{geometry}
\usepackage{amsmath}
\usepackage{appendix}
\usepackage{verbatim}
\usepackage{amssymb}
\usepackage{esint}
\usepackage{makecell,multirow,diagbox}
\usepackage{mathtools}
\usepackage{epstopdf}
\usepackage[latin9]{inputenc}
\usepackage{pict2e}
\usepackage{xcolor}
\usepackage{colortbl,booktabs}
\usepackage{stmaryrd}
\usepackage{esint}
\usepackage[all,cmtip]{xy}
\usepackage{mathtools}
\usepackage{float}
\usepackage{setspace}
\usepackage{authblk}
\usepackage{amsthm}
\usepackage{mathrsfs}
\usepackage{stmaryrd}
\usepackage{bbold}
\usepackage{color}
\usepackage{tikz}
\usepackage[all,cmtip]{xy}
\usepackage{algorithm}  
\usepackage{algorithmic}

\theoremstyle{plain}\newtheorem{theorem}{Theorem}[section]
\theoremstyle{plain}
\newtheorem{remark}{Remark}
\usepackage{color}
\definecolor{marin}{rgb} {0., 0.3, 0.7}
\definecolor{rouge}{rgb} {0.8, 0., 0.}
\definecolor{sepia}{rgb} {0.8, 0.5, 0.}
\usepackage[colorlinks,citecolor=sepia,linkcolor=marin,bookmarksopen,bookmarksnumbered]{hyperref}

\def\xb{\mathbf{x}}

\def\be{\begin{equation}}
\def\ee{\end{equation}}

\theoremstyle{definition}

\DeclareSymbolFont{largesymbol}{OMX}{yhex}{m}{n}
\DeclareMathAccent{\Widehat}{\mathord}{largesymbol}{"62}

\newtheorem{proposition}[theorem]{Proposition}
\makeatletter

\usepackage{babel}
\usepackage{subfig}
\graphicspath{{Dirac/}}
           
\makeatother
\DeclareMathSizes{14}{14}{9.8}{7}

\usepackage{lipsum}

\begin{document}

\title{On the Poisson brackets of hybrid plasma  models with kinetic ions and massless electrons}
\date{}
\author[1]{Yingzhe Li}
\author[2]{Philip J. Morrison}
\author[1]{Stefan Possanner}
\author[1,3]{Eric Sonnendr\"ucker}
\affil[1]{Max Planck Institute for Plasma Physics, Boltzmannstrasse 2, 85748 Garching, Germany}
\affil[2]{Department of Physics and Institute for Fusion Studies, The University of Texas at Austin, Austin, TX78712, USA}
\affil[3]{Department of Mathematics, Technical University of Munich, Boltzmannstrasse 3, 85748 Garching, Germany}

\maketitle
\begin{abstract}
We investigate the conditions under which the Jacobi identity holds for a class of recently introduced anti-symmetric brackets for the hybrid plasma models with kinetic ions and massless electrons. In particular, we establish the precise conditions under which the brackets for the vector-potential-based formulations satisfy the Jacobi identity, and demonstrate that these conditions are fulfilled by all physically relevant functionals. Moreover, for the magnetic-field-based formulation, we show that the corresponding anti-symmetric bracket constitutes a Poisson bracket under the divergence-free condition of the magnetic field, and we provide a direct proof of the Jacobi identity. These results are further extended to models incorporating electron entropy as well as more general hybrid kinetic-fluid models.

\end{abstract}


\section{Introduction}

Many important physical models can be formulated in a Hamiltonian framework, 
 such as Maxwell's equations, and Schr\"odinger equation, the equations of which can be derived with the appropriate Hamiltonians and Poisson brackets~\cite{symmetry}. In plasma physics, Hamiltonian formulations with corresponding Poisson brackets have been developed for various models, including the Vlasov--Maxwell equations~\cite{morrisonvm, MW}, magnetohydrodynamics (MHD)~\cite{morrisonmhd}, two fluid model~\cite{twofluid}, and extended MHD~\cite{holm,Abdelhamid}. Recently, some Poisson brackets have been proposed for the hybrid plasma models in~\cite{Tronci}.

For the hybrid plasma model with kinetic ions and massless electrons (HKM), the Poisson brackets proposed in~\cite{Tronci} involve the charge densities of the ions and the electrons, which are assumed to be equal under the quasi-neutrality condition. Preserving quasi-neutrality numerically requires special care of both charge densities during simulations, particularly when particle-in-cell methods are used for ions and grid-based methods are used for electrons, as discussed in~\cite{LHPS, LHPS2}. From a theoretical perspective, the quasi-neutrality condition can be understood as a constraint of non-dissipative plasma models. The Dirac constraint theory for imposing constraints has been considered and developed in recent works~\cite{poissondirac, Kaltsasvp, burbyvp}.

In this work, we show that 
the quasi-neutrality condition of the HKM model corresponds to a momentum map in the canonical momentum based formulation due to the gauge symmetry and a Casimir functional in the velocity based formulations. In~\cite{LHPS, LHPS2}, the simplified anti-symmetric brackets are derived via neglecting terms associated with the functional derivatives of the electron charge density in the Poisson brackets given in~\cite{Tronci}.  In this work, we investigate the Jacobi identities for these simplified anti-symmetric brackets~\cite{LHPS, LHPS2}. Specifically, 
we determine the conditions for the Jacobi identities of the vector potential-based anti-symmetric brackets proposed in~\cite{LHPS2}, which are satisfied by all the physically meaningful functionals. Furthermore,
we show that the magnetic field-based anti-symmetric bracket proposed in~\cite{LHPS} satisfies the Jacobi identity as a Poisson bracket, provided that the divergence free condition of the magnetic field $\nabla \cdot {\mathbf B} = 0$ holds, and we present a proof via direct calculations of the Jacobi identity. 

Firstly, we present the dimensionless HKM model~\cite{LHPS} as follows, \begin{equation} \label{model}
\begin{aligned}
\text{kinetic ions:}\qquad &\frac{\partial f}{\partial t} + {\mathbf v} \cdot \frac{\partial f}{\partial \xb} +  ({\mathbf E} + {\mathbf v} \times {\mathbf B}) \cdot \frac{\partial f}{\partial {\mathbf v}} = 0\,,
\\[2mm]
\text{Faraday's law:}\qquad &\frac{\partial {\mathbf B}}{\partial t} = - \nabla \times {\mathbf E}\,,\qquad \nabla \cdot {\mathbf B} = 0\,,
\\[1mm]
\text{mass-less fluid  electrons:} \qquad & \frac{\partial n_e}{\partial t} +  \nabla \cdot \left(n_e {\mathbf u}_e  \right) = 0, \quad {\mathbf u}_e = {\mathbf u} - \frac{\mathbf J}{ne}\,,
\\[1mm]
\text{Ohm's law:}  \qquad & {\mathbf E} = - T_e \frac{\nabla n_e}{n_e} - \left(\mathbf u - \frac{\mathbf J}{n_e }  \right) \times {\mathbf B} \,.
\end{aligned}
\end{equation}
Here we take the single ion species case for an example, have assumed ion has an unit charge, 
  $f({\mathbf x}, {\mathbf v},t)$ denotes the ion distribution function depending on time $t \in \mathbb{R}$, position ${\mathbf x} \in \mathbb{R}^3$ and velocity ${\mathbf v} \in \mathbb{R}^3$, quasi-neutrality number density is $n_e = n$, where $n = \int f\,\mathrm{d}{\mathbf v}$ and $n_e$ are ion and electron charge densities, respectively, ${\mathbf u} =  \int {\mathbf v} f\, \mathrm{d}{\mathbf v}/n$ is the ion current carrying drift velocity,
$\mathbf E({\mathbf x},t)$ and $\mathbf B({\mathbf x},t)$ stand for the electromagnetic fields, $\mathbf J =\nabla \times {\mathbf B}$ denotes the plasma current, and the constant $T_e$ is the electron temperature. The Hamiltonian for the above hybrid model~\eqref{model} is given in~\cite{Tronci, LHPS} as
\begin{equation*}
    H(f, {\mathbf B}, n_e) = \frac{1}{2} \int |{\mathbf v}|^2 f \, \mathrm{d}{\mathbf v}\mathrm{d}{\mathbf x} + \frac{1}{2}\int |{\mathbf B}|^2 \, \mathrm{d}{\mathbf x} +T_e \int n_e \ln n_e \,\mathrm{d}{\mathbf x}.
\end{equation*}

The Poisson bracket of~\eqref{model} (called the xvB formulation in~\cite{LHPS2}) 
is obtained by simply transforming the Poisson bracket (the formula (76) in~\cite{Tronci}) in terms of the vector potential ${\mathbf A}$ to the Poisson bracket in terms of the magnetic field via the following two relations~\cite{Tronci,MW}
$
{\mathbf B} = \nabla \times {\mathbf A}, \frac{\delta F}{\delta {\mathbf A}} = \nabla \times \frac{\delta F}{\delta {\mathbf B}}.
$ Also the canonical momentum based Poisson bracket is given in~\cite{Tronci}, in which the distribution function depends on the canonical momentum ${\mathbf p}$ defined as ${\mathbf p} = {\mathbf v} + {\mathbf A}$. This formulation is called the xpA formulation in~\cite{LHPS2}. The third formulation, i.e., also refered to as the xvA formulation, includes the distribution function $f$ depending on velocity ${\mathbf v}$, and the vector potential ${\mathbf A}$.
Note that in the hybrid model~\eqref{model} and the Poisson brackets given in~\cite{Tronci}, the electron charge density $n_e$ is regarded as an independent unknown. 


 This work is organized as follows. The xpA, xvA, and xvB formulations of the HKM model are investigated in Section~\ref{sec:brackets}. 
 The conditions of the Jacobi identity of the simplified brackets proposed in~\cite{LHPS, LHPS2} are given in subsection~\ref{sec:xpa}-\ref{sec:xvB}, respectively. The direct proof of the Jacobi identity of the simplified Poisson bracket of the xvB formulation is presented  in subection~\ref{sec:xvB}. We establish relationships among the simplified Poisson brackets of these three formulations using the chain rules for functional derivatives in subection~\ref{sec:relation}. The extensions to more general hybrid models are briefly mentioned in Section~\ref{sec:general}.


\section{The Poisson brackets of the HKM model}\label{sec:brackets}
In this section, we investigate the three formulations of the HKM model. The quasi-neutrality condition corresponds to a momentum map in the canonical momentum based formulation due to the gauge symmetry and a Casimir functional in the velocity based formulations.
We also give the conditions of the Jacobi identities of the three anti-symmetric brackets proposed in~\cite{LHPS,LHPS2}.

\subsection{xpA formulation}\label{sec:xpa}
We start with the xpA formulation, the Poisson brackets of which are simple.
The following Poisson bracket is proposed in~\cite{Tronci} (in the following the time variable $t$ is omitted in the unknowns for simplicity), 
\begin{equation}
\label{eq:xpanebrackethelp}
    \begin{aligned}
 \{F, G  \}(f({\mathbf x}, {\mathbf p}), {\mathbf A}, n_e) & =  \int f \left[ \frac{\delta F}{\delta f}, \frac{\delta G}{\delta f} \right]_{xp}  \mathrm{d}{\mathbf x}\mathrm{d}{\mathbf p} - \int \frac{1}{n_e} \nabla \times {\mathbf A} \cdot \left( \frac{\delta F}{\delta {\mathbf A}} \times  \frac{\delta G}{\delta {\mathbf A}} \right)  \mathrm{d}{\mathbf x} \\
 & + \int \frac{\delta F}{\delta {\mathbf A}} \cdot \frac{\partial }{\partial {\mathbf x}} \frac{\delta G}{\delta n_e} -  \frac{\delta G}{\delta {\mathbf A}} \cdot \frac{\partial }{\partial {\mathbf x}}\frac{\delta F}{\delta n_e}\, \mathrm{d}{\mathbf x}.
        \end{aligned}
\end{equation}
Here we give a proof of the Jacobi identity of the Poisson bracket~\eqref{eq:xpanebrackethelp}.
\begin{theorem}\label{thm:xpa}
The bracket~\eqref{eq:xpanebrackethelp} satisfies the Jacobi identity.
\end{theorem}
\begin{proof}
 By the bracket theorem in~\cite{lifting}, to prove the Jacobi identity
for bracket $$\{F, G\}(\psi) = \int \frac{\delta F}{\delta \psi} J(\psi) \frac{\delta G}{\delta \psi}\, \mathrm{d}{\mathbf a},$$ where $J$ is a operator depending on $\psi$ and $\psi$ depends on ${\mathbf a}$, only the explicit dependence of $J$ on $\psi$ is needed to consider when
taking the functional derivative $\frac{\delta \{F, G \}}{\delta \psi}$. We calculate the functional derivatives of $ \{F, G  \}$ (omitting the second order functional derivatives) and get  
\begin{equation*}
    \begin{aligned}
    \frac{ \delta  \{F, G  \}  }{\delta f}  =\left[ \frac{\delta F}{\delta f}, \frac{\delta G}{\delta f} \right]_{xp}, 
    \frac{ \delta  \{F, G  \}  }{\delta {\mathbf A}}  = - \nabla \times \left( \frac{1}{n_e} \frac{\delta F}{\delta {\mathbf A}} \times  \frac{\delta G}{\delta {\mathbf A}} \right), 
     \frac{ \delta  \{F, G  \}  }{\delta n_e} = \frac{\nabla \times {\mathbf A}}{n_e^2}  \cdot \left( \frac{\delta F}{\delta {\mathbf A}} \times  \frac{\delta G}{\delta {\mathbf A}} \right).
      \end{aligned}
\end{equation*}
Then for any three smooth functionals $F, G$, and $K$ depending on $f, {\mathbf A}$, and $n_e$, we have
\begin{equation}
\label{eq:ne}
    \begin{aligned}
\{ \{F, G  \}, K \}(f({\mathbf x}, {\mathbf p}), & {\mathbf A}, n_e) = \int f \left[ \left[ \frac{\delta F}{\delta f}, \frac{\delta G}{\delta f} \right]_{xp}, \frac{\delta K}{\delta f} \right]_{xp} \mathrm{d}{\mathbf x} \mathrm{d}{\mathbf p} \\
& +  \int \frac{1}{n_e} \nabla \times {\mathbf A} \cdot \left( \nabla \times \left( \frac{1}{n_e} \frac{\delta F}{\delta {\mathbf A}} \times  \frac{\delta G}{\delta {\mathbf A}} \right) \times \frac{\delta K}{\delta {\mathbf A}} \right) \mathrm{d}{\mathbf x} \\
&- \underbrace{\int \nabla \times \left( \frac{1}{n_e}\left( \frac{\delta F}{\delta {\mathbf A}} \times  \frac{\delta G}{\delta {\mathbf A}} \right)  \right) \cdot \frac{\partial}{\partial {\mathbf x}} \frac{\delta K}{\delta n_e}\, \mathrm{d}{\mathbf x}}_{=0 \text{\ via integration by parts}} \\ 
& - \int \frac{\delta K}{\delta {\mathbf A}} \cdot \frac{\partial }{\partial {\mathbf x}} \left( \frac{1}{n_e^2} \nabla \times {\mathbf A} \cdot \left( \frac{\delta F}{\delta {\mathbf A}} \times  \frac{\delta G}{\delta {\mathbf A}} \right) \right)  \mathrm{d}{\mathbf x}.
  \end{aligned}
\end{equation}
The sum of the permutation of the first term is zero as it is a Lie--Poisson bracket~\cite{morrisonvm,MW}. The sum of the permutation of the second and fourth terms is zero can be proved using the lemma 4.1 in~\cite{Abdelhamid} by letting ${\mathbf p} = {\mathbf q} = {\mathbf r} = {\mathbf A}$ and $\rho = n_e$. 
Then we have shown that the bracket~\eqref{eq:xpanebrackethelp} satisfies the Jacobi identity, i.e., 
$\{ \{F, G  \}, K \} + \{ \{G, K  \}, F \} + \{ \{K, F  \}, G \} = 0 $.
\end{proof}
Based on theorem~\ref{thm:xpa}, we know that $P = C^{\infty}(\mathbb{R}^6) \times [C^{\infty}(\mathbb{R}^3)]^3 \times C^{\infty}(\mathbb{R}^3)$ is a Poisson manifold~\cite{symmetry} with the above Poisson bracket~\eqref{eq:xpanebrackethelp}.
We define the action $\Phi$ of the Lie group $G = C^\infty(\mathbb{R}^3)$ on the Poisson manifold $P$, i.e., $\Phi: G \times P \rightarrow P$, as
\begin{equation}\label{eq:actionxpa}
\Phi_\psi\left(f({\mathbf x}, {\mathbf p}), {\mathbf A}, n_e \right) = \left(f({\mathbf x}, {\mathbf p} - \nabla \psi), {\mathbf A} + \nabla \psi, n_e \right), \quad \forall \psi \in G.
\end{equation}
The Hamiltonian $$
H(f({\mathbf x}, {\mathbf p}), {\mathbf A}, n_e)= \frac{1}{2} \int |{\mathbf p} - {\mathbf A}|^2 f\, \mathrm{d}{\mathbf p}\mathrm{d}{\mathbf x} + \frac{1}{2} \int |\nabla \times {\mathbf A}|^2 \, \mathrm{d}{\mathbf x} + T_e \int n_e \ln n_e \, \mathrm{d}{\mathbf x}
$$
is invariant under the group action~\eqref{eq:actionxpa}. And the Lie group $G$ acts on the Poisson manifold $P$ canonically~\cite{symmetry}, i.e., 
$$
\Phi_\psi^*\{F, G \} = \{\Phi_\psi^*F, \Phi_\psi^*G \},
$$
where $\Phi_\psi^*$ denotes  the pull back of the map $\Phi_\psi$. Next we prove the quasi-neutrality condition is related with a momentum map associated with the Lie group action~\eqref{eq:actionxpa}.
\begin{proposition}
The map ${\mathbf J}: P \rightarrow g^*$
\begin{equation}\label{eq:mm}
{\mathbf J}\left(f({\mathbf x}, {\mathbf p}), {\mathbf A}, n_e \right) =  n_e - \int f({\mathbf x}, {\mathbf p}) \, \mathrm{d}{\mathbf p},
\end{equation}
 is a momentum map of the action defined in~\eqref{eq:actionxpa}, where $g = C^{\infty}(\mathbb{R}^3)$ is the Lie algebra of the Lie group $G = C^{\infty}(\mathbb{R}^3)$, and $g^*$ is the dual space of $g$ in the sense of $L^2$ integral.
\end{proposition}
\begin{proof}
For $\phi \in g$, 
the infinitesimal generator $\phi_P$ is 
$$
\phi_P(f({\mathbf x}, {\mathbf p}), {\mathbf A}, n_e) = \frac{\mathrm{d}}{\mathrm{d}s} \Phi_{\exp(s\phi)}\left(f({\mathbf x}, {\mathbf p}), {\mathbf A}, n_e\right)|_{s=0} = \left(-\frac{\partial \phi}{\partial {\mathbf x}} \cdot \frac{\partial f}{\partial {\mathbf p}}, \frac{\partial \phi}{\partial {\mathbf x}}, 0\right).
$$
For the linear map $J: g \rightarrow \mathcal{F}(P)$, where $\mathcal{F}(P)$ is the set of all smooth functionals defined on $P$,
$$
J(\phi)(f, {\mathbf A}, n_e) = \int n_e\phi\, \mathrm{d}{\mathbf x} - \int f \phi\, \mathrm{d}{\mathrm{x}}\mathrm{d}{\mathbf p},
$$
we can check that the vector field $\phi_P$ is the Hamiltonian vector field (with the Poisson bracket~\eqref{eq:xpanebrackethelp}) of the functional $J(\phi)$.
By the definition of momentum map ${\mathbf J}: P \rightarrow g^*$ in the definition 11.2.1 in~\cite{symmetry}, we know that 
$$
<{\mathbf J}(f, {\mathbf A}, n_e), \phi> = J(\phi)(f, {\mathbf A}, n_e),
$$
where $<\cdot, \cdot>$ denotes the $L^2$ inner product. Then we have 
${\mathbf J}(f, {\mathbf A}, n_e) = n_e - \int f({\mathbf x}, {\mathbf p},t) \, \mathrm{d}{\mathbf p}$.
\end{proof}

\noindent{\bf Conservation of the Momentum Map} The above Lie algrbra $g$ action, which comes from the canonical left Lie group action $\Phi$, acts canonically on the Poisson manifold $P$ and admits the above momentum mapping~\eqref{eq:mm}. From the theorem 11.4.1 in~\cite{symmetry}, we know the momentum map is a constant of the motion for the Hamiltonian, i.e., the quasi-neutrality always holds if it holds initially.

\noindent{\bf Simplified bracket} In~\cite{LHPS2} we prefer to use the formulation with the following anti-symmetric bracket for the numerical discretization, due to the low level numerical noise  and good conservation properties, which is obtained by removing the last term  containing the functional derivative of $n_e$ in~\eqref{eq:xpanebrackethelp}, and replacing $n_e$ by $\int f \, \mathrm{d}{\mathbf p}$ because of the quasi-neutrality condition,
\begin{equation}
\label{eq:1}
    \begin{aligned}
        \{F, G  \}_l(f({\mathbf x}, {\mathbf p}), {\mathbf A}) =  \int f \left[ \frac{\delta F}{\delta f}, \frac{\delta G}{\delta f} \right]_{xp}  \mathrm{d}{\mathbf x}\mathrm{d}{\mathbf p} - \int \frac{1}{n} \nabla \times {\mathbf A} \cdot \left( \frac{\delta F}{\delta {\mathbf A}} \times  \frac{\delta G}{\delta {\mathbf A}} \right)  \mathrm{d}{\mathbf x},
    \end{aligned}
\end{equation}
where $n = \int f \, \mathrm{d}{\mathbf p}$ denotes the ion/electron charge density. 
This bracket
does not satisfy Jacobi identity, as shown below.
\begin{theorem}
The bracket~\eqref{eq:1} does not satisfy Jacobi identity.
\end{theorem}
\begin{proof}
In order to consider the Jacobi identity, we need to 
calculate the functional derivatives of $ \{F, G  \}$ (ignoring the second order derivatives by the bracket theorem in~\cite{lifting}), then we have 
\begin{equation*}
    \begin{aligned}
\frac{\delta \{F, G  \}_l}{\delta f} = \left[ \frac{\delta F}{\delta f}, \frac{\delta G}{\delta f} \right]_{xp} + \frac{1}{n^2} \nabla \times {\mathbf A} \cdot  \left( \frac{\delta F}{\delta {\mathbf A}} \times  \frac{\delta G}{\delta {\mathbf A}} \right), \quad 
\frac{\delta \{F, G  \}_l}{\delta {\mathbf A}} = - \nabla \times \left( \frac{1}{n} \frac{\delta F}{\delta {\mathbf A}} \times  \frac{\delta G}{\delta {\mathbf A}} \right).
       \end{aligned}
\end{equation*} 
Based on the above functional derivatives, we have
\begin{equation}
\label{eq:nne}
    \begin{aligned}
&\{ \{F, G  \}_l, K \}_l =  \int \frac{1}{n} \nabla \times {\mathbf A} \cdot \left( \nabla \times \left( \frac{1}{n} \frac{\delta F}{\delta {\mathbf A}} \times  \frac{\delta G}{\delta {\mathbf A}} \right) \times \frac{\delta K}{\delta {\mathbf A}} \right) \mathrm{d}{\mathbf x}\\ 
& + \int f \left[ \left[ \frac{\delta F}{\delta f}, \frac{\delta G}{\delta f} \right]_{xp}, \frac{\delta K}{\delta f} \right]_{xp} \mathrm{d}{\mathbf x} \mathrm{d}{\mathbf p} + \int f  \left[  \frac{1}{n^2} \nabla \times {\mathbf A} \cdot  \left( \frac{\delta F}{\delta {\mathbf A}} \times  \frac{\delta G}{\delta {\mathbf A}} \right), \frac{\delta K}{\delta f} \right]_{xp} \mathrm{d}{\mathbf x} \mathrm{d}{\mathbf p}.
    \end{aligned}
\end{equation}  
For functionals $F, G$ depending only on ${\mathbf A}$, we define specially a functional depending only on $f$ as $K = \int f {\mathbf p} \cdot \nabla h({\mathbf x})\, \mathrm{d}{\mathbf p}\mathrm{d}{\mathbf x}$, where  $h({\mathbf x}) =  \frac{1}{n^2} \nabla \times {\mathbf A}_0 \cdot  \left( \frac{\delta F}{\delta {\mathbf A}}|_{{\mathbf A} =
{\mathbf A}_0} \times  \frac{\delta G}{\delta {\mathbf A}}|_{{\mathbf A} = {\mathbf A}_0} \right)$. Then we have 
$$\{ \{F, G  \}_l, K \}_l|_{(f, {\mathbf A}_0)} + \text{cyc} = \int f |\nabla h|^2\, \mathrm{d}{\mathbf x}\mathrm{d}{\mathbf p},$$
where the 'cyc' denotes the results via the permutations of $F, G$, and $K$.
The integrand is non-negative when $f$ is non-negative, and thus $\{ \{F, G  \}_l, K \}_l + \text{cyc}$ is not zero at $(f, {\mathbf A}_0)$ for positive $f|\nabla h|^2$. Thus we have shown that the anti-symmetric bracket~\eqref{eq:1} does not satisfy Jacobi identity and is not a Poisson bracket.
\end{proof}

A natural question is under what kind of condition is the Jacobi identity satisfied by the bracket~\eqref{eq:1}.
\begin{theorem}
The bracket~\eqref{eq:1} satisfies the Jacobi identity when the functionals satisfy the condition,
\begin{equation}\label{eq:conditionxpa}
    \nabla \cdot \frac{\delta K}{\delta {\mathbf A}} = - \nabla \cdot \int f \frac{\partial}{\partial {\mathbf p}} \frac{\delta K}{\delta f} \,\mathrm{d}{\mathbf p}, \quad \forall K.
\end{equation}

\end{theorem}
\begin{proof}
By comparing~\eqref{eq:ne} with~\eqref{eq:nne}, we know that when the condition~\eqref{eq:conditionxpa} is satisfied, we have 
$$
 \{ \{F, G  \}_l, K \}_l(f({\mathbf x}, {\mathbf p}), {\mathbf A}) =
 \{ \{F, G  \}, K \}(f({\mathbf x}, {\mathbf p}), {\mathbf A}, n_e) = 0
$$
where the '$=$' holds in the sense that $n_e = n$ and we ignore the second order functional derivatives when calculating the functional derivatives of the brackets~\eqref{eq:xpanebrackethelp} and~\eqref{eq:1}.  Then this theorem is proved.
\end{proof}

\begin{remark}
The condition~\eqref{eq:conditionxpa} is satisfied by the Hamiltonian of the hybrid model
\begin{equation*}
H(f({\mathbf x}, {\mathbf p}),{\mathbf A}) = \frac{1}{2}\int |{\mathbf p} - {\mathbf A}|^2 f\, \mathrm{d}{\mathbf p} \mathrm{d}{\mathbf x} + \frac{1}{2} \int |\nabla \times {\mathbf A}|^2 \mathrm{d}{\mathbf x} + T_e \int n \ln n\, \mathrm{d}{\mathbf x}.
\end{equation*}
\end{remark}

We here define the following set of functionals satisfying the condition~\eqref{eq:conditionxpa}, 
\begin{equation}\label{eq:classxpa}
\mathcal{A} = \{ F(f({\mathbf x}, \, {\mathbf p}), \, {\mathbf A}) =  \tilde{F}(\tilde{f}({\mathbf x}, \, {\mathbf v}), \, {\mathbf B})| {\mathbf B} = \nabla \times {\mathbf A}, \, {\mathbf v} = {\mathbf p} - {\mathbf A}, \, f({\mathbf x}, {\mathbf p}) := \tilde{f}({\mathbf x}, {\mathbf v}) = \tilde{f}({\mathbf x}, {\mathbf p} - {\mathbf A}) \},
\end{equation}
where the time variable is omitted for convenience. This set $\mathcal{A}$ includes  functionals determined by all smooth functionals about $\tilde{f}({\mathbf x}, \, {\mathbf v}), \, {\mathbf B}$ with $\nabla \cdot {\mathbf B} = 0$.
\begin{remark}
The bracket~\eqref{eq:1} can be obtained by the following transformation from the bracket~\eqref{eq:xpanebrackethelp},
$$
\left(f({\mathbf x}, {\mathbf p}),\, {\mathbf A}, \,n_e \right) \rightarrow \left(f({\mathbf x}, {\mathbf p}),\, {\mathbf A},\, C = n_e - \int f \, \mathrm{d}{\mathbf p} \right).
$$
Any functional $\tilde{F}$ of $ \left(f({\mathbf x}, {\mathbf p}), {\mathbf A}, C \right)$ can be regarded as a functional $F$ of $\left(f({\mathbf x}, {\mathbf p}), {\mathbf A}, n_e \right)$ via 
$F\left(f({\mathbf x}, {\mathbf p}), {\mathbf A}, n_e \right) = \tilde{F}\left(f({\mathbf x}, {\mathbf p}), {\mathbf A}, C\right)$.
By the chain rules of the functional derivatives we have 
$$
\frac{\delta F}{\delta f} = \frac{\delta \tilde{F}}{\delta f} - \frac{\delta \tilde{F}}{\delta C},\quad  \frac{\delta F}{\delta {\mathbf A}} = \frac{\delta \tilde{F}}{\delta {\mathbf A}}, \quad \frac{\delta F}{\delta C} =  \frac{\delta \tilde{F}}{\delta C}. 
$$
By plugging the above relations into bracket~\eqref{eq:xpanebrackethelp}, we get
\begin{equation}
\label{eq:transxpa}
    \begin{aligned}
 \{\tilde{F}, & \tilde{G}  \}(f({\mathbf x}, {\mathbf p}), {\mathbf A}, C)  =  \int f \left[ \frac{\delta \tilde{F}}{\delta f}, \frac{\delta \tilde{G}}{\delta f} \right]_{xp}  \mathrm{d}{\mathbf x}\mathrm{d}{\mathbf p} - \int \frac{1}{C + \int f \, \mathrm{d}{\mathbf p}} \nabla \times {\mathbf A} \cdot \left( \frac{\delta \tilde{F}}{\delta {\mathbf A}} \times  \frac{\delta \tilde{G}}{\delta {\mathbf A}} \right)  \mathrm{d}{\mathbf x} \\
 & + \int \frac{\delta \tilde{F}}{\delta {\mathbf A}} \cdot \frac{\partial }{\partial {\mathbf x}} \frac{\delta \tilde{G}}{\delta C} -  \frac{\delta \tilde{G}}{\delta {\mathbf A}} \cdot \frac{\partial }{\partial {\mathbf x}}\frac{\delta \tilde{F}}{\delta C} \, \mathrm{d}{\mathbf x} + \int f \left( -\frac{\partial}{\partial {\mathbf x}}\frac{\delta \tilde{F}}{\delta C} \cdot \frac{\partial}{\partial {\mathbf p}}\frac{\delta \tilde{G}}{\delta f} + \frac{\partial}{\partial {\mathbf x}}\frac{\delta \tilde{G}}{\delta C} \cdot \frac{\partial}{\partial {\mathbf p}}\frac{\delta \tilde{F}}{\delta f}  \right) \mathrm{d}{\mathbf p} \mathrm{d}{\mathbf x}.
        \end{aligned}
\end{equation}
The sum of the last two terms is zero when $\tilde{F}$ and $\tilde{G}$ are in $\mathcal{A}$~\eqref{eq:classxpa}. Then for functionals $\tilde{F}$ and $\tilde{G}$ in $\mathcal{A}$, bracket~\eqref{eq:transxpa} becomes
$$
\{\tilde{F}, \tilde{G}  \}(f({\mathbf x}, {\mathbf p}), {\mathbf A}, C)  =  \int f \left[ \frac{\delta \tilde{F}}{\delta f}, \frac{\delta \tilde{G}}{\delta f} \right]_{xp}  \mathrm{d}{\mathbf x}\mathrm{d}{\mathbf p} - \int \frac{1}{C + \int f \, \mathrm{d}{\mathbf v}} \nabla \times {\mathbf A} \cdot \left( \frac{\delta \tilde{F}}{\delta {\mathbf A}} \times  \frac{\delta \tilde{G}}{\delta {\mathbf A}} \right)  \mathrm{d}{\mathbf x},
$$
which is the bracket~\eqref{eq:1} when $C = 0$, no functional derivative about $C$ is conducted, and $C$ can be regarded as a parameter function.
\end{remark}

We still need to answer that if $\{F, G\}_l \in \mathcal{A}$ for $\forall F, G \in \mathcal{A}$, before we can say that the bracket~\eqref{eq:1} is a Poisson bracket for the functionals in $\mathcal{A}$. The answer is 'yes',  which shall be proved in Theorem~\ref{thm:relation} after the magnetic field based formulation, i.e., the xvB formulation,  is investigated.

\noindent{\bf Coulomb gauge}
When Coulomb gauge is used, i.e., $\nabla \cdot {\mathbf A} = 0$, we let $\psi_A = - \Delta^{-1}(\nabla \cdot {\mathbf A})$, then the map $\left(f({\mathbf x}, {\mathbf p}), {\mathbf A}\right) \rightarrow \left(f({\mathbf x}, {\mathbf p} - \nabla \psi_A), {\mathbf A} + \nabla \psi_A\right)$~\cite{MW} transforms the bracket~\eqref{eq:1} as
\begin{equation*}
    \begin{aligned}
         \{F, G  \}^{\text{Coulomb}}_A(f({\mathbf x}, {\mathbf p}, t), {\mathbf A})  =  \int f \left[ \frac{\delta F}{\delta f}, \frac{\delta G}{\delta f} \right]_{xp}  \mathrm{d}{\mathbf x}\mathrm{d}{\mathbf p} - \int \frac{1}{n} \nabla \times {\mathbf A} \cdot \left({\mathbf M}\times {\mathbf N} \right) \mathrm{d}{\mathbf x},
    \end{aligned}
\end{equation*}
where $${\mathbf M} =  \frac{\delta F}{\delta {\mathbf A}} + \nabla \Delta^{-1} \nabla \cdot \frac{\delta F}{\delta {\mathbf A}} + \nabla \Delta^{-1} \nabla \cdot 
 \int f \frac{\partial}{\partial {\mathbf p}}\frac{\delta F}{\delta f} \, \mathrm{d}{\mathbf p},$$
 $${\mathbf N} = \frac{\delta G}{\delta {\mathbf A}} + \nabla \Delta^{-1} \nabla \cdot \frac{\delta G}{\delta {\mathbf A}} + \nabla \Delta^{-1} \nabla \cdot  
 \int f \frac{\partial}{\partial {\mathbf p}}\frac{\delta G}{\delta f} \, \mathrm{d}{\mathbf p}. $$

When the magnetic field ${\mathbf B}$ is adopted, we have the following bracket by ${\mathbf B} = \nabla \times {\mathbf A}$ and $\frac{\delta F}{\delta {\mathbf A}} = \nabla \times \frac{\delta F}{\delta {\mathbf B}}$,
\begin{equation*}
    \begin{aligned}
         \{F, G  \}^{\text{Coulomb}}_{B}(f({\mathbf x}, {\mathbf p}, t), {\mathbf B})  =  \int f \left[ \frac{\delta F}{\delta f}, \frac{\delta G}{\delta f} \right]_{xp}  \mathrm{d}{\mathbf x}\mathrm{d}{\mathbf p} - \int \frac{1}{n} {\mathbf B} \cdot \left(\tilde{\mathbf M}\times \tilde{\mathbf N} \right) \mathrm{d}{\mathbf x},
    \end{aligned}
\end{equation*}
where $$\tilde{\mathbf M} =  \nabla \times \frac{\delta F}{\delta {\mathbf B}} + \nabla \Delta^{-1} \nabla \cdot 
 \int f \frac{\partial}{\partial {\mathbf p}}\frac{\delta F}{\delta f} \, \mathrm{d}{\mathbf p}, \quad \tilde{\mathbf N} = \nabla \times \frac{\delta G}{\delta {\mathbf B}} + \nabla \Delta^{-1} \nabla \cdot  
 \int f \frac{\partial}{\partial {\mathbf p}}\frac{\delta G}{\delta f} \, \mathrm{d}{\mathbf p}. $$

\subsection{xvA formulation}\label{sec:xva}
Next we consider the xvA formulation~\cite{Tronci,chacon,LHPS2}, for which the Poisson bracket is given in the formula (76) in~\cite{Tronci} as,
\begin{equation}
\label{eq:fxvane}
    \begin{aligned}
        &\{ F, G\}(f({\mathbf x}, {\mathbf v}), {\mathbf A}, n_e)  = \int f \left[ \frac{\delta F}{\delta f}, \frac{\delta G}{\delta f}  \right]_{xv}\mathrm{d}{\mathbf x} \mathrm{d}{\mathbf v} + \int f \nabla \times {\mathbf A} \cdot \left( \frac{\partial}{\partial {\mathbf v}} \frac{\delta F}{\delta f} \times  \frac{\partial}{\partial {\mathbf v}} \frac{\delta G}{\delta f}   \right) \mathrm{d}{\mathbf x} \mathrm{d}{\mathbf v}\\
        & - \int \frac{1}{n_e} \nabla \times {\mathbf A} \cdot \left( \frac{\delta F}{\delta {\mathbf A}} \times  \frac{\delta G}{\delta {\mathbf A}}  \right)  \mathrm{d}{\mathbf x} - \int \frac{f}{n_e} \nabla \times {\mathbf A} \cdot \left( \frac{\partial }{\partial {\mathbf v}} \frac{\delta G}{\delta f} \times \frac{\delta F}{\delta {\mathbf A}} -  \frac{\partial }{\partial {\mathbf v}} \frac{\delta F}{\delta f} \times \frac{\delta G}{\delta {\mathbf A}}  \right) \mathrm{d}{\mathbf x} \mathrm{d}{\mathbf v}\\
        & - \int \frac{1}{n_e} \nabla \times {\mathbf A} \cdot \left( \int f \frac{\partial }{\partial {\mathbf v}} \frac{\delta F}{\delta f} \mathrm{d}{\mathbf v} \times \int f \frac{\partial }{\partial {\mathbf v}} \frac{\delta G}{\delta f} \mathrm{d}{\mathbf v} 
         \right)  \mathrm{d} {\mathbf x} + \int \frac{\delta F}{\delta {\mathbf A}} \cdot  \frac{\partial}{\partial {\mathbf x}}\frac{\delta G}{\delta n_e} 
         - \frac{\delta G}{\delta {\mathbf A}} \cdot  \frac{\partial}{\partial {\mathbf x}}\frac{\delta F}{\delta n_e} \mathrm{d}{\mathbf x} \\
         & + \int f \left(-\frac{\partial}{\partial {\mathbf v}} 
         \frac{\delta F}{\delta f} \cdot  \frac{\partial}{\partial {\mathbf x}} \frac{\delta G}{\delta n_e} + \frac{\partial}{\partial {\mathbf v}} 
         \frac{\delta G}{\delta f} \cdot  \frac{\partial}{\partial {\mathbf x}} \frac{\delta F}{\delta n_e} \right)\mathrm{d}{\mathbf x}  \mathrm{d}{\mathbf v} = \{F,G \}_0 + \{F,G \}_{n_e},
    \end{aligned}
\end{equation}
where we denote the sum of the first six terms in the above bracket as $\{F,G \}_0$, and we denote the last term as $\{F,G \}_{n_e}$.
In order to verify the Jacobi identity, as above we firstly calculate the functional derivatives of the $\{F, G\}$ (omit the second order derivatives by the bracket theorem in~\cite{lifting}),
then we have 
\begin{equation}
\label{eq:xvanefgh}
    \begin{aligned}
        \{\{ F, G\},K\} & =  - \int \frac{1}{n_e} \nabla \times {\mathbf A} \cdot \left( \frac{\delta \{ F, G\}}{\delta {\mathbf A}} \times  \frac{\delta K}{\delta {\mathbf A}}  \right)  \mathrm{d}{\mathbf x}\\
        &+ \int f \left[ \frac{\delta \{ F, G\}_0}{\delta f}, \frac{\delta K}{\delta f}  \right]_{xv}  \mathrm{d}{\mathbf x} \mathrm{d}{\mathbf v} + \int f \nabla \times {\mathbf A} \cdot \left( \frac{\partial}{\partial {\mathbf v}} \frac{\delta \{ F, G\}_0}{\delta f}\times  \frac{\partial}{\partial {\mathbf v}} \frac{\delta K}{\delta f}   \right) \mathrm{d}{\mathbf x} \mathrm{d}{\mathbf v}\\
        & - \int \frac{f}{n_e} \nabla \times {\mathbf A} \cdot \left( \frac{\partial }{\partial {\mathbf v}} \frac{\delta K}{\delta f} \times \frac{\delta \{ F, G\}}{\delta {\mathbf A}} -  \frac{\partial }{\partial {\mathbf v}} \frac{\delta \{ F, G\}_0}{\delta f} \times \frac{\delta K}{\delta {\mathbf A}}  \right) \mathrm{d}{\mathbf x} \mathrm{d}{\mathbf v}\\
        & - \int \frac{1}{n_e} \nabla \times {\mathbf A} \cdot \left( \int f \frac{\partial }{\partial {\mathbf v}} \frac{\delta \{ F, G\}_0}{\delta f} \mathrm{d}{\mathbf v} \times \int f \frac{\partial }{\partial {\mathbf v}} \frac{\delta K}{\delta f} \mathrm{d}{\mathbf v} 
         \right)  \mathrm{d} {\mathbf x}\\
         & + \int f \left(\frac{\partial}{\partial {\mathbf v}} 
         \frac{\delta K}{\delta f}  \cdot \frac{\partial}{\partial {\mathbf x}} \frac{\delta \{ F, G\}}{\delta n_e} \right)\mathrm{d}{\mathbf x}  \mathrm{d}{\mathbf v} - \underbrace{\int \frac{\delta K}{\delta {\mathbf A}} \cdot  \frac{\partial}{\partial {\mathbf x}}\frac{\delta \{ F, G\}}{\delta n_e} \mathrm{d}{\mathbf x}}_{T^1_{n_e}} + T_{n_e}^2,
    \end{aligned}
\end{equation}
where all the other terms related with the functional derivative about $n_e$ are included in $T^2_{n_e}$, 
\begin{equation*}
    \begin{aligned}
         T_{n_e}^2 & = \underbrace{\int \frac{\delta \{ F, G\}}{\delta {\mathbf A}} \cdot  \frac{\partial}{\partial {\mathbf x}}\frac{\delta K}{\delta n_e}  \mathrm{d}{\mathbf x}}_{=0, \ \text{via integration by parts}} + \int f \left(-\frac{\partial}{\partial {\mathbf v}} 
          \frac{\delta \{ F, G\}}{\delta f} \cdot \frac{\partial}{\partial {\mathbf x}} \frac{\delta K}{\delta n_e} \right)\mathrm{d}{\mathbf x}  \mathrm{d}{\mathbf v}\\
         & + \int f \left[ \frac{\delta \{ F, G\}_{n_e}}{\delta f}, \frac{\delta K}{\delta f}  \right]_{xv} \mathrm{d}{\mathbf x} \mathrm{d}{\mathbf v}  + \int f \nabla \times {\mathbf A} \cdot \left( \frac{\partial}{\partial {\mathbf v}} \frac{\delta \{ F, G\}_{n_e}}{\delta f}\right) \times  \frac{\partial}{\partial {\mathbf v}} \frac{\delta K}{\delta f}\,  \mathrm{d}{\mathbf x} \mathrm{d}{\mathbf v}\\
         & + \int \frac{f}{n_e} \nabla \times {\mathbf A} \cdot \left(  \frac{\partial }{\partial {\mathbf v}} \left( \frac{\delta \{ F, G\}_{n_e}}{\delta f}\right) \times \frac{\delta K}{\delta {\mathbf A}}  \right) \mathrm{d}{\mathbf x} \mathrm{d}{\mathbf v} \\
         & - \int \frac{1}{n_e} \nabla \times {\mathbf A} \cdot \left( \int f \frac{\partial }{\partial {\mathbf v}} \frac{\delta \{ F, G\}_{n_e}}{\delta f} \mathrm{d}{\mathbf v} \times \int f \frac{\partial }{\partial {\mathbf v}} \frac{\delta K}{\delta f} \mathrm{d}{\mathbf v} 
         \right)  \mathrm{d} {\mathbf x},
    \end{aligned}
\end{equation*}
which is compared with the results presented below for determining the condition of the Jacobi identity of the simplified bracket given 
in the appendix of~\cite{LHPS2}.
\begin{remark}
The quasi-neutrality condition $n_e = \int f\, \mathrm{d}{\mathbf v}$ corresponds a class of Casimir functional of the bracket~\eqref{eq:fxvane}, i.e., for any given ${\mathbf x}_0$,
$$
C(f, n_e, {\mathbf A}) = \int \delta({\mathbf x} - {\mathbf x}_0) \left( n_e - \int f \, \mathrm{d}{\mathbf v} \right) \mathrm{d}{\mathbf x}.
$$
\end{remark}

\noindent{\bf Simplified bracket} The simplified bracket given in~\cite{LHPS2} is obtained via removing the terms involving the functional derivatives about $n_e$ and replacing $n_e$ with $\int f \,\mathrm{d}{\mathbf v}$,
\begin{equation}
\label{eq:XVAbracket}
    \begin{aligned}
        & \{ F, G\}_l(f({\mathbf x}, {\mathbf v}), {\mathbf A})  = \int f \left[ \frac{\delta F}{\delta f}, \frac{\delta G}{\delta f}  \right]_{xv} \mathrm{d}{\mathbf x} \mathrm{d}{\mathbf v} + \int f \nabla \times {\mathbf A} \cdot \left( \frac{\partial}{\partial {\mathbf v}} \frac{\delta F}{\delta f} \times  \frac{\partial}{\partial {\mathbf v}} \frac{\delta G}{\delta f}   \right) \mathrm{d}{\mathbf x} \mathrm{d}{\mathbf v}\\
        & - \int \frac{1}{n} \nabla \times {\mathbf A} \cdot \left( \frac{\delta F}{\delta {\mathbf A}} \times  \frac{\delta G}{\delta {\mathbf A}}  \right)  \mathrm{d}{\mathbf x}- \int \frac{f}{n} \nabla \times {\mathbf A} \cdot \left( \frac{\partial }{\partial {\mathbf v}} \frac{\delta G}{\delta f} \times \frac{\delta F}{\delta {\mathbf A}} -  \frac{\partial }{\partial {\mathbf v}} \frac{\delta F}{\delta f} \times \frac{\delta G}{\delta {\mathbf A}}  \right) \mathrm{d}{\mathbf x} \mathrm{d}{\mathbf v}\\
        & - \int \frac{1}{n} \nabla \times {\mathbf A} \cdot \left( \int f \frac{\partial }{\partial {\mathbf v}} \frac{\delta F}{\delta f} \mathrm{d}{\mathbf v} \times \int f \frac{\partial }{\partial {\mathbf v}} \frac{\delta G}{\delta f} \mathrm{d}{\mathbf v} 
         \right)  \mathrm{d} {\mathbf x}.
    \end{aligned}
\end{equation}
In the following theorem, we give the condition of the Jacobi identity of the bracket~\eqref{eq:XVAbracket}.
\begin{theorem}
The bracket~\eqref{eq:XVAbracket} satisfies the Jacobi identity when the functionals satisfy the condition,  $$\nabla \cdot \frac{\delta K}{\delta {\mathbf A}} = 0, \quad \forall K.$$
\end{theorem}
\begin{proof}
In order to prove the Jacobi identity, we need to calculate the functional derivatives of the bracket~\eqref{eq:XVAbracket}, 
\begin{equation*}
    \begin{aligned}
        \frac{\delta \{F, G \}_l}{\delta f} = \frac{\delta \{F, G \}_0}{\delta f} + \frac{\delta \{F, G \}}{\delta n_e}, \quad \frac{\delta \{F, G \}_l}{\delta {\mathbf A}} =\frac{\delta \{F, G \}}{\delta {\mathbf A}},
    \end{aligned}
\end{equation*}
where $\{\cdot , \cdot\}$ is the bracket \eqref{eq:fxvane}, the '$=$' holds in the sense that we replace $n_e$ with $\int f\, \mathrm{d}{\mathbf v}$ on the right hand side and ignore the second order functional derivatives when calculating the functional derivatives of $\{F,G\}$ and $\{F,G\}_l$. 
Then we have 
\begin{equation}
\label{eq:xvafgh}
    \begin{aligned}
        \{\{ F, G\}_l,K\}_l & = \int f \left[ \frac{\delta \{F,G\}_0}{\delta f} +  \frac{\delta \{F,G\}}{\delta n_e}, \frac{\delta K}{\delta f}  \right]_{xv} \\
        & + \int f \nabla \times {\mathbf A} \cdot \left( \frac{\partial}{\partial {\mathbf v}} \frac{\delta \{F,G\}_0}{\delta f} \times  \frac{\partial}{\partial {\mathbf v}} \frac{\delta K}{\delta f}   \right) \mathrm{d}{\mathbf x} \mathrm{d}{\mathbf v}\\
        & - \int \frac{1}{n} \nabla \times {\mathbf A} \cdot \left( \frac{\delta \{F,G\}}{\delta {\mathbf A}} \times  \frac{\delta K}{\delta {\mathbf A}}  \right)  \mathrm{d}{\mathbf x}\\
        & - \int \frac{f}{n} \nabla \times {\mathbf A} \cdot \left( \frac{\partial }{\partial {\mathbf v}} \frac{\delta K}{\delta f} \times \frac{\delta \{F,G\}}{\delta {\mathbf A}} -  \frac{\partial }{\partial {\mathbf v}} \frac{\delta \{F,G\}_0}{\delta f} \times \frac{\delta K}{\delta {\mathbf A}}  \right) \mathrm{d}{\mathbf x} \mathrm{d}{\mathbf v}\\
        & - \int \frac{1}{n} \nabla \times {\mathbf A} \cdot \left( \int f \frac{\partial }{\partial {\mathbf v}} \frac{\delta \{F,G\}_0}{\delta f} \mathrm{d}{\mathbf v} \times \int f \frac{\partial }{\partial {\mathbf v}} \frac{\delta K}{\delta f} \mathrm{d}{\mathbf v} 
         \right)  \mathrm{d} {\mathbf x}.
    \end{aligned}
\end{equation}
The sum of the permutation of $T^2_{n_e}$ in~\eqref{eq:xvanefgh} (containing all terms involving the functional derivatives of $n_e$) is zero, and the other terms left in~\eqref{eq:xvanefgh} except $T^1_{n_e}$ are the same as the terms in~\eqref{eq:xvafgh}. Thus, when $$\nabla \cdot \frac{\delta K}{\delta {\mathbf A}} = 0, \quad \forall K, $$  we have $T^1_{n_e} = 0$,  and the bracket~\eqref{eq:XVAbracket} satisfies the Jacobi identity. 
\end{proof}

We can define the following set of functionals satisfying the condition $\nabla \cdot \frac{\delta K}{\delta {\mathbf A}} = 0$, 
\begin{equation}\label{eq:classxva}
\mathcal{B} = \{ F(f({\mathbf x}, \, {\mathbf v}), \, {\mathbf A}) =  \tilde{F}(f({\mathbf x}, \, {\mathbf v}), \, {\mathbf B})| {\mathbf B} = \nabla \times {\mathbf A})
\end{equation}
which includes the functionals determined by all smooth functionals about $f({\mathbf x}, \, {\mathbf v}), \, {\mathbf B}$ with $\nabla \cdot {\mathbf B} = 0$.

Similar to the results of the bracket~\eqref{eq:1}, we also need to answer that if $\{F, G\}_l \in \mathcal{B}$ for $\forall F, G \in \mathcal{B}$, before we can say that the bracket~\eqref{eq:XVAbracket} is a Poisson bracket for the functionals in $\mathcal{B}$. The answer is also 'yes',  which shall be proved in theorem~\ref{thm:relation} after the magnetic field based formulation, i.e., the xvB formulation,  is investigated.

\subsection{xvB formulation}\label{sec:xvB}

Finally, we consider the magnetic field based formulation. 
 When the electron charge density $n_e$ is treated as an independent unknown, there is a Poisson bracket that can be derived  from~\cite{Tronci} as follows
\begin{equation}
\label{eq:xvbnebrackethelp}
    \begin{aligned}
        &\{ F, G\}(f({\mathbf x}, {\mathbf v}), {\mathbf B}, n_e)  = \int f \left[ \frac{\delta F}{\delta f}, \frac{\delta G}{\delta f}  \right]_{xv} \mathrm{d}{\mathbf x} \mathrm{d}{\mathbf v} + \int f  {\mathbf B} \cdot \left( \frac{\partial}{\partial {\mathbf v}} \frac{\delta F}{\delta f} \times  \frac{\partial}{\partial {\mathbf v}} \frac{\delta G}{\delta f}   \right) \mathrm{d}{\mathbf x} \mathrm{d}{\mathbf v}\\
        & - \int \frac{1}{n_e}  {\mathbf B} \cdot \left(\nabla \times  \frac{\delta F}{\delta {\mathbf B}} \times \nabla \times  \frac{\delta G}{\delta {\mathbf B}}  \right)  \mathrm{d}{\mathbf x} - \int \frac{1}{n_e}  {\mathbf B} \cdot \left( \int f \frac{\partial }{\partial {\mathbf v}} \frac{\delta F}{\delta f} \mathrm{d}{\mathbf v} \times \int f \frac{\partial }{\partial {\mathbf v}} \frac{\delta G}{\delta f} \mathrm{d}{\mathbf v} 
         \right)  \mathrm{d} {\mathbf x}\\
        & - \int \frac{f}{n_e}  {\mathbf B} \cdot \left( \frac{\partial }{\partial {\mathbf v}} \frac{\delta G}{\delta f} \times \nabla \times \frac{\delta F}{\delta {\mathbf B}} -  \frac{\partial }{\partial {\mathbf v}} \frac{\delta F}{\delta f} \times \nabla \times \frac{\delta G}{\delta {\mathbf B}}  \right) \mathrm{d}{\mathbf x} \mathrm{d}{\mathbf v}\\
         & + \int f \left(-\frac{\partial}{\partial {\mathbf v}} \frac{\delta F}{\delta f}\cdot\frac{\partial}{\partial {\mathbf x}} \frac{\delta G}{\delta n_e} + \frac{\partial}{\partial {\mathbf v}} \frac{\delta G}{\delta f}\cdot\frac{\partial}{\partial {\mathbf x}} \frac{\delta F}{\delta n_e}  \right) \mathrm{d}{\mathbf x} \mathrm{d}{\mathbf v}.
    \end{aligned}
\end{equation}

\begin{remark}
The quasi-neutrality condition $n_e = \int f\, \mathrm{d}{\mathbf v}$ corresponds a class of Casimir functional of the bracket~\eqref{eq:xvbnebrackethelp}, i.e., for any given ${\mathbf x}_0$,
$$
C(f, n_e, {\mathbf B}) = \int \delta({\mathbf x} - {\mathbf x}_0) \left( n_e - \int f \, \mathrm{d}{\mathbf v} \right) \mathrm{d}{\mathbf x}.
$$
\end{remark}

\noindent{\bf Simplified bracket}
In~\cite{LHPS} we propose the following bracket for the xvB formulation via removing the terms of functional derivatives about $n_e$ and replacing $n_e$ with $n = \int f\, \mathrm{d}{\mathbf v}$,
\begin{equation}
\label{eq:xvBbracket}
    \begin{aligned}
        \{ F, G\}_l(f({\mathbf x}, {\mathbf v}), {\mathbf B}) & = \int f \left[ \frac{\delta F}{\delta f}, \frac{\delta G}{\delta f}  \right]_{xv} \mathrm{d}{\mathbf x} \mathrm{d}{\mathbf v} + \int f  {\mathbf B} \cdot \left( \frac{\partial}{\partial {\mathbf v}} \frac{\delta F}{\delta f} \times  \frac{\partial}{\partial {\mathbf v}} \frac{\delta G}{\delta f}   \right) \mathrm{d}{\mathbf x} \mathrm{d}{\mathbf v}\\
        & - \int \frac{1}{n}  {\mathbf B} \cdot \left(\nabla \times  \frac{\delta F}{\delta {\mathbf B}} \times \nabla \times  \frac{\delta G}{\delta {\mathbf B}}  \right)  \mathrm{d}{\mathbf x}\\
        & - \int \frac{f}{n}  {\mathbf B} \cdot \left( \frac{\partial }{\partial {\mathbf v}} \frac{\delta G}{\delta f} \times \nabla \times \frac{\delta F}{\delta {\mathbf B}} -  \frac{\partial }{\partial {\mathbf v}} \frac{\delta F}{\delta f} \times \nabla \times \frac{\delta G}{\delta {\mathbf B}}  \right) \mathrm{d}{\mathbf x} \mathrm{d}{\mathbf v}\\
        & - \int \frac{1}{n}  {\mathbf B} \cdot \left( \int f \frac{\partial }{\partial {\mathbf v}} \frac{\delta F}{\delta f} \mathrm{d}{\mathbf v} \times \int f \frac{\partial }{\partial {\mathbf v}} \frac{\delta G}{\delta f} \mathrm{d}{\mathbf v} 
         \right)  \mathrm{d} {\mathbf x}.
    \end{aligned}
\end{equation}

\begin{remark}\label{rm:darboux}
    As $n_e = \int f \,\mathrm{d}{\mathbf v}$ corresponds to the Casimir functional of the Poisson bracket~\eqref{eq:xvbnebrackethelp}, we choose the Casimir functional as an unknown to simplify the Poisson bracket as in the  Darboux' theorem. Specifically, we conduct the transformation of unknowns as 
    $$
    \left(f({\mathbf x}, {\mathbf v}), \ {\mathbf B}, \ n_e \right) \rightarrow \left( f({\mathbf x}, {\mathbf v}), \ {\mathbf B}, \ C = n_e - \int f \, \mathrm{d}{\mathbf v} \right).
    $$
    Then any smooth functional $\tilde{F}$ depending on $f, {\mathbf B}, C$ can be considered as a functional $F$ depending on $f, {\mathbf B}, n_e$, and we have the following relations between the functional derivatives,
    $$
    \frac{\delta F}{\delta f} = \frac{\delta \tilde{F}}{\delta f} - \frac{\delta \tilde{F}}{\delta C}, \quad \frac{\delta F}{\delta {\mathbf A}} = \frac{\delta \tilde{F}}{\delta {\mathbf A}}, \quad \frac{\delta F}{\delta n_e} = \frac{\delta \tilde{F}}{\delta C}.
    $$
    By plugging the above relations of the functional derivatives into~\eqref{eq:xvbnebrackethelp}, we get (we omit the tilde for convenience),
    \begin{equation*}
    \begin{aligned}
        &\{ F, G\}(f({\mathbf x}, {\mathbf v}), {\mathbf B}, C)  = \int f \left[ \frac{\delta F}{\delta f}, \frac{\delta G}{\delta f}  \right]_{xv}\mathrm{d}{\mathbf x} \mathrm{d}{\mathbf v} - \int f\left( \frac{\partial}{\partial {\mathbf x}} \frac{\delta F}{\delta C} \cdot \frac{\partial}{\partial {\mathbf v}} \frac{\delta G}{\delta f}\right) \mathrm{d}{\mathbf x} \mathrm{d}{\mathbf v} \\
        & + \int f\left( \frac{\partial}{\partial {\mathbf x}} \frac{\delta G}{\delta C} \cdot \frac{\partial}{\partial {\mathbf v}} \frac{\delta F}{\delta f}\right) \mathrm{d}{\mathbf x} \mathrm{d}{\mathbf v}  + \int f  {\mathbf B} \cdot \left( \frac{\partial}{\partial {\mathbf v}} \frac{\delta F}{\delta f} \times  \frac{\partial}{\partial {\mathbf v}} \frac{\delta G}{\delta f}   \right) \mathrm{d}{\mathbf x} \mathrm{d}{\mathbf v}\\
        & - \int \frac{1}{C + \int f \, \mathrm{d}{\mathbf v}}  {\mathbf B} \cdot \left(\nabla \times  \frac{\delta F}{\delta {\mathbf B}} \times \nabla \times  \frac{\delta G}{\delta {\mathbf B}}  \right)  \mathrm{d}{\mathbf x} \\
        & - \int \frac{1}{C + \int f \, \mathrm{d}{\mathbf v}}  {\mathbf B} \cdot \left( \int f \frac{\partial }{\partial {\mathbf v}} \frac{\delta F}{\delta f} \mathrm{d}{\mathbf v} \times \int f \frac{\partial }{\partial {\mathbf v}} \frac{\delta G}{\delta f} \mathrm{d}{\mathbf v} 
         \right)  \mathrm{d} {\mathbf x}\\
        & - \int \frac{f}{C + \int f \, \mathrm{d}{\mathbf v}}  {\mathbf B} \cdot \left( \frac{\partial }{\partial {\mathbf v}} \frac{\delta G}{\delta f} \times \nabla \times \frac{\delta F}{\delta {\mathbf B}} -  \frac{\partial }{\partial {\mathbf v}} \frac{\delta F}{\delta f} \times \nabla \times \frac{\delta G}{\delta {\mathbf B}}  \right) \mathrm{d}{\mathbf x} \mathrm{d}{\mathbf v}\\
         & + \int f \left(-\frac{\partial}{\partial {\mathbf v}} \frac{\delta F}{\delta f}\cdot\frac{\partial}{\partial {\mathbf x}} \frac{\delta G}{\delta C} + \frac{\partial}{\partial {\mathbf v}} \frac{\delta G}{\delta f}\cdot\frac{\partial}{\partial {\mathbf x}} \frac{\delta F}{\delta C}  \right) \mathrm{d}{\mathbf x} \mathrm{d}{\mathbf v},
    \end{aligned}
\end{equation*}
where the sum of the second and third terms cancel with the last term. In the final result there is no functional derivative about $C$ and $C$ can be considered as a parameter. Note that when $C = 0$ the result is just the bracket~\eqref{eq:xvBbracket}.
    
\end{remark}

Next we give a direct proof of the Jacobi identity of the bracket~\eqref{eq:xvBbracket} under the condition that $\nabla \cdot {\mathbf B} = 0$.

\begin{theorem}\label{theorem:bracket}
    The bracket~\eqref{eq:xvBbracket} satisfies the Jacobi identity under the condition that $\nabla \cdot {\mathbf B} = 0$.
\end{theorem}
\begin{proof}
    By the bracket theorem in~\cite{lifting}, to prove the Jacobi identity
for bracket $$\{F, G\} = \int \frac{\delta F}{\delta \psi} J(\psi) \frac{\delta G}{\delta \psi}\, \mathrm{d}{\mathbf a},$$ where $J$ is a operator depending on $\psi$ and $\psi$ depends on ${\mathbf a}$, only the explicit dependence of $J$ on $\psi$ is needed to consider when
taking the functional derivative $\frac{\delta \{F, G \}}{\delta \psi}$. Then we have the following derivatives of the bracket in~\eqref{eq:xvBbracket}, where we remove the index $l$ of the bracket~\eqref{eq:xvBbracket} for simplicity,
\begin{equation*}
    \begin{aligned}
        \frac{\delta \{F, G \}}{\delta f} &= \left[\frac{\delta F}{\delta f}, \frac{\delta G}{\delta f}  
        \right]_{xv} +  {\mathbf B} \cdot  \left( \frac{\partial}{\partial {\mathbf v}} \frac{\delta F}{\delta f} \times  \frac{\partial}{\partial {\mathbf v}} \frac{\delta G}{\delta f}   \right) + \frac{1}{n^2} {\mathbf B} \cdot \left( \nabla \times \frac{\delta F}{\delta {\mathbf B}} \times  \nabla \times \frac{\delta G}{\delta {\mathbf B}}  \right) \\
         & + \int \frac{f}{n^2}  {\mathbf B} \cdot \left( \frac{\partial }{\partial {\mathbf v}} \frac{\delta G}{\delta f} \times \nabla \times \frac{\delta F}{\delta {\mathbf B}} -  \frac{\partial }{\partial {\mathbf v}} \frac{\delta F}{\delta f} \times \nabla \times \frac{\delta G}{\delta {\mathbf B}}  \right) \mathrm{d}{\mathbf v}\\
        & -  \frac{1}{n} {\mathbf B} \cdot \left( \frac{\partial }{\partial {\mathbf v}} \frac{\delta G}{\delta f} \times \nabla \times \frac{\delta F}{\delta {\mathbf B}} -  \frac{\partial }{\partial {\mathbf v}} \frac{\delta F}{\delta f} \times \nabla \times  \frac{\delta G}{\delta {\mathbf B}}  \right)\\
         & + \frac{1}{n^2} {\mathbf B} \cdot \left( \int f \frac{\partial }{\partial {\mathbf v}} \frac{\delta F}{\delta f} \mathrm{d}{\mathbf v} \times \int f \frac{\partial }{\partial {\mathbf v}} \frac{\delta G}{\delta f} \mathrm{d}{\mathbf v} \right)\\
        & -  \frac{1}{n} {\mathbf B} \cdot \left(\frac{\partial }{\partial {\mathbf v}} \frac{\delta F}{\delta f}\times \int f \frac{\partial }{\partial {\mathbf v}} \frac{\delta G}{\delta f} \mathrm{d}{\mathbf v} \right)  -  \frac{1}{n} {\mathbf B} \cdot \left( \int f \frac{\partial }{\partial {\mathbf v}} \frac{\delta F}{\delta f} \mathrm{d}{\mathbf v} \times \frac{\partial }{\partial {\mathbf v}} \frac{\delta G}{\delta f}  \right),\\
         \frac{\delta \{F, G \}}{\delta {\mathbf B}} &=  \int f \left( \frac{\partial}{\partial {\mathbf v}} \frac{\delta F}{\delta f} \times  \frac{\partial}{\partial {\mathbf v}} \frac{\delta G}{\delta f}   \right) \mathrm{d}{\mathbf v}  - \left( \frac{1}{n} \nabla \times \frac{\delta F}{\delta {\mathbf B}} \times \nabla \times \frac{\delta G}{\delta {\mathbf B}} \right) \\
        & -  \int \frac{f}{n} \left( \frac{\partial}{\partial {\mathbf v}} \frac{\delta G}{\delta f} \times \nabla \times \frac{\delta F}{\delta {\mathbf B}} -  \frac{\partial}{\partial {\mathbf v}} \frac{\delta F}{\delta f} \times \nabla \times \frac{\delta G}{\delta {\mathbf B}}  \right) \mathrm{d}{\mathbf v}\\
        & -\frac{1}{n} \int f \frac{\partial}{\partial {\mathbf v}} \frac{\delta F}{\delta f} \mathrm{d}{\mathbf v} \times \int f \frac{\partial}{\partial {\mathbf v}} \frac{\delta G}{\delta f} \mathrm{d}{\mathbf v}.
    \end{aligned}
\end{equation*}
Here we write out $\{\{F, G\}, K\}$ as 
\begin{equation*}
    \begin{aligned}
        \{\{F, G\}, K\} = \{\{F, G\}, K\}_{fff} + \{\{F, G\}, K\}_{ffB} + \{\{F, G\}, K\}_{fBB} + \{\{F, G\}, K\}_{BBB},
    \end{aligned}
\end{equation*}
where $\{\{F, G\}, K\}_{abc}$ denotes the terms including the functional derivatives about $a,b,c$. Next we write out the further expression of each group of terms.  
\begin{small}
\begin{equation}
\label{eq:each}
    \begin{aligned}
        \{\{F, G\}, K\}_{fff} & = \{\{F, G\}, K\}_{fff}^{f} + \{\{F, G\}, K\}_{fff}^{fB} + \{\{F, G\}, K\}_{fff}^{ffB/n} + \{\{F, G\}, K\}_{fff}^{fffB/n^2}\\ 
        & + 
        \{\{F, G\}, K\}_{fff}^{fB^2}
         +
        \{\{F, G\}, K\}_{fff}^{ffB^2/n} + 
        \{\{F, G\}, K\}_{fff}^{fffB^2/n^2},\\
        \{\{F, G\}, K\}_{ffB} & = \{\{F, G\}, K\}_{ffB}^{fB/n} + \{\{F, G\}, H\}_{ffB}^{ffB/n^2} + \{\{F, G\}, K\}_{ffB}^{fB^2/n} + \{\{F, G\}, K\}_{ffB}^{ffB^2/n^2},\\
        \{\{F, G\}, K\}_{fBB}& = \{\{F, G\}, K\}_{fBB}^{fB/n^2} + \{\{F, G\}, K\}_{fBB}^{fB^2/n^2},\\
        \{\{F, G\}, K\}_{BBB} & =  \{\{F, G\}, K\}_{BBB}^{B/n^2},
    \end{aligned}
\end{equation}
    \end{small}
where the upper indices denote the functions involved except the functional derivatives.

The proof of the sum of the permutation of the above each term is zero can be found in the Appendix~\ref{sec:appendix}, some of which hold under the condition that $\nabla \cdot {\mathbf B} = 0$. Then we know that the Jacobi identity holds, i.e., 
$$
\{\{F, G\}, K\} + \{\{G, K\}, F\} + \{\{K, F\}, G\} = 0,
$$
under the condition that $\nabla \cdot {\mathbf B} = 0$.

\end{proof}

\subsection{The relations between the brackets}\label{sec:relation}
Here we show the relations of the above three brackets, the same approach is used for deriving the magnetic field based Poisson brackets from the canonical momentum based Poisson brackets in~\cite{Tronci, DP, MW}.\\ 
1)
The bracket~\eqref{eq:xvBbracket} can be derived from the bracket~\eqref{eq:1} through the following transformation and relation of the functional derivatives,
\begin{equation}\label{eq:relaxpaxvb}
\begin{aligned}
        &f({\mathbf x}, {\mathbf p}) = \tilde{f}({\mathbf x}, {\mathbf v}),\quad  {\mathbf p} = {\mathbf v} + {\mathbf A}, \quad \nabla \times {\mathbf A} = {\mathbf B}, \\
&\left[f({\mathbf x}, {\mathbf p}),  g({\mathbf x}, {\mathbf p})\right]_{xp} = \left[\tilde{f}({\mathbf x}, {\mathbf v}),  \tilde{g}({\mathbf x}, {\mathbf v})\right]_{xv} + \nabla \times {\mathbf A} \cdot \frac{\partial }{\partial {\mathbf v}} \tilde{f}({\mathbf x}, {\mathbf v}) \times \frac{\partial }{\partial {\mathbf v}}\tilde{g}({\mathbf x}, {\mathbf v}),\\
&\frac{\delta F}{\delta f} = \frac{\delta \tilde{F}}{\delta \tilde{f}}, \quad \frac{\delta F}{\delta {\mathbf A}} = \nabla \times \frac{\delta \tilde{F}}{\delta {\mathbf B}} - \int f  \frac{\partial }{\partial {\mathbf v}} \frac{\delta \tilde{F}}{\delta \tilde{f}} \,\mathrm{d}{\mathbf v}.
\end{aligned}
\end{equation}
2)
The bracket~\eqref{eq:XVAbracket} can be derived from the bracket~\eqref{eq:1} through the following transformation and relation of the functional derivatives,
\begin{equation*}
\begin{aligned}
        &f({\mathbf x}, {\mathbf p}) = \tilde{f}({\mathbf x}, {\mathbf v}),\quad  {\mathbf p} = {\mathbf v} + {\mathbf A}, \\
&\left[f({\mathbf x}, {\mathbf p}),  g({\mathbf x}, {\mathbf p})\right]_{xp} = \left[\tilde{f}({\mathbf x}, {\mathbf v}),  \tilde{g}({\mathbf x}, {\mathbf v})\right]_{xv} + \nabla \times {\mathbf A} \cdot \frac{\partial }{\partial {\mathbf v}} \tilde{f}({\mathbf x}, {\mathbf v}) \times \frac{\partial }{\partial {\mathbf v}}\tilde{g}({\mathbf x}, {\mathbf v}),\\
&\frac{\delta F}{\delta f} = \frac{\delta \tilde{F}}{\delta \tilde{f}}, \quad \frac{\delta F}{\delta {\mathbf A}} = \frac{\delta \tilde{F}}{\delta {\mathbf A}} - \int f  \frac{\partial }{\partial {\mathbf v}} \frac{\delta \tilde{F}}{\delta \tilde{f}} \,\mathrm{d}{\mathbf v}.
\end{aligned}
\end{equation*}
3)
The bracket~\eqref{eq:xvBbracket} can be derived from the bracket~\eqref{eq:XVAbracket} through the following transformation and relation of the functional derivatives,
$$
       \nabla \times {\mathbf A} = {\mathbf B}, \quad
\frac{\delta F}{\delta {\mathbf A}} = \nabla \times \frac{\delta \tilde{F}}{\delta {\mathbf B}}.
$$

\begin{theorem}\label{thm:relation}
  The bracket~\eqref{eq:1} of two arbitrary smooth functionals in $\mathcal{A}$~\eqref{eq:classxpa} is also in $\mathcal{A}$~\eqref{eq:classxpa}. 
    The bracket~\eqref{eq:XVAbracket} of arbitrary smooth functionals in $\mathcal{B}$~\eqref{eq:classxva} is also in $\mathcal{B}$~\eqref{eq:classxva}.
\end{theorem}
\begin{proof}
    For two arbitrary smooth functionals $F$ and $G$ in~\eqref{eq:classxpa} depending on $f({\mathbf x}, {\mathbf p})$ and ${\mathbf A
    }$, we have 
    $$
F(f({\mathbf x}, {\mathbf p}), {\mathbf A
    }) = \tilde{F}(\tilde{f}({\mathbf x}, {\mathbf v}), {\mathbf B}), \quad G(f({\mathbf x}, {\mathbf p}), {\mathbf A
    }) = \tilde{G}(\tilde{f}({\mathbf x}, {\mathbf v}), {\mathbf B}),
    $$
    where $\tilde{F}$ and $\tilde{G}$ are two functionals determining $F$ and $G$ through the definition in~\eqref{eq:classxpa}. 
    Then we have \begin{equation}\label{eq:fgxpa}
\{F, G \}(f({\mathbf x}, {\mathbf p}), {\mathbf A
    }) = \{\tilde{F}, \tilde{G}\}(\tilde{f}({\mathbf x}, {\mathbf v}), {\mathbf B}), 
    \end{equation}
    where the bracket on the left hand side is the bracket~\eqref{eq:1} of the xpA formulation, and the bracket on the right hand side is the bracket~\eqref{eq:xvBbracket} of the xvB formulation. 
    In~\eqref{eq:fgxpa},  the $'='$ holds as the bracket~\eqref{eq:xvBbracket} can be derived from the bracket~\eqref{eq:1} via the  transformation with the chain rules of the functional derivatives~\eqref{eq:relaxpaxvb}. 
    Note that equality~\eqref{eq:fgxpa} tells us that the functional $\{F, G \}(f({\mathbf x}, {\mathbf p}), {\mathbf A})$ is determined by the functional $\{\tilde{F}, \tilde{G} \}(\tilde{f}({\mathbf x}, {\mathbf v}), {\mathbf B})$ and is thus in the functional set defined in $\mathcal{A}$~\eqref{eq:classxpa}, also it satisfies the condition~\eqref{eq:conditionxpa}.

    By the similar arguments, we know that the bracket~\eqref{eq:XVAbracket} of arbitrary smooth functionals in $\mathcal{B}$~\eqref{eq:classxva} is also in $\mathcal{B}$~\eqref{eq:classxva}.
\end{proof}

\section{Generalization}\label{sec:general}
In this section, we generalize the above results to more general cases and models in plasma physics, including the HKM model with electron entropy, kinetic Hall MHD hybrid model~\cite{DP}, Hall MHD~\cite{holm}, and kinetic-multi-fluid model~\cite{Tronci}. The two fluid model has been considered in~\cite{burby}.

\subsection{The HKM model with electron entropy}\label{sec:entropy}
Here we present the results of the general case of the HKM model with entropy. 
We add a new term related with the entropy per unit mass, i.e., $s$,  to the bracket proposed in~\cite{Tronci} and get
\begin{equation}
\label{eq:xpas}
    \begin{aligned}
 & \{F, G  \}(f({\mathbf x}, {\mathbf p}), {\mathbf A}, n_e, s)  =  \int f \left[ \frac{\delta F}{\delta f}, \frac{\delta G}{\delta f} \right]_{xp}  \mathrm{d}{\mathbf x}\mathrm{d}{\mathbf p} - \int \frac{1}{n_e} \nabla \times {\mathbf A} \cdot \left( \frac{\delta F}{\delta {\mathbf A}} \times  \frac{\delta G}{\delta {\mathbf A}} \right)  \mathrm{d}{\mathbf x} \\
 & + \int \frac{\delta F}{\delta {\mathbf A}} \cdot \frac{\partial }{\partial {\mathbf x}} \frac{\delta G}{\delta n_e} -  \frac{\delta G}{\delta {\mathbf A}} \cdot \frac{\partial }{\partial {\mathbf x}}\frac{\delta F}{\delta n_e}\, \mathrm{d}{\mathbf x} + \int \frac{\nabla s}{n_e} \cdot \left(\frac{\delta F}{\delta s}\frac{\delta G}{\delta {\mathbf A}} - \frac{\delta G}{\delta s}\frac{\delta F}{\delta {\mathbf A}} \right) \mathrm{d}{\mathbf x},
        \end{aligned}
\end{equation}
with the following Hamiltonian
$$
H(f({\mathbf x}, {\mathbf p}), {\mathbf A}, n_e, s)= \frac{1}{2} \int |{\mathbf p} - {\mathbf A}|^2 f\, \mathrm{d}{\mathbf p}\mathrm{d}{\mathbf x} + \frac{1}{2} \int |\nabla \times {\mathbf A}|^2 \, \mathrm{d}{\mathbf x} +  \int n_e\, \mathcal{U}(n_e, s) \, \mathrm{d}{\mathbf x}.
$$
The Jacobi identity can be shown satisfied by this bracket via a direct proof.
Similarly as the isothermal electron case in subsection~\ref{sec:xpa}, we get the following simplified bracket satisfying the Jacobi identity under condition~\eqref{eq:conditionxpa} and corresponding Hamiltonian
\begin{equation}
\label{eq:xpass}
    \begin{aligned}
 & \{F, G  \}_l(f({\mathbf x}, {\mathbf p}), {\mathbf A}, s)  =  \int f \left[ \frac{\delta F}{\delta f}, \frac{\delta G}{\delta f} \right]_{xp}  \mathrm{d}{\mathbf x}\mathrm{d}{\mathbf p} - \int \frac{1}{n} \nabla \times {\mathbf A} \cdot \left( \frac{\delta F}{\delta {\mathbf A}} \times  \frac{\delta G}{\delta {\mathbf A}} \right)  \mathrm{d}{\mathbf x} \\
 & + \int \frac{\nabla s}{n} \cdot \left(\frac{\delta F}{\delta s}\frac{\delta G}{\delta {\mathbf A}} - \frac{\delta G}{\delta s}\frac{\delta F}{\delta {\mathbf A}} \right) \mathrm{d}{\mathbf x}, \quad n = \int f \, \mathrm{d}{\mathbf p},\\
&H(f({\mathbf x}, {\mathbf p}), {\mathbf A}, s)= \frac{1}{2} \int |{\mathbf p} - {\mathbf A}|^2 f\, \mathrm{d}{\mathbf p}\mathrm{d}{\mathbf x} + \frac{1}{2} \int |\nabla \times {\mathbf A}|^2 \, \mathrm{d}{\mathbf x} +  \int n\, \mathcal{U}(n, s) \, \mathrm{d}{\mathbf x}.
 \end{aligned}
\end{equation}

Next, we present the results of the magnetic field based formulations. We add two terms related the entropy into the bracket~\eqref{eq:xvbnebrackethelp} and get 
\begin{equation}
\label{eq:xvbnebrackethelps}
    \begin{aligned}
        &\{ F, G\}(f({\mathbf x}, {\mathbf v}), {\mathbf B}, n_e, s)  = \int f \left[ \frac{\delta F}{\delta f}, \frac{\delta G}{\delta f}  \right]_{xv} \mathrm{d}{\mathbf x} \mathrm{d}{\mathbf v} + \int f  {\mathbf B} \cdot \left( \frac{\partial}{\partial {\mathbf v}} \frac{\delta F}{\delta f} \times  \frac{\partial}{\partial {\mathbf v}} \frac{\delta G}{\delta f}   \right) \mathrm{d}{\mathbf x} \mathrm{d}{\mathbf v}\\
        & - \int \frac{1}{n_e}  {\mathbf B} \cdot \left(\nabla \times  \frac{\delta F}{\delta {\mathbf B}} \times \nabla \times  \frac{\delta G}{\delta {\mathbf B}}  \right)  \mathrm{d}{\mathbf x} - \int \frac{1}{n_e}  {\mathbf B} \cdot \left( \int f \frac{\partial }{\partial {\mathbf v}} \frac{\delta F}{\delta f} \mathrm{d}{\mathbf v} \times \int f \frac{\partial }{\partial {\mathbf v}} \frac{\delta G}{\delta f} \mathrm{d}{\mathbf v} 
         \right)  \mathrm{d} {\mathbf x}\\
        & - \int \frac{f}{n_e}  {\mathbf B} \cdot \left( \frac{\partial }{\partial {\mathbf v}} \frac{\delta G}{\delta f} \times \nabla \times \frac{\delta F}{\delta {\mathbf B}} -  \frac{\partial }{\partial {\mathbf v}} \frac{\delta F}{\delta f} \times \nabla \times \frac{\delta G}{\delta {\mathbf B}}  \right) \mathrm{d}{\mathbf x} \mathrm{d}{\mathbf v}\\
         & + \int f \left(-\frac{\partial}{\partial {\mathbf v}} \frac{\delta F}{\delta f}\cdot\frac{\partial}{\partial {\mathbf x}} \frac{\delta G}{\delta n_e} + \frac{\partial}{\partial {\mathbf v}} \frac{\delta G}{\delta f}\cdot\frac{\partial}{\partial {\mathbf x}} \frac{\delta F}{\delta n_e}  \right) \mathrm{d}{\mathbf x} \mathrm{d}{\mathbf v}\\
          & + \int \frac{\nabla s}{n_e} \cdot \left(\frac{\delta F}{\delta s} \nabla \times \frac{\delta G}{\delta {\mathbf B}} - \frac{\delta G}{\delta s} \nabla \times \frac{\delta F}{\delta {\mathbf B}} \right) \mathrm{d}{\mathbf x} - \int f\,\frac{\nabla s}{n_e} \cdot \left(\frac{\delta F}{\delta s} \frac{\partial}{\partial {\mathbf v}}\frac{\delta G}{\delta f} - \frac{\delta G}{\delta s} \frac{\partial}{\partial {\mathbf v}} \frac{\delta F}{\delta f} \right) \mathrm{d}{\mathbf x}\mathrm{d}{\mathbf v},\\
&H(f({\mathbf x}, {\mathbf v}), {\mathbf B}, n_e, s)= \frac{1}{2} \int |{\mathbf v}|^2 f\, \mathrm{d}{\mathbf v}\mathrm{d}{\mathbf x} + \frac{1}{2} \int |{\mathbf B}|^2 \, \mathrm{d}{\mathbf x} +  \int n_e\, \mathcal{U}(n_e, s) \, \mathrm{d}{\mathbf x}.
 \end{aligned}
\end{equation}
Similar to subsection~\ref{sec:xvB}, we get the following simplified Poisson bracket and corresponding Hamiltonian
\begin{equation}
\label{eq:xvBbrackets}
    \begin{aligned}
        & \{ F, G\}_l(f({\mathbf x}, {\mathbf v}), {\mathbf B}, s)  = \int f \left[ \frac{\delta F}{\delta f}, \frac{\delta G}{\delta f}  \right]_{xv} \mathrm{d}{\mathbf x} \mathrm{d}{\mathbf v} + \int f  {\mathbf B} \cdot \left( \frac{\partial}{\partial {\mathbf v}} \frac{\delta F}{\delta f} \times  \frac{\partial}{\partial {\mathbf v}} \frac{\delta G}{\delta f}   \right) \mathrm{d}{\mathbf x} \mathrm{d}{\mathbf v}\\
        & - \int \frac{1}{n}  {\mathbf B} \cdot \left(\nabla \times  \frac{\delta F}{\delta {\mathbf B}} \times \nabla \times  \frac{\delta G}{\delta {\mathbf B}}  \right)  \mathrm{d}{\mathbf x} - \int \frac{f}{n}  {\mathbf B} \cdot \left( \frac{\partial }{\partial {\mathbf v}} \frac{\delta G}{\delta f} \times \nabla \times \frac{\delta F}{\delta {\mathbf B}} -  \frac{\partial }{\partial {\mathbf v}} \frac{\delta F}{\delta f} \times \nabla \times \frac{\delta G}{\delta {\mathbf B}}  \right) \mathrm{d}{\mathbf x} \mathrm{d}{\mathbf v}\\
        & - \int \frac{1}{n}  {\mathbf B} \cdot \left( \int f \frac{\partial }{\partial {\mathbf v}} \frac{\delta F}{\delta f} \mathrm{d}{\mathbf v} \times \int f \frac{\partial }{\partial {\mathbf v}} \frac{\delta G}{\delta f} \mathrm{d}{\mathbf v} 
         \right)  \mathrm{d} {\mathbf x} + \int \frac{\nabla s}{n} \cdot \left(\frac{\delta F}{\delta s} \nabla \times \frac{\delta G}{\delta {\mathbf B}} - \frac{\delta G}{\delta s} \nabla \times \frac{\delta F}{\delta {\mathbf B}} \right) \mathrm{d}{\mathbf x}\\
         & - \int f\,\frac{\nabla s}{n} \cdot \left(\frac{\delta F}{\delta s} \frac{\partial}{\partial {\mathbf v}}\frac{\delta G}{\delta f} - \frac{\delta G}{\delta s} \frac{\partial}{\partial {\mathbf v}} \frac{\delta F}{\delta f} \right) \mathrm{d}{\mathbf x}\mathrm{d}{\mathbf v},\quad n = \int f \, \mathrm{d}{\mathbf v},\\
&H(f({\mathbf x}, {\mathbf v}), {\mathbf B}, s)= \frac{1}{2} \int |{\mathbf v}|^2 f\, \mathrm{d}{\mathbf v}\mathrm{d}{\mathbf x} + \frac{1}{2} \int |{\mathbf B}|^2 \, \mathrm{d}{\mathbf x} +  \int n\, \mathcal{U}(n, s) \, \mathrm{d}{\mathbf x}.
\end{aligned}
\end{equation}

The Poisson brackets expressed with the entropy density $\sigma= n_e s$ or the electron pressure $p_e = n_e^2 \frac{\partial U}{\partial n_e}$ can be obtained from the above brackets via the chain rules of the functional derivatives. Similar simplifications can be done as above via the quasi-neutrality condition. Also similar numerical discretizations can be conducted based on the above brackets as~\cite{LHPS,LHPS2}.

\subsection{Kinetic Hall MHD hybrid model}
Then we consider a more general hybrid model~\cite{DP} with only part of the ions described by kinetic equations. Here we consider the single ion (with unit charge) species case for simplicity,  
\begin{equation}\label{eq:hybridequations}
\begin{aligned}
&\frac{\partial f}{\partial t} = - {\mathbf v} \cdot \nabla f - ({\mathbf E} + {\mathbf v} \times {\mathbf B}) \cdot \nabla_v f,\\
& \frac{\partial {\mathbf B}}{\partial t} = - \nabla \times {\mathbf E},\\
& \frac{\partial n_e}{\partial t} = - \nabla \cdot (n_i{\mathbf V}) - \nabla \cdot {\mathbf J}_k,\\
& \frac{\partial n_i}{\partial t} = - \nabla \cdot (n_i{\mathbf V}),\\
& n_i (\frac{\partial {\mathbf V}}{\partial t} + {\mathbf V}\cdot \nabla {\mathbf V}) = \left[ \frac{n_i}{n_e}\sigma_k{\mathbf V} + \frac{n_i}{n_e}(\nabla \times {\mathbf B} - {\mathbf J}_k)  \right] \times {\mathbf B} - \frac{n_i}{n_e} \nabla \cdot {\mathbf P}_e - \nabla P_i
\end{aligned}
\end{equation}
where ${\mathbf J} = \nabla \times {\mathbf B}, {\mathbf J}_k = \int {\mathbf v} f \mathrm{d}{\mathbf v}$, 
\begin{equation}\label{eq:ohmslaw}
 {\mathbf E} = -\frac{n_i}{n_e}{\mathbf V} \times {\mathbf B}  + \frac{1}{n_e}(\nabla \times {\mathbf B} - {\mathbf J}_k) \times {\mathbf B} - \frac{\nabla \cdot {\mathbf P}_e}{n_e},\quad {\mathbf P}_e = \frac{P_{e\parallel}-P_{\perp}}{B^2}{\mathbf B}{\mathbf B} + P_{e\perp}{\mathbf I},
\end{equation}
$$
P_{e\parallel} = {\mathbf P}_e : {\mathbf b}{\mathbf b}, \quad P_{e\perp} = \frac{1}{2} {\mathbf P}_e : ({\mathbf I} - {\mathbf b}{\mathbf b}), \quad {\mathbf b} := {\mathbf B}/|{\mathbf B}| 
$$
Here $f({\mathbf x}, {\mathbf v}, t)$ represents the distribution function of the ions depending on time $t$, space ${\mathbf x}$, and velocity ${\mathbf v}$. ${\mathbf E}$ and ${\mathbf B}$ are electric and magnetic fields, respectively. The kinetic ion charge density is $\int f \,\mathrm{d}{\mathbf v}$, the fluid ion charge density is $n_i$, the electron charge density is $n_e$, which satisfy $n_e = n_i + \int f \mathrm{d}{\mathbf v}$ because of the quasi-neutrality condition. ${\mathbf V}$ is the ion fluid velocity field. The electron internal energy has the form 
$
\mathcal{U}_e = \mathcal{U}_e(n_e, |{\mathbf B}|),
$
which gives 
$$
\frac{\partial \mathcal{U}_e}{\partial n_e} = \frac{P_{e\parallel}}{n_e^2}, \quad \frac{\partial \mathcal{U}_e}{\partial |{\mathbf B}|} = \frac{P_{e\perp} - P_{e\parallel}}{n_e|{\mathbf B}|}, \quad \frac{\partial \mathcal{U}_e}{\partial{\mathbf B}} = -\gamma{\mathbf B}, \quad \gamma = \frac{P_{e\parallel} - P_{e\perp}}{n_e|{\mathbf B}|^2}.
$$

In~\cite{DP}, various Poisson brackets are proposed for the different equivalent formulations of the above kinetic Hall MHD hybrid model~\eqref{eq:hybridequations}.
Similar to the HKM model, for the kinetic Hall MHD hybrid model~\eqref{eq:hybridequations}, simplified anti-symmetric brackets can be obtained by neglecting the terms involving the functional derivative with respect to the electron charge density $n_e$, while still preserving the equivalence of the derived equations with the corresponding Hamiltonians. In the following we determine the condition for the Jacobi identity of the simplified 'canonical bracket', and remark the magnetic field based brackets in remark~\ref{rm:DPmagnetic}.

\noindent{\bf Canonical bracket}
In~\cite{DP} the following canonical bracket with only 5 terms is proposed for the kinetic Hall MHD hybrid model~\eqref{eq:hybridequations}. Here the 'canonical' means canonical fluid momentum and canonical kinetic momentum.
\begin{equation}
\label{eq:fluidxvpne}
\begin{aligned}
& \{F, G\}(n_i, \bar{\mathbf M}, {\mathbf A}, f({\mathbf x}, {\mathbf p}), n_e) =  \int f \left[ \frac{\delta F}{\delta f}, \frac{\delta G}{\delta f} \right]_{xp}  \mathrm{d}{\mathbf x}\mathrm{d}{\mathbf p}\\
& + \int \bar{\mathbf M} \cdot \left(\frac{\delta G}{\delta \bar{\mathbf M}} \cdot \nabla \frac{\delta F}{\delta \bar{\mathbf M}} - \frac{\delta F}{\delta \bar{\mathbf M}} \cdot \nabla \frac{\delta G}{\delta \bar{\mathbf M}} \right) \mathrm{d}{\mathbf x} + \int n_i \left(\frac{\delta G}{\delta \bar{\mathbf M}} \cdot \nabla \frac{\delta F}{\delta n_i} - \frac{\delta F}{\delta \bar{\mathbf M}} \cdot \nabla \frac{\delta G}{\delta n_i}  \right) \mathrm{d}{\mathbf x}\\
& - \int \left( \frac{\delta G}{\delta {\mathbf A}} \cdot \nabla \frac{\delta F}{\delta n_e} - \frac{\delta F}{\delta {\mathbf A}} \cdot \nabla \frac{\delta G}{\delta n_e} \right)  \mathrm{d}{\mathbf x} - \int \frac{1}{n_e} \nabla \times {\mathbf A} \cdot \left( \frac{\delta F}{\delta {\mathbf A}}  \times \frac{\delta G}{\delta {\mathbf A}}  \right)  \mathrm{d}{\mathbf x}.
\end{aligned}
\end{equation}
We define the following transformation of the unknowns,
$$
\left(f, n_i, \bar{\mathbf M}, {\mathbf A}, n_e\right) \rightarrow \left(f, n_i, \bar{\mathbf M}, {\mathbf A}, C = n_e - n_i - \int f\, \mathrm{d}{\mathbf p}\right).
$$
Then
any functional $\tilde{F}$ about 
$\left(f, n_i, \bar{\mathbf M}, {\mathbf A}, C = n_e - n_i - \int f\, \mathrm{d}{\mathbf p}\right)$ can be regarded as a functional $F$ about $\left(f, n_i, \bar{\mathbf M}, {\mathbf A}, n_e\right)$. And we have the following relations of the functional derivatives
$$
\frac{\delta F}{\delta f} = \frac{\delta \tilde{F}}{\delta f} - \frac{\delta \tilde{F}}{\delta C},\ \frac{\delta F}{\delta n_i} = \frac{\delta \tilde{F}}{\delta n_i} - \frac{\delta \tilde{F}}{\delta C},\
\frac{\delta F}{\delta \bar{\mathbf M}} = \frac{\delta \tilde{F}}{\delta \bar{\mathbf M}},\  \frac{\delta F}{\delta \bar{\mathbf A}} = \frac{\delta \tilde{F}}{\delta \bar{\mathbf A}}, \ \frac{\delta F}{\delta n_e} = \frac{\delta \tilde{F}}{\delta C}.
$$
By plugging the above relations into the Poisson bracket~\eqref{eq:fluidxvpne}, we have (tilde is omitted for convenience)
\begin{equation}
\label{eq:DPsim}
\begin{aligned}
& \{F, G\}(n_i, \bar{\mathbf M}, {\mathbf A}, f({\mathbf x}, {\mathbf p}), C) =  \int f \left[ \frac{\delta F}{\delta f}, \frac{\delta G}{\delta f} \right]_{xp}  \mathrm{d}{\mathbf x}\mathrm{d}{\mathbf p} -  \int f \left[ \frac{\delta F}{\delta f}, \frac{\delta G}{\delta C} \right]_{xp}  \mathrm{d}{\mathbf x}\mathrm{d}{\mathbf p} \\
& -  \int f \left[ \frac{\delta F}{\delta C}, \frac{\delta G}{\delta f} \right]_{xp}  \mathrm{d}{\mathbf x}\mathrm{d}{\mathbf p} + \int \bar{\mathbf M} \cdot \left(\frac{\delta G}{\delta \bar{\mathbf M}} \cdot \nabla \frac{\delta F}{\delta \bar{\mathbf M}} - \frac{\delta F}{\delta \bar{\mathbf M}} \cdot \nabla \frac{\delta G}{\delta \bar{\mathbf M}} \right) \mathrm{d}{\mathbf x} \\
& + \int n_i \left(\frac{\delta G}{\delta \bar{\mathbf M}} \cdot \nabla \frac{\delta F}{\delta n_i} - \frac{\delta F}{\delta \bar{\mathbf M}} \cdot \nabla \frac{\delta G}{\delta n_i}  \right) \mathrm{d}{\mathbf x} - \int n_i \left(\frac{\delta G}{\delta \bar{\mathbf M}} \cdot \nabla \frac{\delta F}{\delta C} - \frac{\delta F}{\delta \bar{\mathbf M}} \cdot \nabla \frac{\delta G}{\delta C}  \right) \mathrm{d}{\mathbf x}\\
& - \int \left( \frac{\delta G}{\delta {\mathbf A}} \cdot \nabla \frac{\delta F}{\delta C} - \frac{\delta F}{\delta {\mathbf A}} \cdot \nabla \frac{\delta G}{\delta C} \right)  \mathrm{d}{\mathbf x} - \int \frac{1}{C + n_i + \int f \, \mathrm{d}{\mathbf p}} \nabla \times {\mathbf A} \cdot \left( \frac{\delta F}{\delta {\mathbf A}}  \times \frac{\delta G}{\delta {\mathbf A}}  \right)  \mathrm{d}{\mathbf x}.
\end{aligned}
\end{equation}
The sum of the second, the third, the sixth, and the seventh terms is zero when the the following condition is satisfied by the functionals
\begin{equation}\label{eq:DPcondition}
    -\nabla \cdot \frac{\delta K}{\delta {\mathbf A}} = \nabla \cdot \left( \int f \frac{\partial }{\partial {\mathbf p}} \frac{\delta K}{\delta f} \mathrm{d}{\mathbf p} + n_i \frac{\delta K}{\delta \bar{\mathbf M}} \right), \quad \forall K.
\end{equation} 
We can verify that
this condition~\eqref{eq:DPcondition} is satisfied by the Hamiltonian
\begin{equation*}
\begin{aligned}
& \int \frac{1}{2}n_i |n_i^{-1}\bar{\mathbf M} -  {\mathbf A}|^2 \, \mathrm{d}{\mathbf x} + \int n_i \mathcal{U}(n_i)\,  \mathrm{d}{\mathbf x} + \int \frac{1}{2} |\nabla \times {\mathbf A}|^2\,  \mathrm{d}{\mathbf x} + \frac{1}{2}\int |{\mathbf p} - {\mathbf A}|^2f\,\mathrm{d}{\mathbf x} \mathrm{d}{\mathbf p}\\
& + \int n_e \mathcal{U}_e(n_e, |\nabla \times {\mathbf A}|)\,  \mathrm{d}{\mathbf x}.
\end{aligned}
\end{equation*}
We can define the following set of functionals satisfying the above condition~\eqref{eq:DPcondition}, 
\begin{equation*}
\begin{aligned}
\mathcal{D} = \{ F(f({\mathbf x}, \, {\mathbf p}), \, {\mathbf A}, n_i,\, \bar{\mathbf M}) & = F(\tilde{f}({\mathbf x}, \, {\mathbf v}), \, {\mathbf B}, \ n_i, \ {\mathbf V})| {\mathbf B} = \nabla \times {\mathbf A}, \quad {\mathbf v} = {\mathbf p} - {\mathbf A},\\
& {\mathbf V} = n_i^{-1}\bar{\mathbf M} - {\mathbf A}, \quad f({\mathbf x}, {\mathbf p},t) := \tilde{f}({\mathbf x}, {\mathbf v}) = \tilde{f}({\mathbf x}, {\mathbf p} - {\mathbf A}) \},
\end{aligned}
\end{equation*}
which is a set of functionals determined by all smooth functionals about $\tilde{f}({\mathbf x}, \, {\mathbf v}), \, {\mathbf B}, \ n_i, \ {\mathbf V}$ with $\nabla \cdot {\mathbf B} = 0$.

In the simplified bracket~\eqref{eq:DPsim}, the functional derivative about $C$ does not appear, and when $C=0$ the resulting bracket is a Poisson bracket for the functionals in $\mathcal{D}$, which can be proved similarly as the Theorem~\ref{thm:relation} by going to the magnetic field formulation.

\begin{remark}\label{rm:DPmagnetic}
    For the magnetic field based formulations of kinetic Hall MHD hybrid model~\eqref{eq:hybridequations}, the Poisson brackets given in~\cite{DP} can be simplified via removing the terms involving the functional derivative of $n_e$, and replacing $n_e$ with $n_i + \int f \, \mathrm{d}{\mathbf v}$, which can also be obtained by the transformation using the Casimir function $C = n_e - n_i - \int f \, \mathrm{d}{\mathbf v}$ as remark~\ref{rm:darboux}. The obtained simplified bracket is a Poisson bracket under the condition that $\nabla \cdot {\mathbf B} = 0$.
\end{remark}


\begin{remark}
As Hall MHD is a subset of the hybrid model~\eqref{eq:hybridequations},
removing the terms related with the functional derivative of $n_e$ in the Poisson bracket proposed in~\cite{holm} can also be done. For the simplified bracket, the Jacobi identity  is satisfied by the functionals satisfying the condition 
$$
 -\nabla \cdot \frac{\delta K}{\delta {\mathbf A}} = \nabla \cdot \left( n_i \frac{\delta K}{\delta \bar{\mathbf M}} \right), \quad \forall K.
$$
The simplified magnetic field based Poisson bracket of Hall MHD has been proposed in~\cite{Abdelhamid}.
\end{remark}

\noindent{\bf Equivalence of the equations derived from the canonical brackets}
Here we prove that the equations derived from the two above brackets are equivalent. 
The following is the Hamiltonian proposed in~\cite{DP}.
\begin{equation}
\label{eq:h1}
\begin{aligned}
H_1(f({\mathbf x}, {\mathbf p}), n_i, \bar{\mathbf M}, {\mathbf A}, n_e) & = \int \frac{1}{2}n_i |n_i^{-1}\bar{\mathbf M} - {\mathbf A}|^2\,  \mathrm{d}{\mathbf x} + \int n_i \mathcal{U}(n_i)\,  \mathrm{d}{\mathbf x} + \int \frac{1}{2} |\nabla \times {\mathbf A}|^2\,  \mathrm{d}{\mathbf x} \\
& + \frac{1}{2}\int |{\mathbf p} - {\mathbf A}|^2f\,\mathrm{d}{\mathbf x} \mathrm{d}{\mathbf p} + \int n_e \mathcal{U}_e(n_e, |\nabla \times {\mathbf A}|)\,  \mathrm{d}{\mathbf x}.
\end{aligned}
\end{equation}
In~\eqref{eq:DPsim} when $C = 0$,
the simplified Poisson bracket for the functionals satisfying the condition~\eqref{eq:DPcondition} is 
\begin{equation}
\label{eq:canonical2}
\begin{aligned}
& \{F, G\}(n_i, \bar{\mathbf M}, {\mathbf A}, f({\mathbf x}, {\mathbf p}))  = \int \bar{\mathbf M} \cdot \left(\frac{\delta G}{\delta \bar{\mathbf M}} \cdot \nabla \frac{\delta F}{\delta \bar{\mathbf M}} - \frac{\delta F}{\delta \bar{\mathbf M}} \cdot \nabla \frac{\delta G}{\delta \bar{\mathbf M}} \right) \mathrm{d}{\mathbf x} + \int f \left[ \frac{\delta F}{\delta f}, \frac{\delta G}{\delta f} \right]_{xp}  \mathrm{d}{\mathbf x}\mathrm{d}{\mathbf p}\\
& + \int n_i \left(\frac{\delta G}{\delta \bar{\mathbf M}} \cdot \nabla \frac{\delta F}{\delta n_i} - \frac{\delta F}{\delta \bar{\mathbf M}} \cdot \nabla \frac{\delta G}{\delta n_i}  \right) \mathrm{d}{\mathbf x} - \int \frac{1}{n_e} \nabla \times {\mathbf A} \cdot \left( \frac{\delta F}{\delta {\mathbf A}}  \times \frac{\delta G}{\delta {\mathbf A}}  \right)  \mathrm{d}{\mathbf x}.
\end{aligned}
\end{equation}
Also in the Hamiltonian~\eqref{eq:h1}, we represent $n_e$ by $n_i$ and $f$ via the quasi-neutrality condition, and get 
\begin{equation}
\label{eq:h2}
\begin{aligned}
H_2(f({\mathbf x}, {\mathbf p}), n_i, \bar{\mathbf M}, {\mathbf A}) & = \int \frac{1}{2}n_i |n_i^{-1}\bar{\mathbf M} -  {\mathbf A}|^2\,  \mathrm{d}{\mathbf x} + \int n_i \mathcal{U}(n_i)\,  \mathrm{d}{\mathbf x}  + \int \frac{1}{2} |\nabla \times {\mathbf A}|^2\,  \mathrm{d}{\mathbf x}\\
& + \frac{1}{2}\int |{\mathbf p} -  {\mathbf A}|^2 f\, \mathrm{d}{\mathbf x} \mathrm{d}{\mathbf p} + \int \left( n_i +  \int f\, \mathrm{d}{\mathbf p} \right) \mathcal{U}_e(  n_i + \int f\, \mathrm{d}{\mathbf p}, |\nabla \times {\mathbf A}| ) \, \mathrm{d}{\mathbf x}.
\end{aligned}
\end{equation}

\begin{theorem}
    The equations derived from bracket~\eqref{eq:fluidxvpne} are equivalent to the equations derived from bracket~\eqref{eq:canonical2} with energy $H_1$~\eqref{eq:h1} and $H_2$~\eqref{eq:h2}, respectively.
\end{theorem}
\begin{proof}
    From bracket~\eqref{eq:fluidxvpne} we have
    \begin{equation}
    \label{eq:rman1}
        \begin{aligned}
            &  \frac{\partial n_i}{\partial t} = - \nabla \cdot \left(n_i \frac{\delta H_1}{\delta \bar{\mathbf M}} \right) =  - \nabla \cdot \left( \bar{\mathbf M} - n_i {\mathbf A} \right) = - \nabla \cdot (n_i {\mathbf V}),\\
            &  \frac{\partial \bar{\mathbf M}_i}{\partial t} = \int \bar{\mathbf M} \cdot \left( \frac{\delta H_1}{\delta \bar{\mathbf M}} \cdot \nabla \delta_i \right) \mathrm{d}{\mathbf x} - \int \bar{\mathbf M} \cdot \left(\delta_i \cdot \nabla \frac{\delta H_1}{\delta \bar{\mathbf M}} \right) \mathrm{d}{\mathbf x} - n_i \frac{\partial}{\partial x_i} \frac{\delta H_1}{\delta n_i},\\
            & \frac{\partial {\mathbf A}}{\partial t} =  \nabla \frac{\delta H_1}{\delta n_e} + \frac{1}{n_e} \nabla \times {\mathbf A} \times \frac{\delta H_1}{\delta {\mathbf A}},\\
            & \frac{\partial n_e}{\partial t} =\nabla \cdot \frac{\delta H_1}{\delta {\mathbf A}}, \quad \frac{\partial f}{\partial t} = \int f \left[ \delta, \frac{\delta H_1}{\delta f}\right]_{xp} \mathrm{d}{\mathbf x} \mathrm{d}{\mathbf p},
        \end{aligned}
    \end{equation}
    where $\delta_i = \delta({\mathbf x} - {\mathbf x}'){\mathbf e}_i$, ${\mathbf e}_i$ is the $i$-th unit vector.\\
    From bracket~\eqref{eq:canonical2} we have 
    \begin{equation}
        \label{eq:rma2}
        \begin{aligned}
            & \frac{\partial n_i}{\partial t} = - \nabla \cdot \left(n_i \frac{\delta H_2}{\delta \bar{\mathbf M}} \right) =  - \nabla \cdot \left( \bar{\mathbf M} -  n_i {\mathbf A} \right) = - \nabla \cdot (n_i {\mathbf V}),\\
            &  \frac{\partial \bar{\mathbf M}_i}{\partial t}  = \int \bar{\mathbf M} \cdot \left( \frac{\delta H_2}{\delta \bar{\mathbf M}} \cdot \nabla \delta_i \right) \mathrm{d}{\mathbf x} - \int \bar{\mathbf M} \cdot \left(\delta_i \cdot \nabla \frac{\delta H_2}{\delta \bar{\mathbf M}} \right) \mathrm{d}{\mathbf x} - n_i \frac{\partial}{\partial x_i} \frac{\delta H_2}{\delta n_i},\\
            & \frac{\partial {\mathbf A}}{\partial t} = \frac{1}{n_e} \nabla \times {\mathbf A} \times \frac{\delta H_2}{\delta {\mathbf A}},\quad \frac{\partial f}{\partial t} = \int f \left[ \delta, \frac{\delta H_2}{\delta f}\right]_{xp} \mathrm{d}{\mathbf x} \mathrm{d}{\mathbf p},
        \end{aligned}
    \end{equation}
    where again $\delta_i = \delta({\mathbf x} - {\mathbf x}'){\mathbf e}_i$, ${\mathbf e}_i$ is the $i$-th unit vector.

    As the difference of the equations satisfied by vector potential ${\mathbf A}$ is that if there is a gradient function term $ \nabla \frac{\delta H_1}{\delta n_e} $, the equations satisfied by the magnetic field ${\mathbf B} =  \nabla \times {\mathbf A}$ are the same. 
    Also it is easy to see that the equation satisfied by $n_i$ is the same.

    As for the distribution function, the equation satisfied by $f$ in~\eqref{eq:rman1} is 
    \begin{equation}\label{eq:dis1}
        \frac{\partial f}{\partial t} + \left({\mathbf p} - {\mathbf A}\right)\cdot \frac{\partial f}{\partial {\mathbf x}} - \left(\frac{\partial {\mathbf A}}{\partial {\mathbf x}}\right)^{\top} ({\mathbf A} - {\mathbf p}) \cdot \frac{\partial f}{\partial {\mathbf p}} = 0.
    \end{equation}
    We define a new distribution function depending on velocity ${\mathbf v}$ via the relation ${\mathbf p} = {\mathbf v} - {\mathbf A}$ as 
    \begin{equation}\label{eq:defineF}
    F({\mathbf x}, {\mathbf v}, t) = f({\mathbf x}, {\mathbf p}, t).
    \end{equation}
    Then by plugging the following relations into  equation~\eqref{eq:dis1}
    $$
\frac{\partial f}{\partial t} = \frac{\partial F}{\partial t} - \frac{\partial {\mathbf A}}{\partial t} \cdot \frac{\partial F}{\partial {\mathbf v}}, \quad \frac{\partial f}{\partial {\mathbf x}} = \frac{\partial F}{\partial {\mathbf x}} - \frac{\partial {\mathbf A}}{\partial {\mathbf x}} \frac{\partial F}{\partial {\mathbf v}}, \quad \frac{\partial f}{\partial {\mathbf p}} =  \frac{\partial F}{\partial {\mathbf v}},
    $$
 we have 
 \begin{equation}\label{eq:dis2}
        \frac{\partial F}{\partial t} + {\mathbf v} \cdot \frac{\partial F}{\partial {\mathbf x}} - \frac{\partial {\mathbf A}}{\partial t} \cdot \frac{\partial F}{\partial {\mathbf v}} = 0.
    \end{equation}
    By the similar calculations for the distribution function of $f$ in~\eqref{eq:rma2}, we have 
     \begin{equation}\label{eq:dis3}
        \frac{\partial F}{\partial t} + {\mathbf v} \cdot \frac{\partial F}{\partial {\mathbf x}} - \left(\frac{\partial {\mathbf A}}{\partial t}  + \nabla \frac{\delta H_1}{\delta n_e} \right)\cdot \frac{\partial F}{\partial {\mathbf v}} = 0.
    \end{equation}
    equation~\eqref{eq:dis2} is the same as equation~\eqref{eq:dis3}, due to the difference between the $\frac{\partial {\mathbf A}}{\partial t}$ in~\eqref{eq:dis2} and~\eqref{eq:dis3} is $\nabla \frac{\delta H_1}{\delta n_e}$, which can be seen from the equations satisfied by the vector potential.

    Next we show the second equation in~\eqref{eq:rman1} and~\eqref{eq:rma2} give the same equation satisfied by $n_i{\mathbf V}:=\bar{\mathbf M} - n_i {\mathbf A}$. Note that the unknown $\bar{\mathbf M}$ in the above two equations is different, as the gauge of the vector potential chosen is different, which is involved in the definition of $\bar{\mathbf M}$.

   For the equation~\eqref{eq:rman1} derived from the bracket~\eqref{eq:fluidxvpne}, with the definition of $\bar{\mathbf M}$, we have 
    \begin{equation*}
    \begin{aligned}
        &\frac {\partial (n_i {\mathbf V})}{\partial t}  =  \frac{\partial \bar{\mathbf M}}{\partial t} -  \frac{\partial (n_i {\mathbf A})}{\partial t}, \\
        &  \frac{\partial \bar{\mathbf M}}{\partial t} =  \int (n_i{\mathbf V} + n_i {\mathbf A}) \cdot \left( \frac{\delta H_1}{\delta \bar{\mathbf M}} \cdot \nabla \sum_i \delta_i \right) \mathrm{d}{\mathbf x} - \int \bar{\mathbf M} \cdot \left(\sum_i \delta_i \cdot \nabla \frac{\delta H_1}{\delta \bar{\mathbf M}} \right) \mathrm{d}{\mathbf x} - n_i \nabla \frac{\delta H_1}{\delta n_i},\\
        & = \int (n_i{\mathbf V} + n_i {\mathbf A}) \cdot \left( {\mathbf V} \cdot \nabla \sum_i \delta_i \right) \mathrm{d}{\mathbf x} - \int (n_i{\mathbf V} + n_i {\mathbf A}) \cdot \left(\sum_i \delta_i \cdot \nabla {\mathbf V} \right) \mathrm{d}{\mathbf x} - n_i \nabla \frac{\delta H_1}{\delta n_i},\\
        & = \int n_i{\mathbf V} \cdot \left( {\mathbf V} \cdot \nabla \sum_i \delta_i \right) \mathrm{d}{\mathbf x} - \int n_i{\mathbf V} \cdot \left(\sum_i \delta_i \cdot \nabla {\mathbf V} \right) \mathrm{d}{\mathbf x} - n_i \nabla(-\frac{|{\mathbf V}|^2}{2} + \mathcal{U}(n_i) + n_i \mathcal{U}'(n_i)),\\
        & + \int  n_i {\mathbf A} \cdot \left( {\mathbf V} \cdot \nabla \sum_i \delta_i \right) \mathrm{d}{\mathbf x} - \int n_i {\mathbf A} \cdot \left(\sum_i \delta_i \cdot \nabla {\mathbf V} \right) \mathrm{d}{\mathbf x} +  n_i \nabla ({\mathbf V}\cdot{\mathbf A}),\\
        & = - \nabla \cdot (n_i {\mathbf V}) {\mathbf V} - n_i {\mathbf V} \cdot \nabla {\mathbf V} - n_i \nabla( \mathcal{U}(n_i) + n_i \mathcal{U}'(n_i)) + n_i {\mathbf V} \times \nabla \times {\mathbf A} - \nabla \cdot (n_i {\mathbf V}) {\mathbf A},\\
        &  \frac{\partial (n_i {\mathbf A})}{\partial t}  = 
 \frac{\partial n_i}{\partial t}  {\mathbf A}+ n_i\frac{\partial {\mathbf A}}{\partial t} 
  = -  \nabla \cdot (n_i {\mathbf V})  {\mathbf A}, \\ 
  & + n_i \left(  \nabla \frac{\delta H_1}{\delta n_e} + \frac{1}{n_e} \nabla \times {\mathbf A} \times (\nabla \times \nabla \times {\mathbf A} -  n_i {\mathbf V} - \int ({\mathbf p} -  {\mathbf A}) f\, \mathrm{d}{\mathbf p} ) - \nabla \times (\gamma \nabla \times {\mathbf A})\right).
    \end{aligned}
    \end{equation*}
    Then we obtain the equation satisfied by $n_i {\mathbf V}$ 
\begin{equation}
\label{eq:rv1}
    \begin{aligned}
        \frac{\partial (n_i {\mathbf V})}{\partial t}  &= - \nabla \cdot (n_i {\mathbf V}) {\mathbf V} - n_i {\mathbf V} \cdot \nabla {\mathbf V} - n_i \nabla( \mathcal{U}(n_i) + n_i \mathcal{U}'(n_i)) + n_i {\mathbf V} \times \nabla \times {\mathbf A}\\
        &  -  n_i \left(  \nabla \frac{\delta H_1}{\delta n_e} + \frac{1}{n_e} {\mathbf B} \times ( \nabla \times {\mathbf B} -  n_i {\mathbf V} -  \int({\mathbf p} -  {\mathbf A}) f\, \mathrm{d}{\mathbf p}) - \nabla \times (\gamma {\mathbf B})\right)
    \end{aligned}
\end{equation}
    Next for the equation~\eqref{eq:rma2} derived from the bracket~\eqref{eq:canonical2},
    we show that equations satisfied by $n_i {\mathbf V}$ is the same as~\eqref{eq:rv1}. By the similar calculations, we get 
     \begin{equation*}
    \begin{aligned}
         \frac{\partial (n_i {\mathbf V})}{\partial t}  &= \frac{\partial \bar{\mathbf M}}{\partial t} - \frac{\partial (n_i {\mathbf A})}{\partial t} \\
         \frac{\partial \bar{\mathbf M}}{\partial t} & = - \nabla \cdot (n_i {\mathbf V}) {\mathbf V} - n_i {\mathbf V} \cdot \nabla {\mathbf V}\\
         & - n_i \nabla\left( \mathcal{U}(n_i) + n_i \mathcal{U}'(n_i)   +  \frac{\delta H_1}{\delta n_e} \right) + n_i {\mathbf V} \times \nabla \times {\mathbf A} -  \nabla \cdot (n_i {\mathbf V}) {\mathbf A}\\
         \frac{\partial (n_i {\mathbf A})}{\partial t} &=   -  \nabla \cdot (n_i {\mathbf V})  {\mathbf A} \\
         & + n_i \left( \frac{1}{n_e} \nabla \times {\mathbf A} \times (\nabla \times \nabla \times {\mathbf A} -  n_i {\mathbf V} -  \int ({\mathbf p} -  {\mathbf A}) f\, \mathrm{d}{\mathbf p}) - \nabla \times (\gamma \nabla \times {\mathbf A})\right).
    \end{aligned}
    \end{equation*}
      Then we have 
\begin{equation}\label{eq:rv2}
    \begin{aligned}
        \frac{\partial (n_i {\mathbf V})}{\partial t}  & = - \nabla \cdot (n_i {\mathbf V}) {\mathbf V} - n_i {\mathbf V} \cdot \nabla {\mathbf V} - n_i \nabla( \mathcal{U}(n_i) + n_i \mathcal{U}'(n_i))  + n_i {\mathbf V} \times \nabla \times {\mathbf A}\\
        &  -  n_i \left( \nabla \frac{\delta H_1}{\delta n_e} + \frac{1}{n_e} {\mathbf B} \times ( \nabla \times {\mathbf B} - n_i {\mathbf V} -  \int ({\mathbf p} -  {\mathbf A}) f\, \mathrm{d}{\mathbf p} ) - \nabla \times (\gamma {\mathbf B})\right),
    \end{aligned}
\end{equation}
which is the same as~\eqref{eq:rv1}.
Note that the kinetic currents in~\eqref{eq:rv1} and~\eqref{eq:rv2} are equal, as $$\int ({\mathbf p} - {\mathbf A}) f \, \mathrm{d}{\mathbf p} = \int {\mathbf v} F({\mathbf x}, {\mathbf v}, t)\, \mathrm{d}{\mathbf v},$$ where $F$ is defined in~\eqref{eq:defineF}.
\end{proof}

\subsection{Kinetic-multifluid model}
For the two fluid model, removing the electron charge density via the Gauss' constraint  has been used in~\cite{burby}. As a direct generalization, here we consider the kinetic-multifluid model.
\begin{equation}
\label{eq:multifluidkinetic}
\begin{aligned}
    &\rho_i \frac{\partial {\mathbf u}_i}{\partial t} + \rho_i ({\mathbf u}_i \cdot \nabla){\mathbf u}_i = a_i \rho_i({\mathbf E} + {\mathbf u}_i \times {\mathbf B}) - \nabla p_i,\\
     &\rho_e \frac{\partial {\mathbf u}_e}{\partial t} + \rho_e ({\mathbf u}_e \cdot \nabla){\mathbf u}_e = a_e \rho_i({\mathbf E} + {\mathbf u}_e \times {\mathbf B}) - \nabla p_e,\\
     &\frac{\partial \rho_i}{\partial t} + \nabla \cdot (\rho_i {\mathbf u}_i) = 0,\quad \frac{\partial \rho_e}{\partial t} + \nabla \cdot (\rho_e {\mathbf u}_e) = 0,\\
     &\frac{\partial f}{\partial t} + \frac{\mathbf p}{m_h}\cdot \frac{\partial f}{\partial {\mathbf x}} + q_h \left( {\mathbf E} + \frac{\mathbf p}{m_h} \times {\mathbf B}\right) \cdot \frac{\partial f}{\partial {\mathbf p}} = 0,\\
    & \mu_0 \epsilon_0 \frac{\partial {\mathbf E}}{\partial t} = \nabla \times {\mathbf B} - \mu_0 (a_i \rho_i {\mathbf u}_i + a_e \rho_e {\mathbf u}_e) - \mu_0 a_h \int {\mathbf p} f \,\mathrm{d}{\mathbf p},\\
     &\frac{\partial {\mathbf B}}{\partial t} = - \nabla \times {\mathbf E},\quad \epsilon_0 \nabla \cdot {\mathbf E} = a_i \rho_i + a_e \rho_e + q_h \int f \,\mathrm{d}{\mathbf p}, \quad \nabla \cdot {\mathbf B} = 0,
\end{aligned}
\end{equation}
where we use the notations: the charge-to-mass ratio of the fluid $s$-th species $a_s = q_s / m_s$, mass density $\rho_s$, velocity ${\mathbf u}_s$, the distribution function $f({\mathbf x}, {\mathbf p}, t)$ of the ions depending on time $t$, position ${\mathbf x}$, and momentum ${\mathbf p}$, the kinetic ion mass $m_h$, kinetic ion charge $q_h$, electric field ${\mathbf E}$ , and magnetic field ${\mathbf B}$.

In~\cite{Tronci}, the following Poisson bracket is proposed, 
\begin{equation}
\label{eq:troncimulti}
\begin{aligned}
    \{ F, G \} & = \sum_{s = i, e} \int {\mathbf m}_s \cdot \left[ \frac{\delta F}{\delta {\mathbf m}_s}, \frac{\delta G}{\delta {\mathbf m}_s}\right] \mathrm{d}{\mathbf x} + \int f \left[ \frac{\delta F}{\delta f}, \frac{\delta G}{\delta f}\right]_{xp} \mathrm{d}{\mathbf x} \mathrm{d}{\mathbf p}\\
    & - \sum_{s = i, e} \int \rho_s \left(\frac{\delta F}{\delta {\mathbf m}_s} \cdot \nabla \frac{\delta G}{\delta \rho_s} - \frac{\delta G}{\delta {\mathbf m}_s} \cdot \nabla \frac{\delta F}{\delta \rho_s} \right) \mathrm{d}{\mathbf x}\\
    & + \sum_{s=i,e} \int a_s \rho_s \left( \frac{1}{\epsilon_0}\frac{\delta F}{\delta {\mathbf m}_s} \cdot \frac{\delta G}{\delta {\mathbf E}} -  \frac{1}{\epsilon_0}\frac{\delta G}{\delta {\mathbf m}_s} \cdot \frac{\delta F}{\delta {\mathbf E}} + {\mathbf B} \cdot  \frac{\delta F}{\delta {\mathbf m}_s} \times  \frac{\delta G}{\delta {\mathbf m}_s}   \right)  \mathrm{d}{\mathbf x}\\
    & + q_h \int f \left(\frac{1}{\epsilon_0} \frac{\partial }{\partial {\mathbf p}} \frac{\delta F}{\delta f} \cdot \frac{\delta G}{\delta {\mathbf E}}  -  \frac{1}{\epsilon_0}\frac{\partial }{\partial {\mathbf p}} \frac{\delta G}{\delta f} \cdot \frac{\delta F}{\delta {\mathbf E}}  + {\mathbf B} \cdot \left(\frac{\partial }{\partial {\mathbf p}} \frac{\delta F}{\delta f} \times \frac{\partial }{\partial {\mathbf p}} \frac{\delta G}{\delta f}  \right)  \right)  \mathrm{d}{\mathbf x} \mathrm{d}{\mathbf p}\\
    & + \frac{1}{\epsilon_0} \int \left( \frac{\delta F}{\delta {\mathbf E}} \cdot \nabla \times \frac{\delta G}{\delta {\mathbf B}} -   \frac{\delta G}{\delta {\mathbf E}} \cdot \nabla \times \frac{\delta F}{\delta {\mathbf B}}      \right)  \mathrm{d}{\mathbf x},
\end{aligned}
\end{equation}
where ${\mathbf m}_s = \rho_s {\mathbf u}_s$. Hamiltonian is 
\begin{equation}
\label{eq:multihamil}
    \begin{aligned}
        \frac{1}{2} \sum_{s = i, e} \int \frac{|{\mathbf m}_s|^2}{\rho_s} \mathrm{d}{\mathbf x} + \frac{1}{2m_h} \int f |{\mathbf p}|^2\, \mathrm{d}{\mathbf x} \mathrm{d}{\mathbf p} + \sum_{s=i,e} \int \rho_s \mathcal{U}(\rho_s)\, \mathrm{d}{\mathbf x} + \frac{\epsilon_0}{2} \int |{\mathbf E}|^2 \mathrm{d}{\mathbf x} + \frac{1}{2\mu_0} \int |{\mathbf B}|^2 \mathrm{d}{\mathbf x}.
    \end{aligned}
\end{equation}

By the constraint of the model, i.e.,
$
\epsilon_0 \nabla \cdot {\mathbf E} = a_i \rho_i + a_e \rho_e + q_h \int f\, \mathrm{d}{\mathbf p}
$,
we have 
$$
\rho_e = a_e^{-1} \left( \epsilon_0 \nabla \cdot {\mathbf E}  -  a_i \rho_i  -  q_h \int f \,\mathrm{d}{\mathbf p} \right),
$$
based on which we can define the following the transformation and get a new bracket,
$$
\left(\rho_i, \rho_e, {\mathbf m}_i, {\mathbf m}_e, f, {\mathbf E}, {\mathbf B} \right) \rightarrow \left(\rho_i, C = \rho_e  - a_e^{-1} \left( \epsilon_0 \nabla \cdot {\mathbf E}  -  a_i \rho_i  -  q_h \int f\, \mathrm{d}{\mathbf p} \right), {\mathbf m}_i, {\mathbf m}_e, f, {\mathbf E}, {\mathbf B}  \right).
$$
Any smooth functional $\tilde{F}$ about $\left(\rho_i, C, {\mathbf m}_i, {\mathbf m}_e, f, {\mathbf E}, {\mathbf B}  \right)$ can be considered as a functional $F$ about  $\left(\rho_i, \rho_e, {\mathbf m}_i, {\mathbf m}_e, f, {\mathbf E}, {\mathbf B} \right)$, and we have the following relations of the functional derivatives,
\begin{equation}
\begin{aligned}
&\frac{\delta F}{\delta f} = \frac{\delta \tilde{F}}{\delta f} + a_e^{-1}q_h\frac{\delta \tilde{F}}{\delta C}, \, \frac{\delta F}{\delta \rho_i} = \frac{\delta \tilde{F}}{\delta \rho_i} + a_e^{-1}a_i\frac{\delta \tilde{F}}{\delta C},\, \frac{\delta F}{\delta \rho_e} = \frac{\delta \tilde{F}}{\delta C}, \, \frac{\delta F}{\delta {\mathbf m}_i} = \frac{\delta \tilde{F}}{\delta {\mathbf m}_i},\\
& \frac{\delta F}{\delta {\mathbf m}_e} = \frac{\delta \tilde{F}}{\delta {\mathbf m}_e},\,
\frac{\delta F}{\delta {\mathbf B}} = \frac{\delta \tilde{F}}{\delta {\mathbf B}}, \,\frac{\delta F}{\delta {\mathbf E}} = \frac{\delta \tilde{F}}{\delta {\mathbf E}} + a_e^{-1}\epsilon_0\nabla \frac{\delta \tilde{F}}{\delta C}.
\end{aligned}
\end{equation}
Plugging the above relations into the bracket~\eqref{eq:troncimulti}, we get the simplified Poisson racket (we omit the tilde here for convenience), 
\begin{equation}
\label{eq:bracketkl}
\begin{aligned}
    \{ F, G \}_l & = \sum_{s = i, e} \int {\mathbf m}_s \cdot \left[ \frac{\delta F}{\delta {\mathbf m}_s}, \frac{\delta G}{\delta {\mathbf m}_s}\right] \mathrm{d}{\mathbf x} + \int f \left[ \frac{\delta F}{\delta f}, \frac{\delta G}{\delta f}\right]_{xp} \mathrm{d}{\mathbf x} \mathrm{d}{\mathbf p}\\
    & - \int \rho_i \left(\frac{\delta F}{\delta {\mathbf m}_i} \cdot \nabla \frac{\delta G}{\delta \rho_i} - \frac{\delta G}{\delta {\mathbf m}_i} \cdot \nabla \frac{\delta F}{\delta \rho_i} \right) \mathrm{d}{\mathbf x}\\
    & + \sum_{s=i,e} \int a_s \rho_s \left( \frac{1}{\epsilon_0}\frac{\delta F}{\delta {\mathbf m}_s} \cdot \frac{\delta G}{\delta {\mathbf E}} -  \frac{1}{\epsilon_0}\frac{\delta G}{\delta {\mathbf m}_s} \cdot \frac{\delta F}{\delta {\mathbf E}} + {\mathbf B} \cdot  \frac{\delta F}{\delta {\mathbf m}_s} \times  \frac{\delta G}{\delta {\mathbf m}_s}   \right)  \mathrm{d}{\mathbf x}\\
    & + q_h \int f \left(\frac{1}{\epsilon_0} \frac{\partial }{\partial {\mathbf p}} \frac{\delta F}{\delta f} \cdot \frac{\delta G}{\delta {\mathbf E}}  -  \frac{1}{\epsilon_0}\frac{\partial }{\partial {\mathbf p}} \frac{\delta G}{\delta f} \cdot \frac{\delta F}{\delta {\mathbf E}}  + {\mathbf B} \cdot \left(\frac{\partial }{\partial {\mathbf p}} \frac{\delta F}{\delta f} \times \frac{\partial }{\partial {\mathbf p}} \frac{\delta G}{\delta f}  \right)  \right)  \mathrm{d}{\mathbf x} \mathrm{d}{\mathbf p}\\
    & + \frac{1}{\epsilon_0} \int \left( \frac{\delta F}{\delta {\mathbf E}} \cdot \nabla \times \frac{\delta G}{\delta {\mathbf B}} -   \frac{\delta G}{\delta {\mathbf E}} \cdot \nabla \times \frac{\delta F}{\delta {\mathbf B}}      \right)  \mathrm{d}{\mathbf x}.
\end{aligned}
\end{equation}
We can see in~\eqref{eq:bracketkl} there is no term involving $C$ and the functional derivative about $C$, so $C$ can be regarded as a parameter. 

We plug this relation $
\rho_e = a_e^{-1} \left( \epsilon_0 \nabla \cdot {\mathbf E}  -  a_i \rho_i  -  q_h \int f\, \mathrm{d}{\mathbf p} \right)
$ into the Hamiltonian~\eqref{eq:multihamil} and get a new Hamiltonian which does not depend on $\rho_e$. The hybrid model~\eqref{eq:multifluidkinetic} without the equation about $\rho_e$ and without the constraint $\epsilon_0 \nabla \cdot {\mathbf E} = a_i \rho_i + a_e \rho_e + q_h \int f\, \mathrm{d}{\mathbf p}$ can be derived using the new bracket~\eqref{eq:bracketkl} and new Hamiltonian. Numerically there is no need to consider about preserving the constraint, 
$
\epsilon_0 \nabla \cdot {\mathbf E} = a_i \rho_i + a_e \rho_e + q_h \int f\, \mathrm{d}{\mathbf p},
$
which is built in and holds automatically.

\section{Conclusion}
In this work, we investigate the quasi-neutrality condition of the hybrid model with kinetic ions and massless electrons. This condition corresponds to a momentum map for the canonical momentum based formulation and a Casimir functional for the velocity based formulations. Moreover, we obtain the conditions of the Jacobi identity of the three simplified Poisson brackets given in~\cite{LHPS, LHPS2}. A direct proof of the Jacobi identity of the magnetic field based bracket~\cite{LHPS} is given. 
Generalizations to some plasma models including the kinetic Hall MHD hybrid model are also conducted.

\section{Appendix}\label{sec:appendix}
In this appendix, we show more details of the Theorem~\ref{theorem:bracket}. Specifically, we present the detailed proof of the sum of the permutation of each term in~\eqref{eq:each} is zero.
In the following $'\text{cyc}'$ denotes the terms by permutation of $F, G$, and $K$.
For convenience, we use the following notations,
\begin{equation}
\begin{aligned}
&{\mathbf F} = \int f \frac{\partial}{\partial {\mathbf v}} \frac{\delta F}{\delta f} \mathrm{d}{\mathbf v}, \quad {\mathbf G} = \int f \frac{\partial}{\partial {\mathbf v}} \frac{\delta G}{\delta f} \mathrm{d}{\mathbf v}, \quad {\mathbf K} = \int f \frac{\partial}{\partial {\mathbf v}} \frac{\delta K}{\delta f} \mathrm{d}{\mathbf v},\quad
{\mathbf F}^A = \nabla \times \frac{\delta F}{\delta {\mathbf B}},\\
&{\mathbf G}^A = \nabla \times \frac{\delta G}{\delta {\mathbf B}}, \quad {\mathbf K}^A = \nabla \times \frac{\delta K}{\delta {\mathbf B}},\quad F_f = \frac{\delta F}{\delta f},\quad G_f = \frac{\delta G}{\delta f},\quad K_f = \frac{\delta K}{\delta f}.
\end{aligned}
\end{equation}
The Lemma 2 in~\cite{lifting} and the following three vector identities are used in many places,
\begin{equation}\label{eq:idcurl}
\nabla \times ({\mathbf X} \times {\mathbf Y}) = {\mathbf X}(\nabla \cdot {\mathbf Y}) -{\mathbf Y}(\nabla \cdot {\mathbf X}) + ({\mathbf Y}\cdot \nabla){\mathbf X} - ({\mathbf X}\cdot \nabla){\mathbf Y}, 
\end{equation}
\begin{equation}\label{eq:idwai}
{\mathbf Z} \times ({\mathbf X} \times {\mathbf Y}) =({\mathbf Z} \cdot {\mathbf Y}) {\mathbf X} - ({\mathbf Z} \cdot {\mathbf X}) {\mathbf Y},
\end{equation}
\begin{equation}\label{eq:iddiv}
\nabla \cdot ({\mathbf X} \times {\mathbf Y}) = (\nabla \times {\mathbf X}) \cdot {\mathbf Y} - (\nabla \times {\mathbf Y}) \cdot {\mathbf X},
\end{equation}
 where ${\mathbf X}, {\mathbf Y}, {\mathbf Z}$ are three vector fields defined on $\mathbb{R}^3$.

\noindent{\bf The first term}\\ 
$\{\{F, G\}, K\}_{fff}^{f} + \text{cyc} = \int f \left[ \frac{\delta F}{\delta f}, \frac{\delta G}{\delta f}  \right]_{xv} \mathrm{d}{\mathbf x} \mathrm{d}{\mathbf v} + \text{cyc} = 0$, because $ \int f \left[ \frac{\delta F}{\delta f}, \frac{\delta G}{\delta f}  \right]_{xv} \mathrm{d}{\mathbf x} \mathrm{d}{\mathbf v}$ is a Lie--Poisson bracket~\cite{morrisonvm,MW}. 

\noindent{\bf The second term} 
\begin{small}
\begin{equation}
\label{eq:fff-fB}
    \begin{aligned}
        & \{\{F, G\}, K\}_{fff}^{fB} + \text{cyc} \\
        & = \int f \left[ {\mathbf B} \cdot \left( \frac{\partial}{\partial {\mathbf v}}\frac{\delta F}{\delta f} \times \frac{\partial}{\partial {\mathbf v}}\frac{\delta G}{\delta f}\right), \frac{\delta K}{\delta f}\right]_{xv} \mathrm{d}{\mathbf x} \mathrm{d}{\mathbf v} + \int f {\mathbf B} \cdot \left(\frac{\partial}{\partial {\mathbf v}} \left[ \frac{\delta F}{\delta f}, \frac{\delta G}{\delta f} \right]_{xv} \times \frac{\partial}{\partial {\mathbf v}}\frac{\delta K}{\delta f} \right)\mathrm{d}{\mathbf x} \mathrm{d}{\mathbf v} + \text{cyc} \\
        & = \underbrace{\int f \frac{\partial}{\partial {\mathbf v}}\frac{\delta K}{\delta f} \cdot \nabla {\mathbf B} \cdot \left( \frac{\partial}{\partial {\mathbf v}}\frac{\delta F}{\delta f} \times \frac{\partial}{\partial {\mathbf v}}\frac{\delta G}{\delta f} \right)  \mathrm{d}{\mathbf x} \mathrm{d}{\mathbf v}}_{\text{Term}\ 1} + \underbrace{\int f \frac{\partial}{\partial {\mathbf v}}\frac{\delta K}{\delta f} \cdot \nabla \left( \frac{\partial}{\partial {\mathbf v}}\frac{\delta F}{\delta f} \times \frac{\partial}{\partial {\mathbf v}}\frac{\delta G}{\delta f} \right) \cdot {\mathbf B}\,  \mathrm{d}{\mathbf x} \mathrm{d}{\mathbf v}}_{\text{Term}\,2} \\
        & \underbrace{- \int f \frac{\partial}{\partial {\mathbf x}}\frac{\delta K}{\delta f} \cdot \frac{\partial}{\partial {\mathbf v}} \left( \frac{\partial}{\partial {\mathbf v}}\frac{\delta F}{\delta f} \times \frac{\partial}{\partial {\mathbf v}}\frac{\delta G}{\delta f} \right) \cdot {\mathbf B} \,\mathrm{d}{\mathbf x} \mathrm{d}{\mathbf v}}_{\text{Term}\,3}  + \underbrace{\int f {\mathbf B} \cdot \left(\frac{\partial}{\partial {\mathbf v}} \left[ \frac{\delta F}{\delta f}, \frac{\delta G}{\delta f} \right]_{xv} \times \frac{\partial}{\partial {\mathbf v}}\frac{\delta K}{\delta f} \right)\mathrm{d}{\mathbf x} \mathrm{d}{\mathbf v}}_{\text{Term}\,4} + \text{cyc},
    \end{aligned}
\end{equation}
\end{small}
where we use the notation that $M \cdot \mathcal{D} N \cdot K = M_j \mathcal{D}_j N_i K_i$, $M, N, K$ are three vector fields in $\mathbb{R}^3$, and $\mathcal{D}$ is $\nabla$ or $\frac{\partial}{\partial {\mathbf v}}$. 

For the Term 4 $+ \, \text{cyc}$, we have
\begin{small}
\begin{equation*}
    \begin{aligned}
        & \int f {\mathbf B} \cdot \left(\frac{\partial}{\partial {\mathbf v}} \left[ \frac{\delta F}{\delta f}, \frac{\delta G}{\delta f} \right]_{xv} \times \frac{\partial}{\partial {\mathbf v}}\frac{\delta K}{\delta f} \right)\mathrm{d}{\mathbf x} \mathrm{d}{\mathbf v} + \text{cyc}\\
        & = \underbrace{\int f {\mathbf B} \cdot \left( \left(\frac{\partial}{\partial {\mathbf v}}\frac{\partial}{\partial \underline{\mathbf x}}\frac{\delta F}{\delta f} \cdot \frac{\partial}{\partial \underline{\mathbf v}} \frac{\delta G}{\delta f} \right)\times \frac{\partial}{\partial {\mathbf v}}\frac{\delta K}{\delta f} \right)\mathrm{d}{\mathbf x} \mathrm{d}{\mathbf v}}_{c} + \underbrace{\int f {\mathbf B} \cdot \left( \left( \frac{\partial}{\partial \underline{\mathbf x}}\frac{\delta F}{\delta f} \cdot \frac{\partial}{\partial {\mathbf v}} \frac{\partial}{\partial \underline{\mathbf v}} \frac{\delta G}{\delta f}\right)\times \frac{\partial}{\partial {\mathbf v}}\frac{\delta K}{\delta f} \right)\mathrm{d}{\mathbf x} \mathrm{d}{\mathbf v}}_{a}\\
        & \underbrace{- \int f {\mathbf B} \left(\cdot \left( \frac{\partial}{\partial {\mathbf v}}\frac{\partial}{\partial \underline{\mathbf v}} \frac{\delta F}{\delta f} \cdot \frac{\partial}{\partial \underline{\mathbf x}} \frac{\delta G}{\delta f} \right) \times \frac{\partial }{\partial {\mathbf v}} \frac{\delta K}{\delta f}\right) \mathrm{d}{\mathbf x} \mathrm{d}{\mathbf v}}_{b} \underbrace{- \int f {\mathbf B} \cdot \left(\left( \frac{\partial}{\partial \underline{\mathbf v}} \frac{\delta F}{\delta f} \cdot \frac{\partial}{\partial {\mathbf v}}\frac{\partial}{\partial \underline{\mathbf x}} \frac{\delta G}{\delta f} \right) \times \frac{\partial }{\partial {\mathbf v}} \frac{\delta K}{\delta f}\right) \mathrm{d}{\mathbf x} \mathrm{d}{\mathbf v}}_{d}\\
        & + \text{cyc},
    \end{aligned}
\end{equation*}
\end{small}
where the functions/operators with/without the underline symbol are calculated or contracted together, the same rules are also used in the remaining parts of this work.
For example, in the term $c$ above, the $j$-th component of the $\frac{\partial}{\partial {\mathbf v}}\frac{\partial}{\partial \underline{\mathbf x}}\frac{\delta F}{\delta f} \cdot \frac{\partial}{\partial \underline{\mathbf v}} \frac{\delta G}{\delta f} $ is 
$$
\frac{\partial}{\partial v_j}\frac{\partial}{\partial \underline{\mathbf x}}\frac{\delta F}{\delta f} \cdot \frac{\partial}{\partial \underline{\mathbf v}} \frac{\delta G}{\delta f} = \sum_{i=1}^3\frac{\partial^2}{\partial v_j \partial x_i}\frac{\delta F}{\delta f}\frac{\partial}{\partial v_i}\frac{\delta G}{\delta f} 
$$

For the Term 2 $+ \, \text{cyc}$, by the permutation of  $F, G, K$, we have
\begin{small}
\begin{equation*}
    \begin{aligned}
   & \underbrace{\int f \frac{\partial}{\partial {\mathbf v}}\frac{\delta F}{\delta f} \cdot \left(\frac{\partial}{\partial {\mathbf x}}\frac{\partial}{\partial \underline{\mathbf v}} \frac{\delta G}{\delta f} \times \frac{\partial}{\partial \underline{\mathbf v}}\frac{\delta K}{\delta f}\right) \cdot \underline{\mathbf B}\mathrm{d}{\mathbf x} \mathrm{d}{\mathbf v}}_{d}  +\underbrace{ \int f \frac{\partial}{\partial {\mathbf v}}\frac{\delta G}{\delta f} \cdot \left(\frac{\partial}{\partial \underline{\mathbf v}} \frac{\delta K}{\delta f} \times \frac{\partial}{\partial {\mathbf x}}\frac{\partial}{\partial \underline{\mathbf v}}\frac{\delta F}{\delta f}\right) \cdot \underline{\mathbf B} \mathrm{d}{\mathbf x} \mathrm{d}{\mathbf v}}_{c} + \text{cyc}.
    \end{aligned} 
\end{equation*}
\end{small}
For the Term 3 $+ \text{cyc}$, by the permutation of  $F, G, K$, we have
\begin{small}
\begin{equation*}
    \begin{aligned}
   & \underbrace{-\int f \frac{\partial}{\partial {\mathbf x}}\frac{\delta F}{\delta f} \cdot \left(\frac{\partial}{\partial {\mathbf v}}\frac{\partial}{\partial \underline{\mathbf v}} \frac{\delta G}{\delta f} \times \frac{\partial}{\partial \underline{\mathbf v}}\frac{\delta K}{\delta f}\right) \cdot \underline{\mathbf B}\mathrm{d}{\mathbf x} \mathrm{d}{\mathbf v}}_{a}  \underbrace{- \int f \frac{\partial}{\partial {\mathbf x}}\frac{\delta G}{\delta f} \cdot \left(\frac{\partial}{\partial \underline{\mathbf v}} \frac{\delta K}{\delta f} \times\frac{\partial}{\partial {\mathbf v}} \frac{\partial}{\partial \underline{\mathbf v}}\frac{\delta F}{\delta f}\right) \cdot \underline{\mathbf B} \mathrm{d}{\mathbf x} \mathrm{d}{\mathbf v}}_{b} + \text{cyc}.
    \end{aligned} 
    \end{equation*}
\end{small}
 The terms with the same label cancel with each other. 
 
 By Lemma 2 in~\cite{lifting}, we have the sum of the permutation of the Term 1 as 
\begin{equation*}
\int f (\nabla \cdot {\mathbf B}) \frac{\partial}{\partial {\mathbf v}}\frac{\delta K}{\delta f} \cdot  \left(\frac{\partial}{\partial {\mathbf v}}\frac{\delta F}{\delta f} \times \frac{\partial}{\partial {\mathbf v}}\frac{\delta G}{\delta f} \right) \mathrm{d}{\mathbf x} \mathrm{d}{\mathbf v},
\end{equation*}
which gives that $\{\{F, G\}, K\}_{fff}^{fB} + \text{cyc} = 0$ under the condition that $\nabla \cdot {\mathbf B} = 0$.

\noindent{\bf The third term}  
\begin{small}
\begin{equation*}
    \begin{aligned}
        & \{\{F, G\}, K\}_{fff}^{ffB/n} + \text{cyc} \\
        & = - \int f \left[ \frac{\mathbf B}{n} \cdot \left( \frac{\partial}{\partial {\mathbf v}}\frac{\delta F}{\delta f} \times \int f \frac{\partial}{\partial {\mathbf v}}\frac{\delta G}{\delta f} \mathrm{d}{\mathbf v}\right), \frac{\delta K}{\delta f}\right]_{xv} \mathrm{d}{\mathbf x} \mathrm{d}{\mathbf v} \\
        &- \int f \left[ \frac{\mathbf B}{n} \cdot \left( \int f \frac{\partial}{\partial {\mathbf v}}\frac{\delta F}{\delta f} \mathrm{d}{\mathbf v} \times \frac{\partial}{\partial {\mathbf v}}\frac{\delta G}{\delta f} \right), \frac{\delta K}{\delta f}\right]_{xv} \mathrm{d}{\mathbf x} \mathrm{d}{\mathbf v}\\
        & - \int \frac{f}{n} {\mathbf B} \cdot \left(\frac{\partial}{\partial {\mathbf v}}\frac{\delta K}{\delta f} \times \nabla \times \int f \left( \frac{\partial}{\partial {\mathbf v}}\frac{\delta F}{\delta f} \times  \frac{\partial}{\partial {\mathbf v}} \frac{\delta G}{\delta f} \right) \mathrm{d}{\mathbf v}  \right)\mathrm{d}{\mathbf x} \mathrm{d}{\mathbf v} \\
        & - \int \frac{\mathbf B}{n} \left( \int f \frac{\partial}{\partial {\mathbf v}}\left[ \frac{\delta F}{\delta f}, \frac{\delta G}{\delta f} \right]_{xv}  \mathrm{d}{\mathbf v} \times \int f \frac{\partial}{\partial {\mathbf v}}\frac{\delta K}{\delta f} \mathrm{d}{\mathbf v} \right) \mathrm{d}{\mathbf x}+ \text{cyc} \\
        & = \int \left( \nabla f \cdot \frac{\partial}{\partial {\mathbf v}} \frac{\delta K}{\delta f} + f \nabla \cdot \frac{\partial}{\partial {\mathbf v}} \frac{\delta K}{\delta f} \right) \frac{\mathbf B}{n}\cdot \left( \frac{\partial}{\partial {\mathbf v}} \frac{\delta F}{\delta f} \times \int f \frac{\partial}{\partial {\mathbf v}} \frac{\delta G}{\delta f} \mathrm{d}{\mathbf v}\right) \mathrm{d}{\mathbf v}\mathrm{d}{\mathbf x}\\
        & + \int \frac{f}{n} \nabla \frac{\delta K}{\delta f} \cdot \frac{\partial}{\partial {\mathbf v}} \left( \frac{\partial}{\partial {\mathbf v}} \frac{\delta F}{\delta f} \times \int f \frac{\partial}{\partial {\mathbf v}} \frac{\delta G}{\delta f} \mathrm{d}{\mathbf v}\right) \cdot {\mathbf B} \mathrm{d}{\mathbf v}\mathrm{d}{\mathbf x}\\
        & + \int \left( \nabla f \cdot \frac{\partial}{\partial {\mathbf v}} \frac{\delta K}{\delta f} + f \nabla \cdot \frac{\partial}{\partial {\mathbf v}} \frac{\delta K}{\delta f} \right)  \frac{\mathbf B}{n} \cdot\left(  \int f \frac{\partial}{\partial {\mathbf v}} \frac{\delta F}{\delta f} \mathrm{d}{\mathbf v}\times \frac{\partial}{\partial {\mathbf v}} \frac{\delta G}{\delta f} \right) \mathrm{d}{\mathbf v}\mathrm{d}{\mathbf x}\\
        & + \int \frac{f}{n} \nabla \frac{\delta K}{\delta f} \cdot \frac{\partial}{\partial {\mathbf v}} \left(  \int f \frac{\partial}{\partial {\mathbf v}} \frac{\delta F}{\delta f} \mathrm{d}{\mathbf v}\times \frac{\partial}{\partial {\mathbf v}} \frac{\delta G}{\delta f} \right) \cdot {\mathbf B} \mathrm{d}{\mathbf v}\mathrm{d}{\mathbf x}\\
         & - \int \frac{f}{n} {\mathbf B} \cdot \left(\frac{\partial}{\partial {\mathbf v}}\frac{\delta K}{\delta f} \times  \int \nabla f \times \left( \frac{\partial}{\partial {\mathbf v}}\frac{\delta F}{\delta f} \times  \frac{\partial}{\partial {\mathbf v}} \frac{\delta G}{\delta f} \right) + f \nabla \times \left( \frac{\partial}{\partial {\mathbf v}}\frac{\delta F}{\delta f} \times  \frac{\partial}{\partial {\mathbf v}} \frac{\delta G}{\delta f} \right)
      \mathrm{d}{\mathbf v}  \right)\mathrm{d}{\mathbf x} \mathrm{d}{\mathbf v} \\
        & - \int \frac{\mathbf B}{n} \left( \int f \frac{\partial}{\partial {\mathbf v}}\left[ \frac{\delta F}{\delta f}, \frac{\delta G}{\delta f} \right]_{xv}  \mathrm{d}{\mathbf v} \times \int f \frac{\partial}{\partial {\mathbf v}}\frac{\delta K}{\delta f} \mathrm{d}{\mathbf v} \right) \mathrm{d}{\mathbf x}+ \text{cyc}\\
        & = \int f \int f \frac{\partial}{\partial {\mathbf v}}\frac{\delta K}{\delta f} \mathrm{d}{\mathbf v} \times \frac{\mathbf B}{n} \cdot \text{Term 2} \ \mathrm{d}{\mathbf x}\mathrm{d}{\mathbf v} + \int \int f \frac{\partial}{\partial {\mathbf v}}\frac{\delta K}{\delta f} \mathrm{d}{\mathbf v} \times \frac{\mathbf B}{n} \cdot \text{Term 3} \ \mathrm{d}{\mathbf x}\mathrm{d}{\mathbf v} + \text{cyc},
    \end{aligned}
\end{equation*}
\end{small}
where $\text{Term 3} = \frac{\partial}{\partial {\mathbf v}}\frac{\delta G}{\delta f} \nabla f \cdot \frac{\partial}{\partial {\mathbf v}}\frac{\delta F}{\delta f} - \frac{\partial}{\partial {\mathbf v}}\frac{\delta F}{\delta f} \nabla f \cdot \frac{\partial}{\partial {\mathbf v}}\frac{\delta G}{\delta f} + \nabla f \times \left(\frac{\partial}{\partial {\mathbf v}}\frac{\delta F}{\delta f} \times \frac{\partial}{\partial {\mathbf v}}\frac{\delta G}{\delta f}\right) = 0$ with the help of the formula~\eqref{eq:idwai}.
As for Term 2, we have 
\begin{small}
\begin{equation*}
\begin{aligned}
    \text{Term 2} & = \frac{\partial}{\partial {\mathbf v}}\frac{\delta G}{\delta f} \nabla \cdot \frac{\partial}{\partial {\mathbf v}}\frac{\delta F}{\delta f} + \frac{\partial}{\partial {\mathbf x}}\frac{\delta F}{\delta f}\cdot\frac{\partial}{\partial {\mathbf v}}\left(\frac{\partial}{\partial {\mathbf v}}\frac{\delta G}{\delta f}\right) - \frac{\partial}{\partial {\mathbf v}}\frac{\delta F}{\delta f} \nabla \cdot \frac{\partial}{\partial {\mathbf v}}\frac{\delta G}{\delta f} - \frac{\partial}{\partial {\mathbf x}}\frac{\delta G}{\delta f}\cdot\frac{\partial}{\partial {\mathbf v}}\left(\frac{\partial}{\partial {\mathbf v}}\frac{\delta F}{\delta f}\right)\\
    & + \nabla  \times \left(\frac{\partial}{\partial {\mathbf v}}\frac{\delta F}{\delta f} \times \frac{\partial}{\partial {\mathbf v}}\frac{\delta G}{\delta f}\right) - \frac{\partial}{\partial {\mathbf v}} \left( \frac{\partial}{\partial {\mathbf x}}\frac{\delta F}{\delta f} \cdot \frac{\partial}{\partial {\mathbf v}}\frac{\delta G}{\delta f} - \frac{\partial}{\partial {\mathbf v}}\frac{\delta F}{\delta f} \cdot \frac{\partial}{\partial {\mathbf x}}\frac{\delta G}{\delta f}\right) = 0,
    \end{aligned}
\end{equation*} 
\end{small}
with the help of the formula~\eqref{eq:idcurl}. Then we have shown that $\{\{F, G\}, K\}_{fff}^{ffB/n} + \text{cyc} = 0$.

\noindent{\bf The fourth term} 
\begin{small}
\begin{equation*}
    \begin{aligned}
         & \{\{F, G\}, K\}_{fff}^{fffB/n^2} +\text{cyc}\\
         & = \int f \left[ \frac{\mathbf B}{n^2} \cdot \left( \int f \frac{\partial}{\partial {\mathbf v}} \frac{\delta F}{\delta f} \mathrm{d}{\mathbf v} \times  \int f \frac{\partial}{\partial {\mathbf v}} \frac{\delta G}{\delta f} \mathrm{d}{\mathbf v}
  \right), \frac{\delta K}{\delta f} \right]_{xv} \mathrm{d}{\mathbf x} \mathrm{d}{\mathbf v}\\
  & + \int \frac{f}{n}{\mathbf B} \cdot \left( \frac{\partial}{\partial {\mathbf v}} \frac{\delta K}{\delta f} \times \nabla \times \left(\frac{1}{n}   \int f \frac{\partial}{\partial {\mathbf v}} \frac{\delta F}{\delta f} \mathrm{d}{\mathbf v} 
 \times  \int f \frac{\partial}{\partial {\mathbf v}} \frac{\delta G}{\delta f} \mathrm{d}{\mathbf v} \right)\right)\mathrm{d}{\mathbf x} \mathrm{d}{\mathbf v}\\
 & =  \int \nabla \left( \frac{\mathbf B}{n^2} \cdot \left( \int f \frac{\partial}{\partial {\mathbf v}} \frac{\delta F}{\delta f} \mathrm{d}{\mathbf v} \times  \int f \frac{\partial}{\partial {\mathbf v}} \frac{\delta G}{\delta f} \mathrm{d}{\mathbf v}
  \right)\right) \cdot \left(\int f \frac{\partial}{\partial{\mathbf v}}\frac{\delta K}{\delta f}\mathrm{d}{\mathbf v}\right) \mathrm{d}{\mathbf x}\\ 
  & + \int \frac{1}{n^2}{\mathbf B} \cdot \left( \int f \frac{\partial}{\partial {\mathbf v}} \frac{\delta K}{\delta f}\mathrm{d}{\mathbf v} \times \nabla \times \left(   \int f \frac{\partial}{\partial {\mathbf v}} \frac{\delta F}{\delta f} \mathrm{d}{\mathbf v} 
 \times  \int f \frac{\partial}{\partial {\mathbf v}} \frac{\delta G}{\delta f} \mathrm{d}{\mathbf v} \right)\right)\mathrm{d}{\mathbf x}\\
 & + \frac{1}{2}\int {\mathbf B} \cdot \left( \int f \frac{\partial}{\partial {\mathbf v}} \frac{\delta K}{\delta f} \mathrm{d}{\mathbf v} \times \left(\nabla \left(\frac{1}{n^2}\right) \times \left(   \int f \frac{\partial}{\partial {\mathbf v}} \frac{\delta F}{\delta f} \mathrm{d}{\mathbf v} 
 \times  \int f \frac{\partial}{\partial {\mathbf v}} \frac{\delta G}{\delta f} \mathrm{d}{\mathbf v} \right)\right)\right)\mathrm{d}{\mathbf x} +\text{cyc}\\
 & = \int \frac{1}{n^2} \left(\underbrace{- {\mathbf B} \cdot ({\mathbf F} \times {\mathbf G}) \nabla \cdot {\mathbf K}  +\text{cyc}}_{\text{Term}\,1} + \underbrace{\frac{\mathbf B}{2}\cdot ({\mathbf K} \times \nabla \times ({\mathbf F} \times {\mathbf G}))  +\text{cyc}}_{\text{Term}\,2} \right) \mathrm{d}{\mathbf x},\\
 & + \int \frac{1}{n^2} \underbrace{\frac{1}{2} ({\mathbf F} \times {\mathbf G}) \cdot \nabla \times ({\mathbf B} \times {\mathbf K}) +\text{cyc}}_{\text{Term}\,3} \mathrm{d}{\mathbf x},
    \end{aligned}
\end{equation*}
\end{small}
where we have used the formula~\eqref{eq:iddiv}.

By using the formula~\eqref{eq:idcurl},
the Term 2 can be expanded as 
\begin{equation}
    \begin{aligned}
        & \frac{\mathbf B}{2} \cdot ({\mathbf K} \times \nabla \times ({\mathbf F} \times {\mathbf G})) + \text{cyc} \\
        & = \underbrace{(\nabla \cdot {\mathbf G}) \frac{\mathbf B}{2} \cdot \left({\mathbf K} \times {\mathbf F} \right) - (\nabla \cdot {\mathbf F}) \frac{\mathbf B}{2} \cdot \left({\mathbf K} \times {\mathbf G} \right) + \text{cyc}}_{={\mathbf B} \cdot ({\mathbf F} \times {\mathbf G}) \nabla \cdot {\mathbf K} + \text{cyc}= - \text{Term}\,1}\\
        & + \underbrace{\frac{\mathbf B}{2} \cdot ({\mathbf K} \times ({\mathbf G} \cdot \nabla) {\mathbf F}) - \frac{\mathbf B}{2} \cdot ( {\mathbf K} \times ({\mathbf F} \cdot \nabla) {\mathbf G})}_{\text{Term}\,4} + \text{cyc}
    \end{aligned}
\end{equation}
By the formula~\eqref{eq:idcurl},
the Term 3 can be expanded as, 
\begin{equation}
    \begin{aligned}
        & \underbrace{\frac{1}{2}(\nabla \cdot {\mathbf K}) ({\mathbf F} \times {\mathbf G}) \cdot {\mathbf B}}_{\text{Term}\,6} + \text{cyc} \underbrace{- \frac{1}{2} (\nabla \cdot {\mathbf B}) ({\mathbf F} \times {\mathbf G}) \cdot {\mathbf K}}_{\text{Term}\,7}  + \text{cyc}\\
        & + \underbrace{\frac{1}{2} ({\mathbf F} \times {\mathbf G}) \cdot ({\mathbf K}\cdot \nabla){\mathbf B}}_{\text{Term}\,8}  + \text{cyc} \underbrace{- \frac{1}{2} ({\mathbf F} \times {\mathbf G}) \cdot ({\mathbf B}\cdot \nabla){\mathbf K}}_{\text{Term}\,5 } + \text{cyc}.
    \end{aligned}
\end{equation}
By permutation of the $F, G, K$ in Term 4 and summing up with Term 5, we obtain 
\begin{equation}
\begin{aligned}
& \frac{\mathbf B}{2} \cdot ({\mathbf G} \times ({\mathbf F} \cdot \nabla) {\mathbf K}) - \frac{\mathbf B}{2} \cdot ( {\mathbf F} \times ({\mathbf G} \cdot \nabla) {\mathbf K}) - \frac{1}{2} ({\mathbf F} \times {\mathbf G}) \cdot ({\mathbf B}\cdot \nabla){\mathbf K}\\
& = \frac{1}{2} (\nabla \cdot {\mathbf K}) ({\mathbf F} \cdot ({\mathbf B} \times {\mathbf G})) = - \text{Term} \, 6
\end{aligned}
\end{equation}
where we have used the Lemma 2 in~\cite{lifting}. Then we know that $$\text{Term} \, 4 + \text{cyc} + \text{Term} \, 5 + \text{cyc} + \text{Term} \, 6 + \text{cyc} = 0.$$
By the Lemma 2 in~\cite{lifting}, we have 
$$
\text{Term}\,8 + \text{cyc} = \frac{1}{2} \nabla \cdot {\mathbf B} ({\mathbf K} \cdot ({\mathbf F} \times {\mathbf G})).
$$
As we know that 
$$
\text{Term}\,7 + \text{cyc} = -\frac{3}{2}\nabla \cdot {\mathbf B} ({\mathbf K} \cdot ({\mathbf F} \times {\mathbf G})),
$$
we have 
$$
\text{Term}\,7 + \text{cyc} + \text{Term}\,8+ \text{cyc} = -\nabla \cdot {\mathbf B} ({\mathbf K} \cdot ({\mathbf F} \times {\mathbf G})).
$$
In the end, we know 
$$
\{\{F, G\}, K\}_{fff}^{fffB/n^2} +\text{cyc} = \int \frac{1}{n^2}  \left(-\nabla \cdot {\mathbf B} ({\mathbf K} \cdot ({\mathbf F} \times {\mathbf G}))\right) \mathrm{d}{\mathbf x} = 0,
$$
under the condition that $\nabla \cdot {\mathbf B} = 0$.

\noindent{\bf The fifth term}
\begin{small}
\begin{equation*}
    \begin{aligned}
         & \{\{F, G\}, K\}_{fff}^{fB^2} +\text{cyc}\\
         & = \int f {\mathbf B} \cdot  \left( \frac{\partial}{\partial {\mathbf v}} \left( {\mathbf B} \cdot \left( \frac{\partial}{\partial {\mathbf v}} \frac{\delta F}{\delta f}  \times   \frac{\partial}{\partial {\mathbf v}} \frac{\delta G}{\delta f} \right)
  \right) \times \frac{\partial}{\partial {\mathbf v}} \frac{\delta K}{\delta f} \right) \mathrm{d}{\mathbf x} \mathrm{d}{\mathbf v} +\text{cyc}\\
  & =\int  \underbrace{f \underline{\mathbf B} \cdot  \left(\left( {\mathbf B} \cdot \left( \frac{\partial} {\partial \underline{\mathbf v}} \frac{\partial} {\partial {\mathbf v}} F_f  \times   \frac{\partial G_f}{\partial {\mathbf v}} \right)
  \right) \times \frac{\partial K_f}{\partial \underline{\mathbf v}}  \right) }_{a} + \underbrace{ f \underline{\mathbf B} \cdot  \left(\left( {\mathbf B} \cdot \left( \frac{\partial F_f}{\partial {\mathbf v}}  \times   \frac{\partial }{\partial \underline{\mathbf v}}\frac{\partial }{\partial {\mathbf v}} G_f\right)
  \right) \times \frac{\partial K_f}{\partial \underline{\mathbf v}}  \right)}_{b} \mathrm{d}{\mathbf x} \mathrm{d}{\mathbf v} \\
  & + \int \underbrace{f \underline{\mathbf B} \cdot  \left(\left( {\mathbf B} \cdot \left( \frac{\partial }{\partial \underline{\mathbf v}}\frac{\partial }{\partial {\mathbf v}} G_f \times   \frac{\partial K_f}{\partial {\mathbf v}} \right)
  \right) \times \frac{\partial F_f}{\partial \underline{\mathbf v}} \right)}_{b} + \underbrace{ f \underline{\mathbf B} \cdot  \left(\left( {\mathbf B} \cdot \left( \frac{\partial G_f}{\partial {\mathbf v}}  \times  \frac{\partial }{\partial \underline{\mathbf v}} \frac{\partial }{\partial {\mathbf v}} K_f\right)
  \right) \times \frac{\partial F_f}{\partial \underline{\mathbf v}}  \right)}_{c} \mathrm{d}{\mathbf x} \mathrm{d}{\mathbf v} \\
   & + \int \underbrace{f \underline{\mathbf B} \cdot  \left(\left( {\mathbf B} \cdot \left( \frac{\partial }{\partial \underline{\mathbf v}}\frac{\partial }{\partial {\mathbf v}} K_f \times   \frac{\partial F_f}{\partial {\mathbf v}} \right)
  \right) \times \frac{\partial G_f}{\partial \underline{\mathbf v}} \right)}_{c} + \underbrace{ f \underline{\mathbf B} \cdot  \left(\left( {\mathbf B} \cdot \left( \frac{\partial K_f}{\partial {\mathbf v}}  \times   \frac{\partial }{\partial \underline{\mathbf v}} \frac{\partial }{\partial {\mathbf v}} F_f \right)
  \right) \times \frac{\partial G_f}{\partial \underline{\mathbf v}} \right)}_{a} \mathrm{d}{\mathbf x} \mathrm{d}{\mathbf v}.
    \end{aligned}
\end{equation*}
\end{small}
Here we again explain the notation. The integrad of the first term, i.e., the term labeled with $a$, is equal to $$
  f {\mathbf B} \cdot \left({\mathbf L} \times \frac{\partial G_f}{\partial {\mathbf v}}\right),
$$
where ${\mathbf L} = \left( {\mathbf B} \cdot \left(\frac{\partial}{\partial v_1}\frac{\partial F_f} {\partial {\mathbf v}}  \times \frac{\partial K_f}{\partial {\mathbf v}}\right),  {\mathbf B} \cdot \left(\frac{\partial}{\partial v_2}\frac{\partial F_f} {\partial {\mathbf v}}  \times \frac{\partial K_f}{\partial {\mathbf v}}\right),  {\mathbf B} \cdot \left(\frac{\partial}{\partial v_3}\frac{\partial F_f} {\partial {\mathbf v}}  \times \frac{\partial K_f}{\partial {\mathbf v}}\right) \right)^\top$.
The terms with the same label cancel to each other, so we know $ \{\{F, G\}, K\}_{fff}^{fB^2} +\text{cyc} = 0$.

\noindent{\bf The sixth term}
\begin{small}
\begin{equation*}
    \begin{aligned}
         & \{\{F, G\}, K\}_{fff}^{ffB^2/n} +\text{cyc}\\
         & = - \int f {\mathbf B} \cdot  \left( \frac{\partial}{\partial {\mathbf v}} \left( \frac{\mathbf B}{n} \cdot \left( \frac{\partial}{\partial {\mathbf v}} \frac{\delta F}{\delta f}  \times  \int f \frac{\partial}{\partial {\mathbf v}} \frac{\delta G}{\delta f} \mathrm{d}{\mathbf v}\right)
  \right) \times \frac{\partial}{\partial {\mathbf v}} \frac{\delta K}{\delta f} \right) \mathrm{d}{\mathbf x} \mathrm{d}{\mathbf v}\\
  & - \int f {\mathbf B} \cdot  \left( \frac{\partial}{\partial {\mathbf v}} \left( \frac{\mathbf B}{n} \cdot \left( \int f \frac{\partial}{\partial {\mathbf v}} \frac{\delta F}{\delta f} \mathrm{d}{\mathbf v} \times   \frac{\partial}{\partial {\mathbf v}} \frac{\delta G}{\delta f} \right)
  \right) \times \frac{\partial}{\partial {\mathbf v}} \frac{\delta K}{\delta f} \right) \mathrm{d}{\mathbf x} \mathrm{d}{\mathbf v}\\
  & - \int \frac{\mathbf B}{n} \cdot \left( \int f \frac{\partial}{\partial {\mathbf v}} \left({\mathbf B} \cdot \left( \frac{\partial}{\partial {\mathbf v}} \frac{\delta F}{\delta f} \times  \frac{\partial}{\partial {\mathbf v}} \frac{\delta G}{\delta f} \right) \mathrm{d}{\mathbf v} \right) \times \int f  \frac{\partial}{\partial {\mathbf v}} \frac{\delta K}{\delta f} \mathrm{d}{\mathbf v} \right) \mathrm{d}{\mathbf x} + \text{cyc} \\
  & =- \int f \underline{\mathbf B} \cdot  \left( \left( \frac{\mathbf B}{n} \cdot \left( \frac{\partial}{\partial \underline{\mathbf v}}\frac{\partial F_f}{\partial {\mathbf v}}  \times  {\mathbf G}\right)
  \right) \times \frac{\partial K_f}{\partial \underline{\mathbf v}} \right) \mathrm{d}{\mathbf x} \mathrm{d}{\mathbf v}- \int f \underline{\mathbf B} \cdot  \left( \left( \frac{\mathbf B}{n} \cdot \left({\mathbf F} \times   \frac{\partial}{\partial \underline{\mathbf v}}\frac{\partial G_f}{\partial {\mathbf v}}\right)
  \right) \times \frac{\partial K_f}{\partial \underline{\mathbf v}} \right) \mathrm{d}{\mathbf x} \mathrm{d}{\mathbf v}\\
  & - \int f \underline{\mathbf B} \cdot  \left( \left( \frac{\mathbf B}{n} \cdot \left( \frac{\partial}{\partial \underline{\mathbf v}}  \frac{\partial G_f}{\partial {\mathbf v}}  \times  {\mathbf K}\right)
  \right) \times \frac{\partial F_f}{\partial \underline{\mathbf v}} \right) \mathrm{d}{\mathbf x} \mathrm{d}{\mathbf v} - \int f \underline{\mathbf B} \cdot  \left( \left( \frac{\mathbf B}{n} \cdot \left( {\mathbf G} \times   \frac{\partial}{\partial \underline{\mathbf v}}\frac{\partial K_f}{\partial {\mathbf v}} \right)
  \right) \times \frac{\partial F_f}{\partial \underline{\mathbf v}} \right) \mathrm{d}{\mathbf x} \mathrm{d}{\mathbf v}\\
  & - \int f \underline{\mathbf B} \cdot  \left( \left( \frac{\mathbf B}{n} \cdot \left( \frac{\partial}{\partial \underline{\mathbf v}}\frac{\partial K_f}{\partial {\mathbf v}}  \times  {\mathbf F}\right)
  \right) \times \frac{\partial G_f}{\partial \underline{\mathbf v}} \right) \mathrm{d}{\mathbf x} \mathrm{d}{\mathbf v} - \int f \underline{\mathbf B} \cdot  \left( \left( \frac{\mathbf B}{n} \cdot \left({\mathbf K}\times   \frac{\partial}{\partial \underline{\mathbf v}}\frac{\partial F_f}{\partial {\mathbf v}} \right)
  \right) \times \frac{\partial G_f}{\partial \underline{\mathbf v}}  \right) \mathrm{d}{\mathbf x} \mathrm{d}{\mathbf v}\\
  & - \int f \frac{\underline{\mathbf B}}{n} \cdot \left( \int f \left({\mathbf B} \cdot \left( \frac{\partial}{\partial \underline{\mathbf v}} \frac{\partial F_f}{\partial {\mathbf v}} \times  \frac{\partial G_f}{\partial {\mathbf v}} \right) \mathrm{d}{\mathbf v} \right) \times \underline{\mathbf K} \right) \mathrm{d}{\mathbf x}  - \int \frac{\underline{\mathbf B}}{n} \cdot \left( \int f \left({\mathbf B} \cdot \left( \frac{\partial F_f}{\partial {\mathbf v}}  \times  \frac{\partial }{\partial \underline{\mathbf v}} \frac{\partial G_f}{\partial {\mathbf v}} \right) \mathrm{d}{\mathbf v} \right) \times \underline{\mathbf K}\right) \mathrm{d}{\mathbf x} \\
  &- \int f \frac{\underline{\mathbf B}}{n} \cdot \left( \int f \left({\mathbf B} \cdot \left( \frac{\partial}{\partial \underline{\mathbf v}} \frac{\partial G_f}{\partial {\mathbf v}}  \times  \frac{\partial K_f}{\partial {\mathbf v}} \right) \mathrm{d}{\mathbf v} \right) \times \underline{\mathbf F}\right) \mathrm{d}{\mathbf x} - \int \frac{\underline{\mathbf B}}{n} \cdot \left( \int f \left({\mathbf B} \cdot \left( \frac{\partial G_f}{\partial {\mathbf v}} \times  \frac{\partial }{\partial \underline{\mathbf v}}\frac{\partial K_f}{\partial {\mathbf v}} \right) \mathrm{d}{\mathbf v} \right) \times \underline{\mathbf F} \right) \mathrm{d}{\mathbf x} \\
  & - \int f \frac{\underline{\mathbf B}}{n} \cdot \left( \int f \left({\mathbf B} \cdot \left( \frac{\partial}{\partial \underline{\mathbf v}} \frac{\partial K_f}{\partial {\mathbf v}}  \times  \frac{\partial F_f}{\partial {\mathbf v}}  \right) \mathrm{d}{\mathbf v} \right) \times \underline{\mathbf G} \right) \mathrm{d}{\mathbf x} - \int \frac{\underline{\mathbf B}}{n} \cdot \left( \int f \left({\mathbf B} \cdot \left( \frac{\partial K_f}{\partial {\mathbf v}} \times  \frac{\partial }{\partial \underline{\mathbf v}} \frac{\partial F_f}{\partial {\mathbf v}} \right) \mathrm{d}{\mathbf v} \right) \times \underline{\mathbf G}\right) \mathrm{d}{\mathbf x}.
    \end{aligned}
\end{equation*}
\end{small}
In the above formulas, the first term cancel to the last term, the second term cancel to the ninth term, the third term cancel to eighth term, the fourth term cancel to the eleventh term, the fifth term cancel to the tenth term, the sixth term cancel to the seventh term. Then we know that $ \{\{F, G\}, K\}_{fff}^{ffB^2/n} +\text{cyc} = 0$.

\noindent{\bf The seventh term}
\begin{small}
\begin{equation*}
    \begin{aligned}
         & \{\{F, G\}, K\}_{fff}^{fffB^2/n^2} +\text{cyc}\\
         & =  \int \frac{\mathbf B}{n} \cdot  \left( \int f \frac{\partial}{\partial {\mathbf v}} \left( \frac{\mathbf B}{n} \cdot \left( \frac{\partial}{\partial {\mathbf v}} \frac{\delta F}{\delta f}  \times  \int f \frac{\partial}{\partial {\mathbf v}} \frac{\delta G}{\delta f} \mathrm{d}{\mathbf v}\right)\mathrm{d}{\mathbf v}
  \right) \times \int f \frac{\partial}{\partial {\mathbf v}} \frac{\delta K}{\delta f} \mathrm{d}{\mathbf v}\right) \mathrm{d}{\mathbf x} \\
   & - \int \frac{\mathbf B}{n} \cdot  \left( \int f \frac{\partial}{\partial {\mathbf v}} \left( \frac{\mathbf B}{n} \cdot \left(\int f \frac{\partial}{\partial {\mathbf v}} \frac{\delta F}{\delta f} \mathrm{d}{\mathbf v} \times   \frac{\partial}{\partial {\mathbf v}} \frac{\delta G}{\delta f} \right)\mathrm{d}{\mathbf v}
  \right) \times \int f \frac{\partial}{\partial {\mathbf v}} \frac{\delta K}{\delta f}\mathrm{d}{\mathbf v} \right) \mathrm{d}{\mathbf x}  + \text{cyc}\\
  & =  \int \frac{\underline{\mathbf B}}{n} \cdot  \left( \int f \left( \frac{\mathbf B}{n} \cdot \left( \frac{\partial}{\partial \underline{\mathbf v}} \frac{\partial}{\partial {\mathbf v}}F_f \times {\mathbf G}\right)\mathrm{d}{\mathbf v}
  \right) \times \underline{\mathbf K} \right)  + \frac{\underline{\mathbf B}}{n} \cdot  \left( \int f \left( \frac{\mathbf B}{n} \cdot \left(  \frac{\partial}{\partial \underline{\mathbf v}} \frac{\partial}{\partial {\mathbf v}}G_f \times {\mathbf K}\right)\mathrm{d}{\mathbf v}
  \right) \times \underline{\mathbf F}\right) \mathrm{d}{\mathbf x} \\
  & +  \int \frac{\underline{\mathbf B}}{n} \cdot  \left( \int f \left( \frac{\mathbf B}{n} \cdot \left(  \frac{\partial}{\partial \underline{\mathbf v}} \frac{\partial}{\partial {\mathbf v}}K_f \times {\mathbf F}\right)\mathrm{d}{\mathbf v}
  \right) \times \underline{\mathbf G}\right)  - \frac{\underline{\mathbf B}}{n} \cdot  \left( \int f \left( \frac{\mathbf B}{n} \cdot \left({\mathbf F} \times    \frac{\partial}{\partial \underline{\mathbf v}} \frac{\partial}{\partial {\mathbf v}}G_f\right)\mathrm{d}{\mathbf v}
  \right) \times \underline{\mathbf K}\right) \mathrm{d}{\mathbf x} \\
   & - \int \frac{\underline{\mathbf B}}{n} \cdot  \left( \int f \left( \frac{\mathbf B}{n} \cdot \left({\mathbf G} \times    \frac{\partial}{\partial \underline{\mathbf v}} \frac{\partial}{\partial {\mathbf v}}K_f \right)\mathrm{d}{\mathbf v}
  \right) \times \underline{\mathbf F}\right)  - \frac{\underline{\mathbf B}}{n} \cdot  \left( \int f \left( \frac{\mathbf B}{n} \cdot \left({\mathbf K} \times   \frac{\partial}{\partial \underline{\mathbf v}} \frac{\partial}{\partial {\mathbf v}} F_f  \right)\mathrm{d}{\mathbf v}
  \right) \times \underline{\mathbf G}\right) \mathrm{d}{\mathbf x},
    \end{aligned}
\end{equation*}
\end{small}
where the first term cancel with the last term, the second term cancel with the fourth term, and the third term cancel with the fifth term.

\noindent{\bf The eighth term}
\begin{small}
\begin{equation}
\label{eq:eightfgh}
    \begin{aligned}
         & \{\{F, G\}, K\}_{ffB}^{fB/n} +\text{cyc}\\
         & = - \int f \left[ \frac{\mathbf B}{n} \cdot \left( \frac{\partial}{\partial {\mathbf v}}\frac{\delta G}{\delta f} \times \nabla \times \frac{\delta F}{\delta {\mathbf B}} -  \frac{\partial}{\partial {\mathbf v}}\frac{\delta F}{\delta f} \times \nabla \times \frac{\delta G}{\delta {\mathbf B}} \right), \frac{\delta K}{\delta f}\right]_{xv} \mathrm{d}{\mathbf x} \mathrm{d}{\mathbf v}\\
   & - \int \frac{\mathbf B}{n} \cdot  \left(  \left( \nabla \times \int f\frac{\partial}{\partial {\mathbf v}}\frac{\delta F}{\delta f} \times \frac{\partial}{\partial {\mathbf v}}\frac{\delta G}{\delta f} \mathrm{d}{\mathbf v}\right) \times \nabla \times \frac{\delta K}{\delta {\mathbf B}} 
   \right) \mathrm{d}{\mathbf x}\\
   & + \int \frac{f}{n} {\mathbf B} \cdot \left( \frac{\partial}{\partial {\mathbf v}} \left[\frac{\delta F}{\delta f}, \frac{\delta G}{\delta f} \right]_{xv} \times \nabla \times \frac{\delta K}{\delta {\mathbf B}}\right)\mathrm{d}{\mathbf x} \mathrm{d}{\mathbf v} + \text{cyc}\\
   & = \int \nabla f \cdot \frac{\partial G_f}{\partial {\mathbf v}}\frac{\mathbf B}{n} \cdot \left( \frac{\partial F_f}{\partial {\mathbf v}} \times \frac{\delta K}{\delta {\mathbf B}}\right) \mathrm{d}{\mathbf x}\mathrm{d}{\mathbf v} \\
   & -  \int \nabla f \cdot \frac{\partial F_f}{\partial {\mathbf v}} \frac{\mathbf B}{n} \cdot \left( \frac{\partial G_f}{\partial {\mathbf v}} \times \frac{\delta K}{\delta {\mathbf B}}\right) \mathrm{d}{\mathbf x}\mathrm{d}{\mathbf v} - \int \frac{\mathbf B}{n} \cdot \left( \nabla f \times \left(  \frac{\partial F_f}{\partial {\mathbf v}} \times  \frac{\partial G_f}{\partial {\mathbf v}}\right) \right) \times  \nabla \times \frac{\delta K}{\delta {\mathbf B}}\mathrm{d}{\mathbf x}\mathrm{d}{\mathbf v} \\
   & + \int f \left(\frac{\partial}{\partial {\mathbf v}} \left( \frac{\mathbf B}{n} \cdot \left(\frac{\partial F_f}{\partial {\mathbf v}}  \times \nabla \times \frac{\delta K}{\delta {\mathbf B}} \right)\right) \right) \cdot \frac{\partial G_f}{\partial {\mathbf x}} \mathrm{d}{\mathbf x}\mathrm{d}{\mathbf v}\\
   & - \int f \left(\frac{\partial}{\partial {\mathbf v}} \left( \frac{\mathbf B}{n} \cdot \left(\frac{\partial G_f}{\partial {\mathbf v}}  \times \nabla \times \frac{\delta K}{\delta {\mathbf B}} \right)\right) \right) \cdot \frac{\partial F_f}{\partial {\mathbf x}} \mathrm{d}{\mathbf x}\mathrm{d}{\mathbf v}\\
   & - \int f \frac{\mathbf B}{n} \cdot \left( \nabla \times \left( \frac{\partial F_f}{\partial {\mathbf v}}\times  \frac{\partial G_f}{\partial {\mathbf v}} \right) \times \nabla \times \frac{\delta K}{\delta {\mathbf B}} \right) \mathrm{d}{\mathbf x}\mathrm{d}{\mathbf v} + \int \frac{f}{n} {\mathbf B} \cdot \left( \frac{\partial}{\partial {\mathbf v}} \left[F_f, G_f \right]_{xv} \times \nabla \times \frac{\delta K}{\delta {\mathbf B}}\right)\mathrm{d}{\mathbf x} \mathrm{d}{\mathbf v}\\
   & + \int f \nabla \cdot  \frac{\partial G_f}{\partial {\mathbf v}}  \frac{\mathbf B}{n} \cdot \left( \frac{\partial F_f}{\partial {\mathbf v}}  \times \nabla \times \frac{\delta K}{\delta {\mathbf B}} \right)\mathrm{d}{\mathbf x} \mathrm{d}{\mathbf v} - \int f \nabla \cdot  \frac{\partial F_f}{\partial {\mathbf v}} \frac{\mathbf B}{n} \cdot \left( \frac{\partial G_f}{\partial {\mathbf v}}  \times \nabla \times \frac{\delta K}{\delta {\mathbf B}} \right)\mathrm{d}{\mathbf x} \mathrm{d}{\mathbf v} + \text{cyc},
    \end{aligned}
\end{equation}
\end{small}
where we have used the permutation of $F, G, K$.  The first three terms cancel to each other via using the formula~\eqref{eq:idwai} to the third line. The last six terms cancel to each other by using the formula~\eqref{eq:idcurl} to the last fourth term.

\noindent{\bf The ninth term}
\begin{small}
\begin{equation*}
    \begin{aligned}
         & \{\{F, G\}, K\}_{ffB}^{ffB/n^2} +\text{cyc}\\
         & =  \int f \left[ \int \frac{f}{n^2} {\mathbf B} \cdot \left( \frac{\partial}{\partial {\mathbf v}}\frac{\delta G}{\delta f} \times \nabla \times \frac{\delta F}{\delta {\mathbf B}} -  \frac{\partial}{\partial {\mathbf v}}\frac{\delta F}{\delta f} \times \nabla \times \frac{\delta G}{\delta {\mathbf B}} \right) \mathrm{d}{\mathbf v}, \frac{\delta K}{\delta f}\right]_{xv} \mathrm{d}{\mathbf x} \mathrm{d}{\mathbf v}\\
   & + \int \frac{\mathbf B}{n} \cdot \left( \left( \nabla \times  \left( \frac{1}{n} \int f\frac{\partial}{\partial {\mathbf v}}\frac{\delta F}{\delta f} \mathrm{d}{\mathbf v}\times \int f  \frac{\partial}{\partial {\mathbf v}}\frac{\delta G}{\delta f} \mathrm{d}{\mathbf v}\right)\right) \times \nabla \times \frac{\delta K}{\delta {\mathbf B}} 
   \right) \mathrm{d}{\mathbf x}\\
   & + \int \frac{f}{n} {\mathbf B} \cdot \left( \frac{\partial}{\partial {\mathbf v}} \frac{\delta K}{\delta f} \times \nabla \times \int \frac{f}{n} \left(\frac{\partial}{\partial {\mathbf v}}\frac{\delta G}{\delta f}  \times \nabla \times \frac{\delta F}{\delta {\mathbf B}} -\frac{\partial}{\partial {\mathbf v}}\frac{\delta F}{\delta f}  \times \nabla \times \frac{\delta G}{\delta {\mathbf B}} \right) \mathrm{d}{\mathbf v}\right)\mathrm{d}{\mathbf x} \mathrm{d}{\mathbf v} + \text{cyc}\\
   & = - \int \left(\nabla \cdot \int f \frac{\partial}{\partial {\mathbf v}} \frac{\delta K}{\delta f} \mathrm{d}{\mathbf v} \right) \left(\int \frac{f}{n^2} {\mathbf B} \cdot \left( \frac{\partial}{\partial {\mathbf v}}\frac{\delta G}{\delta f} \times \nabla \times \frac{\delta F}{\delta {\mathbf B}} -  \frac{\partial}{\partial {\mathbf v}}\frac{\delta F}{\delta f} \times \nabla \times \frac{\delta G}{\delta {\mathbf B}} \right) \mathrm{d}{\mathbf v} \right)\mathrm{d}{\mathbf x}\\
   & + \int \frac{\mathbf B}{n^2} \cdot \left( \left( \nabla \times  \left(  \int f\frac{\partial}{\partial {\mathbf v}}\frac{\delta F}{\delta f} \mathrm{d}{\mathbf v} \times \int f  \frac{\partial}{\partial {\mathbf v}}\frac{\delta G}{\delta f} \mathrm{d}{\mathbf v}\right)\right) \times \nabla \times \frac{\delta K}{\delta {\mathbf B}} 
   \right) \mathrm{d}{\mathbf x}\\
   & + \frac{1}{2}\int {\mathbf B} \cdot \left( \left( \nabla \left(\frac{1}{n^2}\right) \times  \left(  \int f\frac{\partial}{\partial {\mathbf v}}\frac{\delta F}{\delta f} \mathrm{d}{\mathbf v}\times \int f  \frac{\partial}{\partial {\mathbf v}}\frac{\delta G}{\delta f} \mathrm{d}{\mathbf v}\right)\right) \times \nabla \times \frac{\delta K}{\delta {\mathbf B}} 
   \right) \mathrm{d}{\mathbf x}\\
   & + \int \frac{f}{n^2} {\mathbf B} \cdot \left( \frac{\partial}{\partial {\mathbf v}} \frac{\delta K}{\delta f} \times \nabla \times \int f \left(\frac{\partial}{\partial {\mathbf v}}\frac{\delta G}{\delta f}  \times \nabla \times \frac{\delta F}{\delta {\mathbf B}} -\frac{\partial}{\partial {\mathbf v}}\frac{\delta F}{\delta f}  \times \nabla \times \frac{\delta G}{\delta {\mathbf B}} \right) \mathrm{d}{\mathbf v}\right)\mathrm{d}{\mathbf x} \mathrm{d}{\mathbf v}\\
   & + \frac{1}{2}\int f {\mathbf B} \cdot \left( \frac{\partial}{\partial {\mathbf v}} \frac{\delta K}{\delta f} \times \left(\nabla \left( \frac{1}{n^2}\right) \times \int f \left(\frac{\partial}{\partial {\mathbf v}}\frac{\delta G}{\delta f}  \times \nabla \times \frac{\delta F}{\delta {\mathbf B}} -\frac{\partial}{\partial {\mathbf v}}\frac{\delta F}{\delta f}  \times \nabla \times \frac{\delta G}{\delta {\mathbf B}} \right) \mathrm{d}{\mathbf v}\right) \right)\mathrm{d}{\mathbf x} \mathrm{d}{\mathbf v} + \text{cyc}\\
   & = - \int \nabla \cdot {\mathbf K} \left( \frac{1}{n^2} {\mathbf B} \cdot \left( {\mathbf G} \times {\mathbf F}^A  - {\mathbf F} \times {\mathbf G}^A \right) \right) \mathrm{d}{\mathbf x}  + \int \frac{1}{n^2} {\mathbf B} \cdot \left( \left(\nabla \times \left( {\mathbf F} \times {\mathbf G}\right) \right) \times {\mathbf K}^A \right)\mathrm{d}{\mathbf x} \\
   & -  \int \frac{1}{2n^2} \nabla \cdot \left( \left( {\mathbf F} \times {\mathbf G}\right) \times \left( {\mathbf K}^A \times {\mathbf B} \right)\right)\mathrm{d}{\mathbf x}  + \int \frac{1}{n^2} {\mathbf B} \cdot \left( {\mathbf K} \times \nabla \times \left( {\mathbf G} \times {\mathbf F}^A - {\mathbf F} \times {\mathbf G}^A \right) \right) \mathrm{d}{\mathbf x}\\
   & - \frac{1}{2} \int \frac{1}{n^2} \nabla \cdot \left( \left( {\mathbf G} \times {\mathbf F}^A - {\mathbf F} \times {\mathbf G}^A \right) \times \left( {\mathbf B} \times {\mathbf K} \right) \right) \mathrm{d}{\mathbf x} + \text{cyc}\\
   & = - \int \frac{1}{n^2} \nabla \cdot {\mathbf G} \left({\mathbf B} \cdot \left( {\mathbf F} \times {\mathbf K}^A \right) \right) - \frac{1}{n^2} \nabla \cdot {\mathbf F} \left( {\mathbf B} \cdot \left( {\mathbf G} \times {\mathbf K}^A \right) \right) \mathrm{d}{\mathbf x}  + \int \frac{1}{n^2} {\mathbf B} \cdot \left( \left(\nabla \times \left( {\mathbf F} \times {\mathbf G}\right) \right) \times {\mathbf K}^A \right)\mathrm{d}{\mathbf x} \\
   & -  \int \frac{1}{2n^2} \nabla \cdot \left( \left( {\mathbf F} \times {\mathbf G}\right) \times \left( {\mathbf K}^A \times {\mathbf B} \right)\right)\mathrm{d}{\mathbf x} + \int \frac{1}{n^2} {\mathbf B} \cdot \left( {\mathbf G} \times \nabla \times \left( {\mathbf F} \times {\mathbf K}^A \right) - {\mathbf F} \times \nabla \times \left( {\mathbf G} \times {\mathbf K}^A\right)\right)\mathrm{d}{\mathbf x} \\
   & -\frac{1}{2}\int \frac{1}{n^2} \nabla \cdot \left(\left( {\mathbf F} \times {\mathbf K}^A\right) \times \left( {\mathbf B} \times {\mathbf G} \right) -  \frac{1}{n^2} \left( {\mathbf G} \times {\mathbf K}^A\right) \times \left( {\mathbf B} \times {\mathbf F} \right) \right)\mathrm{d}{\mathbf x} + \text{cyc}\\
   & = \int \frac{1}{n^2} \nabla \cdot {\mathbf B} \left( {\mathbf K}^A \cdot ({\mathbf F} \times {\mathbf G})\right) \mathrm{d}{\mathbf x} - \int \frac{1}{n^2} \nabla \cdot {\mathbf K}^A \left( {\mathbf B} \cdot ({\mathbf F} \times {\mathbf G})\right) \mathrm{d}{\mathbf x} + \text{cyc},
    \end{aligned}
\end{equation*}
\end{small}
In the above calculations, we used integration by parts, the permutations of $F, G, K$, the formula~\eqref{eq:idcurl}, the formula~\eqref{eq:iddiv}, and the lemma 2 in~\cite{lifting}.
As $\nabla \cdot {\mathbf K}^A = \nabla \cdot \nabla \times \frac{\delta K}{\delta {\mathbf B}}= 0$, we know that $\{\{F, G\}, K\}_{ffB}^{ffB/n^2} +\text{cyc}$ is zero under the condition that $\nabla \cdot {\mathbf B} = 0$.

\noindent{\bf The tenth term}
\begin{small}
\begin{equation*}
    \begin{aligned}
         & \{\{F, G\}, K\}_{ffB}^{fB^2/n} +\text{cyc}\\
         & =  -\int f {\mathbf B} \cdot \frac{\partial}{\partial {\mathbf v}} \left( \frac{\mathbf B}{n} \cdot \left( \frac{\partial}{\partial {\mathbf v}}\frac{\delta G}{\delta f} \times \nabla \times \frac{\delta F}{\delta {\mathbf B}} -  \frac{\partial}{\partial {\mathbf v}}\frac{\delta F}{\delta f} \times \nabla \times \frac{\delta G}{\delta {\mathbf B}} \right)  \right)\times \frac{\partial}{\partial {\mathbf v}}\frac{\delta K}{\delta f}  \mathrm{d}{\mathbf x} \mathrm{d}{\mathbf v}\\
   & + \int \frac{f}{n}{\mathbf B} \cdot \left( \frac{\partial}{\partial {\mathbf v}} \left( {\mathbf B} \cdot  \left(  \frac{\partial}{\partial {\mathbf v}}\frac{\delta F}{\delta f} \times  \frac{\partial}{\partial {\mathbf v}}\frac{\delta G}{\delta f} \right)\right) \right)\times \nabla \times \frac{\delta K}{\delta {\mathbf B}} 
    \mathrm{d}{\mathbf x} \mathrm{d}{\mathbf v} + \text{cyc}\\
    & =  -\int f \underline{\mathbf B} \cdot \left( \frac{\mathbf B}{n} \cdot \left( \frac{\partial}{\partial \underline{\mathbf v}}\frac{\partial}{\partial {\mathbf v}}\frac{\delta F}{\delta f} \times \nabla \times \frac{\delta K}{\delta {\mathbf B}} \right)  \right)\times \frac{\partial}{\partial \underline{\mathbf v}}\frac{\delta G}{\delta f}  \mathrm{d}{\mathbf x} \mathrm{d}{\mathbf v}\\
    & + \int f \underline{\mathbf B} \cdot \left( \frac{\mathbf B}{n} \cdot \left( \frac{\partial}{\partial \underline{\mathbf v}}\frac{\partial}{\partial {\mathbf v}}\frac{\delta G}{\delta f} \times \nabla \times \frac{\delta K}{\delta {\mathbf B}} \right)  \right)\times \frac{\partial}{\partial \underline{\mathbf v}}\frac{\delta F}{\delta f}  \mathrm{d}{\mathbf x} \mathrm{d}{\mathbf v}\\
    & + \int \frac{f}{n}\underline{\mathbf B} \cdot  \left( {\mathbf B} \cdot  \left(  \frac{\partial}{\partial \underline{\mathbf v}}\frac{\partial}{\partial {\mathbf v}}\frac{\delta F}{\delta f} \times  \frac{\partial}{\partial {\mathbf v}}\frac{\delta G}{\delta f}\right) \right)\times \underline{\nabla \times \frac{\delta K}{\delta {\mathbf B}}} 
    \mathrm{d}{\mathbf x} \mathrm{d}{\mathbf v}\\
    & + \int \frac{f}{n}\underline{\mathbf B} \cdot \left( {\mathbf B} \cdot  \left(  \frac{\partial}{\partial {\mathbf v}}\frac{\delta F}{\delta f} \times  \frac{\partial}{\partial \underline{\mathbf v}}\frac{\partial}{\partial {\mathbf v}}\frac{\delta G}{\delta f} \right) \right)\times \underline{\nabla \times \frac{\delta K}{\delta {\mathbf B}}} 
    \mathrm{d}{\mathbf x} \mathrm{d}{\mathbf v} + \text{cyc}.
    \end{aligned}
\end{equation*}
\end{small}
In the above formulas, the first term cancel with the third term, and the second term cancel with the fourth term. So $\{\{F, G\}, K\}_{ffB}^{fB^2/n} +\text{cyc}$ is zero.

\noindent{\bf The eleventh term}
\begin{small}
\begin{equation*}
    \begin{aligned}
         & \{\{F, G\}, K\}_{ffB}^{ffB^2/n^2} +\text{cyc}\\
         & =  -\int \frac{f}{n} {\mathbf B} \cdot \frac{\partial}{\partial {\mathbf v}} \left( \frac{\mathbf B}{n} \cdot \left( \frac{\partial}{\partial {\mathbf v}}\frac{\delta F}{\delta f} \times \int f \frac{\partial}{\partial {\mathbf v}}\frac{\delta G}{\delta f}  \mathrm{d}{\mathbf v} \right)  \right) \times \nabla \times \frac{\delta K}{\delta {\mathbf B}}  \mathrm{d}{\mathbf x} \mathrm{d}{\mathbf v}\\
          &  -\int \frac{f}{n} {\mathbf B} \cdot \frac{\partial}{\partial {\mathbf v}} \left( \frac{\mathbf B}{n} \cdot \left( \int f \frac{\partial}{\partial {\mathbf v}}\frac{\delta F}{\delta f} \mathrm{d}{\mathbf v} \times  \frac{\partial}{\partial {\mathbf v}}\frac{\delta G}{\delta f}  \right)  \right) \times \nabla \times \frac{\delta K}{\delta {\mathbf B}}  \mathrm{d}{\mathbf x} \mathrm{d}{\mathbf v}\\
   & + \int \frac{\mathbf B}{n} \cdot \left( \int f \frac{\partial}{\partial {\mathbf v}} \left( \frac{\mathbf B}{n} \cdot  \left(  \frac{\partial}{\partial {\mathbf v}}\frac{\delta G}{\delta f} \times \nabla \times  \frac{\delta F}{\delta {\mathbf B}} -   \frac{\partial}{\partial {\mathbf v}}\frac{\delta F}{\delta f} \times \nabla \times  \frac{\delta G}{\delta {\mathbf B}} \right) \mathrm{d}{\mathbf v} \right)\times \int f \frac{\partial}{\partial {\mathbf v}} \frac{\delta K}{\delta f} \mathrm{d}{\mathbf v}
    \right) \mathrm{d}{\mathbf x} + \text{cyc}\\
    & =  -\int \frac{f}{n} \underline{\mathbf B} \cdot \left( \frac{\mathbf B}{n} \cdot \left( \frac{\partial}{\partial \underline{\mathbf v}}\frac{\partial}{\partial {\mathbf v}}\frac{\delta F}{\delta f} \times \int f \frac{\partial}{\partial {\mathbf v}}\frac{\delta G}{\delta f}  \mathrm{d}{\mathbf v} \right)  \right) \times \underline{\nabla \times \frac{\delta K}{\delta {\mathbf B}}}  \mathrm{d}{\mathbf x} \mathrm{d}{\mathbf v}\\
    &  -\int \frac{f}{n} \underline{\mathbf B} \cdot  \left( \frac{\mathbf B}{n} \cdot \left( \int f \frac{\partial}{\partial {\mathbf v}}\frac{\delta F}{\delta f} \mathrm{d}{\mathbf v} \times  \frac{\partial}{\partial \underline{\mathbf v}}\frac{\partial}{\partial {\mathbf v}}\frac{\delta G}{\delta f}  \right)  \right) \times \underline{\nabla \times \frac{\delta K}{\delta {\mathbf B}}}  \mathrm{d}{\mathbf x} \mathrm{d}{\mathbf v}\\
    & + \int \frac{\underline{\mathbf B}}{n} \cdot \left( \int f  \left( \frac{\mathbf B}{n} \cdot  \left(  \frac{\partial}{\partial \underline{\mathbf v}}\frac{\partial}{\partial {\mathbf v}}\frac{\delta F}{\delta f} \times \nabla \times  \frac{\delta K}{\delta {\mathbf B}}\right) \mathrm{d}{\mathbf v} \right)\times \underline{\int f \frac{\partial}{\partial {\mathbf v}} \frac{\delta G}{\delta f} \mathrm{d}{\mathbf v}}
    \right) \mathrm{d}{\mathbf x}\\
    & - \int \frac{\underline{\mathbf B}}{n} \cdot \left( \int f \left( \frac{\mathbf B}{n} \cdot  \left(\frac{\partial}{\partial \underline{\mathbf v}}\frac{\partial}{\partial {\mathbf v}}\frac{\delta G}{\delta f} \times \nabla \times  \frac{\delta K}{\delta {\mathbf B}} \right) \mathrm{d}{\mathbf v} \right)\times \underline{\int f \frac{\partial}{\partial {\mathbf v}} \frac{\delta F}{\delta f} \mathrm{d}{\mathbf v}}
    \right) \mathrm{d}{\mathbf x} + \text{cyc}.
    \end{aligned}
\end{equation*}
\end{small}
In the above formulas, the first term cancel with the third term, and the second term cancel with the fourth term, so $ \{\{F, G\}, K\}_{ffB}^{ffB^2/n^2} +\text{cyc}$ is zero.

\noindent{\bf The twelfth term}
\begin{small}
\begin{equation}
\label{eq:12thterm}
    \begin{aligned}
         & \{\{F, G\}, K\}_{fBB}^{fB/n^2} +\text{cyc}\\
         & = \int f \left[ \frac{\mathbf B}{n^2}  \cdot \left(\nabla \times \frac{\delta F}{\delta {\mathbf B}}  \times \nabla \times \frac{\delta G}{\delta {\mathbf B}}  \right), \frac{\delta K}{\delta f}\right]_{xv} \mathrm{d}{\mathbf x} \mathrm{d}{\mathbf v}\\
          & + \int \frac{\mathbf B}{n} \cdot \left( \nabla \times \int \frac{f}{n} \left(\frac{\partial}{\partial {\mathbf v}}\frac{\delta G}{\delta f} \times   \nabla \times \frac{\delta F}{\delta {\mathbf B}}  -   \frac{\partial}{\partial {\mathbf v}}\frac{\delta F}{\delta f} \times   \nabla \times \frac{\delta G}{\delta {\mathbf B}}  \right) \mathrm{d}{\mathbf v}\times \nabla \times \frac{\delta K}{\delta {\mathbf B}} \right) \mathrm{d}{\mathbf x} \\
   & + \int \frac{f}{n}{\mathbf B} \cdot \left(  \frac{\partial}{\partial {\mathbf v}}\frac{\delta K}{\delta f} \times \nabla \times \left( \frac{1}{n} \nabla \times  \frac{\delta F}{\delta {\mathbf B}}  \times \nabla \times  \frac{\delta G}{\delta {\mathbf B}} \right) \right) \mathrm{d}{\mathbf v}
 \mathrm{d}{\mathbf x} + \text{cyc},\\
 & = \int - \frac{\mathbf B}{n^2}  \cdot \left(\frac{\delta F}{\delta {\mathbf A}}  \times \frac{\delta G}{\delta {\mathbf A}}  \right) \nabla \cdot {\mathbf K} \ \mathrm{d}{\mathbf x}\\
 & + \int \frac{\mathbf B}{n} \cdot \left(\left( \nabla \times \frac{ {\mathbf G} \times \frac{\delta F}{\delta {\mathbf A}} - {\mathbf F} \times \frac{\delta G}{\delta {\mathbf A}}  }{n} \right) \times \frac{\delta K}{\delta {\mathbf A}} \right) \mathrm{d}{\mathbf x}\\
 & + \int \frac{\mathbf B}{n} \cdot \left( {\mathbf K} \times \nabla \times \left( \frac{1}{n} \frac{\delta F}{\delta {\mathbf A}} \times \frac{\delta G}{\delta {\mathbf A}} \right) \right) \mathrm{d}{\mathbf x} + \text{cyc},
    \end{aligned}
\end{equation}
\end{small}
The last three terms of~\eqref{eq:12thterm} are in the same form as the formula (45) in~\cite{remarkable} and the result of~\eqref{eq:12thterm} is thus zero, for which the condition $\nabla \cdot {\mathbf B} = 0$ is used, as mentioned in~\cite{remarkable}.

\noindent{\bf The thirteenth term}
\begin{small}
\begin{equation*}
    \begin{aligned}
         & \{\{F, G\}, K\}_{fBB}^{fB^2/n^2} +\text{cyc}\\
         & = - \int \frac{f}{n} {\mathbf B} \cdot \frac{\partial}{\partial {\mathbf v}} \left( \frac{\mathbf B}{n} \cdot  \left( \frac{\partial}{\partial {\mathbf v}} \frac{\delta G}{\delta f} \times \nabla \times \frac{\delta F}{\delta {\mathbf B}} - \frac{\partial}{\partial {\mathbf v}} \frac{\delta F}{\delta f} \times \nabla \times \frac{\delta G}{\delta {\mathbf B}}  \right) \right) \times \nabla \times \frac{\delta K}{\delta {\mathbf B}}\mathrm{d}{\mathbf x} \mathrm{d}{\mathbf v} + \text{cyc}\\
         & = - \int \frac{f}{n} \underline{\mathbf B} \cdot  \left( \frac{\mathbf B}{n} \cdot  \left( \frac{\partial}{\partial \underline{\mathbf v}}\frac{\partial}{\partial {\mathbf v}} \frac{\delta G}{\delta f} \times \nabla \times \frac{\delta F}{\delta {\mathbf B}}   \right) \right) \times \underline{\nabla \times \frac{\delta K}{\delta \mathbf B}}\mathrm{d}{\mathbf x} \mathrm{d}{\mathbf v}\\
         & + \int \frac{f}{n} \underline{\mathbf B} \cdot  \left( \frac{\mathbf B}{n} \cdot  \left( \frac{\partial}{\partial \underline{\mathbf v}} \frac{\partial}{\partial {\mathbf v}} \frac{\delta G}{\delta f} \times \nabla \times \frac{\delta K}{\delta {\mathbf B}}  \right) \right) \times \underline{\nabla \times \frac{\delta F}{\delta {\mathbf B}}}\mathrm{d}{\mathbf x} \mathrm{d}{\mathbf v} + \text{cyc},
    \end{aligned}
\end{equation*}
\end{small}
where the above two terms cancel to each other. So $\{\{F, G\}, K\}_{fBB}^{fB^2/n^2} +\text{cyc} = 0$.

\noindent{\bf The fourteenth term}
\begin{small}
\begin{equation*}
    \begin{aligned}
          \{\{F, G\}, K\}_{BBB}^{B/n^2} +\text{cyc} = \int \frac{\mathbf B}{n} \cdot \left(  \nabla \times \left( \frac{1}{n} \nabla \times \frac{\delta F}{\delta {\mathbf B}} \times \nabla \times \frac{\delta G}{\delta {\mathbf B}} \right) \times \nabla \times \frac{\delta K}{\delta {\mathbf B}}  \right)  \mathrm{d}{\mathbf x} +\text{cyc},
    \end{aligned}
\end{equation*}
\end{small}
which is equal to zero as shown in~\cite{remarkable} under the condition that $\nabla \cdot {\mathbf B} = 0$.

As the above 14 groups of terms are zero under the condition that $\nabla \cdot {\mathbf B} = 0$, we have proven the bracket~\eqref{eq:xvBbracket} satisfies the Jacobi identity under the condition $\nabla \cdot {\mathbf B} = 0$.

\end{document}